\documentclass[preprint]{article}
\usepackage{neurips_2026}

\usepackage{amsmath,amsfonts,bm}

\def\eqref#1{equation~\ref{#1}}

\def\floor#1{\lfloor #1 \rfloor}
\def\1{\bm{1}}

\DeclareMathAlphabet{\mathsfit}{\encodingdefault}{\sfdefault}{m}{sl}
\SetMathAlphabet{\mathsfit}{bold}{\encodingdefault}{\sfdefault}{bx}{n}

\newcommand{\E}{\mathbb{E}}

\newcommand{\Var}{\mathrm{Var}}

\DeclareMathOperator*{\argmin}{arg\,min}

\usepackage{hyperref}
\hypersetup{
  colorlinks=true,
  citecolor=green!80!black,
  linkcolor=green!80!black,
  urlcolor=black
}

\usepackage[most]{tcolorbox}
\usepackage{inconsolata} 
\usepackage{url}
\usepackage{color}
\usepackage{xcolor}
\usepackage{graphicx}
\usepackage{wrapfig}
\usepackage{algorithm}
\usepackage{algorithmic}
\usepackage{amssymb}
\usepackage{comment}
\usepackage{tikz}
\usetikzlibrary{patterns,decorations.pathreplacing}
\usepackage{subcaption} 
\renewcommand{\algorithmicrequire}{\textbf{Input:}}
\renewcommand{\algorithmicensure}{\textbf{Output:}}

\newcommand{\SK}{\mathsf{sk}}
\newcommand{\BER}{\text{BER}}
\newtheorem{theorem}{Theorem}[section]
\newtheorem{proof}{proof}[section]

\newtheorem{lemma}[theorem]{Lemma}

\usepackage{mathtools}
\usepackage{amsmath}
\usepackage{booktabs,makecell,multirow,array,siunitx}
\newcolumntype{L}{@{\hspace{1.5pt}}l@{\hspace{1.5pt}}}
\newcolumntype{C}{@{\hspace{1.5pt}}c@{\hspace{1.5pt}}}
\newcommand{\smallci}[1]{\mbox{\scriptsize$[$#1$]$}}

\newcommand{\meanppl}[1]{\mbox{\footnotesize \num{#1}}}

\title{MirrorMark: Generalizable Mirrored Sampling for Multi-bit LLM Watermarking}

    \author{
  Ya Jiang$^{1,2}$ \quad
  Massieh Kordi Boroujeny$^{2,3}$ \quad
  Surender Suresh Kumar$^{2,3}$ \quad
  Kai Zeng$^{1,2,3}$ \\
  \\
  George Mason University, Fairfax, VA, USA \quad $^1$Department of Computer Science \\
  $^2$Wireless Cyber Center \quad
  $^3$Department of Electrical and Computer Engineering \\
  \texttt{\{yjiang25, mkordibo, skumar43, kzeng2\}@gmu.edu}
}

\begin{document}
\maketitle

\begin{abstract}
As large language models (LLMs) become integral to applications such as question answering and content creation, reliable content attribution has become increasingly important. Watermarking is a promising approach, but most existing methods either provide only binary signals or achieve multi-bit embedding by distorting the generation distribution. We propose MirrorMark, a generalizable mapping-centric approach for multi-bit LLM watermarking. MirrorMark separates the symbol mapping rule from the base watermarking sampler and maps each symbol to a mod-1 mirroring transformation of a detector-reproducible pseudorandom object, such as sampling values or permutation ranks. A binary-tokenizer analysis shows that complementary mappings yield larger matched--mismatched score gaps than independent-key or shift-based mappings. When composed with a distortion-free base sampler, MirrorMark preserves the token probability distribution by design and maintains text quality in practice. To support practical payload embedding, we introduce a Context-Anchored Balanced Scheduler (CABS), which balances token assignments across message positions while localizing edit effects. We further provide theoretical EER analyses for two representative sampler instantiations. Experiments show that MirrorMark achieves strong detectability and bit accuracy while maintaining text quality comparable to non-watermarked generation.
\end{abstract}

\section{Introduction}

The rapid adoption of large language models (LLMs) such as ChatGPT~\citep{chatgpt}, LLaMA~\citep{touvron2023llama}, and Gemini~\citep{team2023gemini} has enabled high-quality text generation for question answering, content creation, and programming assistance~\citep{jo2023promise, austin2021program, perkins2023academic}. At the same time, increasingly human-like synthetic text raises concerns about authenticity, ownership, and responsible use. Reliable content attribution has therefore become important for mitigating misinformation, protecting intellectual property, and supporting accountability in AI deployment~\citep{Reducing_Risks, sok_zhao}.

Watermarking verifies the provenance of LLM-generated content by embedding imperceptible signals during generation that can later be detected. Most existing methods are zero-bit schemes that only answer whether a text is watermarked~\citep{kirchenbauer2023watermark,Aaronson2023,christ2023undetectable,kuditipudi2023robust,Dathathri2024,hu2023unbiased,wu2023dipmark,liu2024an,zhao2025can,he2025theoretically}. They include distortion-based reweighting methods~\citep{kirchenbauer2023watermark,zhao2023provable,liu2024an,zhao2025can}, unbiased or stealthy reweighting methods~\citep{hu2023unbiased,wu2023dipmark}, and distortion-free sampling methods~\citep{Aaronson2023,fugumbelsoft,kuditipudi2023robust,Dathathri2024,he2025theoretically}. A detailed discussion is provided in Appendix~\ref{related-zero-bit}.

While zero-bit watermarking is effective for provenance verification, its binary nature cannot encode metadata such as model identity, generation time, or usage context. This motivates multi-bit watermarking, which embeds payload information for richer attribution and auditing. Existing multi-bit schemes follow different design paths. Distortion-based methods modify the generation distribution to encode information~\citep{wang2023towards,yoo2023advancing,qu2024provably}. StealthInk~\citep{jiang2025stealthink} extends DiPmark, an unbiased reweighting-based zero-bit watermark, to multi-bit watermarking by assigning each message to a contiguous interval in a context-seeded vocabulary permutation and reweighting tokens according to that interval during generation. Distortion-free methods~\citep{zamir2024excuse,boroujeny2024multi} preserve the original distribution by building on the binary-tokenizer setting of~\cite{christ2023undetectable}. These works demonstrate the feasibility of payload embedding, but the role of the symbol mapping rule itself remains underexplored.

A central challenge in multi-bit watermarking is that decoding is no longer a binary hypothesis test. The detector must identify the correct message among many alternatives while controlling false positives. A naive extension assigns an independent key to each message, but incorrect hypotheses then behave like independent noise. Existing structured mappings are often \textbf{location-based}. DISC~\citep{boroujeny2024multi} and ThreeBricks~\citep{fernandez2023three} use cyclic shifts, RSBH~\citep{qu2024provably} uses symbol-dependent shifts based on KGW, and StealthInk uses contiguous intervals in a context-seeded permutation. These shift-based and interval-based mappings place messages at different locations in a shared pseudorandom space, but they do not explicitly create complementary matched--mismatched hypotheses. As a result, strong bias may be needed to achieve high bit accuracy, which can degrade text quality.

In this work, we take a mapping-centric view of multi-bit watermarking. We observe that many in-generation watermarks are driven by a pseudorandom object generated from the secret key and context. This object is reproducible by the detector and follows a known distribution under non-watermarked text, enabling statistical testing. Depending on the base watermark, it may be token-level sampling randomness or a context-seeded vocabulary permutation used for reweighting. We separate the \emph{symbol mapping rule}, which transforms this pseudorandom object according to the embedded symbol, from the \emph{base watermarking sampler}, which uses the transformed object to sample or reweight tokens. This separation highlights the role of symbol mapping in multi-bit reliability, as a good mapping should shape the matched and mismatched score distributions so that the true symbol is easy to distinguish from incorrect alternatives.

We first analyze this effect in a binary-tokenizer setting, where the sampler is fixed and different mappings can be compared directly. The analysis shows that a swapping-style mapping creates stronger matched--mismatched separation than independent-key or shift-based mappings. Motivated by this insight, we propose MirrorMark, a generalizable mapping-centric approach for multi-bit LLM watermarking. MirrorMark maps each symbol to a mod-1 mirroring transformation of the detector-reproducible pseudorandom object. The transformation is measure-preserving, so it preserves the distribution of the pseudorandom object. Thus, MirrorMark preserves the token distribution for distortion-free samplers, and can also be applied to permutation-based reweighting by mirroring normalized token positions while preserving permutation uniformity.

To support practical payload embedding, MirrorMark further uses the Context-Anchored Balanced Scheduler (CABS), which balances token assignments across message positions while localizing the impact of edits. We conduct controlled comparisons under both Gumbel-max sampling and permutation-based reweighting, comparing mirroring with shift- or interval-based mappings such as ThreeBricks and StealthInk. We also instantiate MirrorMark with two representative zero-bit samplers, AA~\citep{Aaronson2023} and SynthID~\citep{Dathathri2024}, and derive theoretical EER estimates for both instantiations. Experiments show that MirrorMark maintains text quality comparable to non-watermarked generation while achieving strong detectability, high bit accuracy, and improved robustness under editing attacks.

\section{A Binary-Tokenizer View of Multi-Bit Symbol Mapping}
\label{sec:binary_mapping_view}

\subsection{Decomposing Multi-Bit Watermarking: Mapping Rule vs. Watermarking Sampler}
\label{sec:decomposition}

We present a mapping-centric view of multi-bit watermarking by separating two design choices that are often coupled: the \emph{symbol mapping rule} and the \emph{base watermarking sampler}. This separation allows us to compare different symbol mappings under the same watermarking mechanism.

Let $p_{\mathrm{LM}}(\cdot \mid x_{<t})$ denote the next-token distribution at generation step $t$, and let $\mathcal{V}=\{x_1,\ldots,x_V\}$. We denote $\mathbf{p}_t(i):=p_{\mathrm{LM}}(x_i\mid x_{<t})$. A watermarking sampler is driven by a detector-accessible random object $\mathbf{z}_t$ generated from the secret key and the context. Depending on the underlying watermarking mechanism, $\mathbf{z}_t$ can take different forms. For randomness-based samplers such as Gumbel-max~\citep{Aaronson2023} and tournament sampling~\citep{Dathathri2024}, $\mathbf{z}_t$ consists of pseudorandom values used during token selection. For permutation-based reweighting schemes, such as DiPmark~\citep{wu2023dipmark}, $\mathbf{z}_t$ can be a context-seeded vocabulary permutation. 

We define a base watermarking sampler as
\begin{equation}
\begin{aligned}
x_t \sim \mathsf{Samp}_{\mathcal{A}}(\mathbf{p}_t,\mathbf{z}_t),
\end{aligned}
\label{eq:abstract_sampler}
\end{equation}
where $\mathcal{A}$ specifies how $\mathbf{z}_t$ is used to generate the next token. For randomness-based samplers, $\mathbf{z}_t$ directly controls token selection. For permutation-based reweighting schemes, $\mathbf{z}_t$ first reshapes or reweights the next-token distribution, and the token is then sampled from the resulting distribution.

A multi-bit watermark embeds a symbol $M\in\mathcal{M}=\{0,1,\ldots,2^m-1\}$, where $m$ is the number of bits carried by the symbol. The message specifies how the watermark random object is transformed before it is used by the base sampler. We define a symbol mapping rule as
\begin{equation}
\begin{aligned}
\mathcal{R}=\{\mathcal{R}_M:\mathcal{Z}\rightarrow\mathcal{Z}\}_{M\in\mathcal{M}},
\qquad
\mathbf{z}_{t,M}=\mathcal{R}_M(\mathbf{z}_t),
\end{aligned}
\label{eq:mapping_family}
\end{equation}
where $\mathcal{Z}$ is the space of detector-accessible watermark randomness. To embed $M^\star$, the encoder samples
\begin{equation}
\begin{aligned}
x_t \sim \mathsf{Samp}_{\mathcal{A}}(\mathbf{p}_t,\mathbf{z}_{t,M^\star}).
\end{aligned}
\label{eq:watermarked_sampler}
\end{equation}

Given a generated token $x_t$, the detector reconstructs $\mathbf{z}_t$ from a pseudorandom function (PRF) by seeding the secret key $\mathrm{sk}$ and the observed context. For each candidate message $M$, it applies the corresponding mapping $\mathcal{R}_M$ and computes a sampler-specific score
\begin{equation}
\begin{aligned}
S_M(x_t,\mathrm{sk})
=
s_{\mathcal{A}}\big(x_t,\mathbf{p}_t,\mathcal{R}_M(\mathbf{z}_t)\big),
\end{aligned}
\label{eq:token_score_general}
\end{equation}
where $s_{\mathcal{A}}(\cdot)$ is the score function induced by the base watermarking sampler and detector. For a sequence of $T$ eligible tokens, the sequence-level score and decoded message are
\begin{equation}
\begin{aligned}
C_M(x_{1:T},\mathrm{sk})
=
\frac{1}{T}\sum_{t=1}^{T}S_M(x_t,\mathrm{sk}),\qquad
\widehat{M}
=
\arg\max_{M\in\mathcal{M}}C_M(x_{1:T},\mathrm{sk}).
\end{aligned}
\label{eq:sequence_score_and_decoder}
\end{equation}

This formulation covers both sampling-randomness and permutation-rank watermarks. In Gumbel-max or tournament sampling, $\mathbf{z}_t$ consists of token-level or layer-wise PRF values. In permutation-based reweighting, $\mathbf{z}_t$ can be the vocabulary permutation $\theta_t$, or normalized token positions in that permutation. A bijective mapping on this space preserves the reference distribution of $\mathbf{z}_t$, so the resulting multi-bit construction can inherit the base sampler's distributional property. The mapping rule then determines the matched--mismatched score separation. For an embedded message $M^\star$, we define
\begin{equation}
\begin{aligned}
\Delta(M^\star)
=
\mathbb{E}\left[
C_{M^\star}(x_{1:T},\mathrm{sk})
-
\max_{M\neq M^\star}C_M(x_{1:T},\mathrm{sk})
\;\middle|\;
M^\star
\right].
\end{aligned}
\label{eq:matched_mismatched_gap}
\end{equation}
Different symbol mapping rules can induce different matched and mismatched score distributions under the same base watermarking sampler. We next use a binary-tokenizer setting, where $\mathbf{z}_t$ is instantiated as a scalar uniform random variable, to isolate this effect and motivate the mirroring mapping construction.
\subsection{Binary-Tokenizer Analysis of Symbol Mapping Rules}
\label{sec:binary_tokenizer_mapping}

We adopt a binary-tokenizer setting, following the zero-bit watermark of~\cite{christ2023undetectable}, to isolate the role of symbol mapping rules. This setting provides a minimal model where the sampler is fixed, allowing different symbol mappings to be compared under the same token-generation rule.

At generation step $t$, suppose the next token belongs to a binary vocabulary $\mathcal{V}_{\mathrm{b}}=\{0,1\}$. Let $\mathbf{p}_t^b:=p^b_{\mathrm{LM}}(1\mid x_{<t})$ and $p^b_{\mathrm{LM}}(0\mid x_{<t})=1-\mathbf{p}_t^b$. The binary sampler draws $\mathbf{u}_t^b\sim\mathrm{Uniform}(0,1)$ and outputs
\begin{equation}
\begin{aligned}
x_t^b
=
\begin{cases}
1, & 0\leq \mathbf{u}_t^b < \mathbf{p}_t^b,\\
0, & \mathbf{p}_t^b\leq \mathbf{u}_t^b < 1.
\end{cases}
\end{aligned}
\label{eq:binary_sampler}
\end{equation}
Since the interval mapped to token `$1$' has length $\mathbf{p}_t^b$, this sampler preserves the original binary-token distribution. We use the same score function as in~\cite{christ2023undetectable}:
\begin{equation}
\begin{aligned}
s(x_t^b,\mathbf{u}_t^b)
=
\begin{cases}
\ln \frac{1}{\mathbf{u}_t^b}, & x_t^b=1,\\
\ln \frac{1}{1-\mathbf{u}_t^b}, & x_t^b=0.
\end{cases}
\end{aligned}
\label{eq:binary_score}
\end{equation}

We compare three representative mappings for embedding one binary symbol $M\in\{0,1\}$, as illustrated in Fig.~\ref{fig:binary_mapping_rules}. 
\textbf{(1) Naive Multi-Key Christ Mapping:} this baseline directly extends the zero-bit Christ-style rule by assigning an independent key to each message to seed PRF. To embed $M^\star$, the encoder generates the pseudorandom value using the key associated with $M^\star$ and applies equation~\ref{eq:binary_sampler}; the decoder tests all symbol hypotheses by regenerating the pseudorandom value under each message key. Hence, for any mismatched message $M\neq M^\star$, the reconstructed randomness is independent of the generated token.
\textbf{(2) DISC Shift Mapping:} DISC uses a shared random value and shifts the token `$1$' interval according to the message. For binary messages, $\delta_0=0$ and $\delta_1=\frac{1}{2}$, so message $M$ maps token `$1$' to the cyclic interval $[\delta_M,\delta_M+\mathbf{p}_t^b)\bmod 1$.
\textbf{(3) Swapping Mapping:} we introduce a swapping-based rule as a bridge to the mod-1 mirroring rule proposed in Section~\ref{sec:mod1}. For $M=0$, token `$1$' is assigned to the left interval $[0,\mathbf{p}_t^b)$; for $M=1$, token `$1$' is assigned to the right interval $(1-\mathbf{p}_t^b,1]$. Thus, the two symbol hypotheses are complementary, i.e., the same token that gives high score to the matched message gives low score to the mismatched message.

\begin{figure}[htbp]
\centering

\begin{minipage}[t]{0.32\textwidth}
\centering
\resizebox{\linewidth}{!}{%
\begin{tikzpicture}
    \draw[thick] (0,0) -- (5,0);
    \filldraw[black] (0,0) circle (2pt) node[anchor=north]{$0$};
    \filldraw[black] (5,0) circle (2pt) node[anchor=north]{$1$};
    \filldraw[black] (3,0) circle (2pt) node[anchor=north]{$\mathbf{p}_t^b$};
    \node [below] at (1.5,-0.5) {$0 \leq \mathbf{u}_t^b < \mathbf{p}_t^b$};
    \node [below] at (4,-0.5) {$\mathbf{p}_t^b \leq \mathbf{u}_t^b < 1$};
    \draw [->, thick, red] (1.5, 0.3) [out=30,in=150] to (7, 0);
    \draw [->, thick, blue] (4, 0.3) [out=30,in=150] to (6, 0);
    \filldraw[black] (6.1,0) circle (2pt) node[anchor=north]{$0$};
    \filldraw[black] (7.1,0) circle (2pt) node[anchor=north]{$1$};
    \draw (6.5, 0) circle (1cm);
    \node [above] at (6.5, 1) {$\mathcal{V}^b$};
    \draw [decorate,decoration = {brace}, thick, red] (0.2,0.1) --  (2.8,0.1);
    \draw [decorate,decoration = {brace}, thick, blue] (3.2,0.1) --  (4.8,0.1);
\end{tikzpicture}}%
\caption*{\footnotesize (a) Naive Multi-Key Christ mapping for embedding bit 0.}
\end{minipage}\hfill
%
\begin{minipage}[t]{0.32\textwidth}
\centering
\resizebox{\linewidth}{!}{%
\begin{tikzpicture}
    \draw[thick] (0,0) -- (5,0);
    \filldraw[black] (0,0) circle (2pt) node[anchor=north]{0};
    \filldraw[black] (5,0) circle (2pt) node[anchor=north]{1};
    \filldraw[black] (1,0) circle (2pt) node[anchor=north]{$\delta_M$};
    \filldraw[black] (4,0) circle (2pt);
    \node [below] at (3.8,0) {\small{$\mathbf{p}_t^b + \delta_M$}};
    \draw [->, thick, blue] (0.5, 0.3) [out=30,in=150] to (6, 0);
    \draw [->, thick, blue] (4.5, 0.3) [out=30,in=150] to (6, 0);
    \draw [->, thick, red] (2.5, 0.3) [out=30,in=150] to (7, 0);
    \filldraw[black] (6.1,0) circle (2pt) node[anchor=north]{0};
    \filldraw[black] (7.1,0) circle (2pt) node[anchor=north]{1};
    \draw (6.5, 0) circle (1cm);
    \node [above] at (6.5, 1) {$\mathcal{V}^b$};
    \draw [decorate,decoration = {brace}, thick, blue] (0.2,0.1) --  (0.8,0.1);
    \draw [decorate,decoration = {brace}, thick, blue] (4.2,0.1) --  (4.8,0.1);
    \draw [decorate,decoration = {brace}, thick, red] (1.2,0.1) --  (3.8,0.1);
\end{tikzpicture}}%
\caption*{\footnotesize (b) DISC shift mapping for embedding bit 1.}
\end{minipage}\hfill
%
\begin{minipage}[t]{0.32\textwidth}
\centering
\resizebox{\linewidth}{!}{%
\begin{tikzpicture}
    \draw[thick] (0,0) -- (5,0);
    \filldraw[black] (0,0) circle (2pt) node[anchor=north]{$0$};
    \filldraw[black] (5,0) circle (2pt) node[anchor=north]{$1$};
    \filldraw[black] (2,0) circle (2pt) node[anchor=north]{$1-\mathbf{p}_t^b$};
    \draw [->, thick, blue] (1.5, 0.3) [out=30,in=150] to (6, 0);
    \draw [->, thick, red] (4, 0.3) [out=30,in=150] to (7, 0);
    \filldraw[black] (6.1,0) circle (2pt) node[anchor=north]{$0$};
    \filldraw[black] (7.1,0) circle (2pt) node[anchor=north]{$1$};
    \draw (6.5, 0) circle (1cm);
    \node [above] at (6.5, 1) {$\mathcal{V}^b$};
    \draw [decorate,decoration = {brace}, thick, blue] (0.2,0.1) --  (1.8,0.1);
    \draw [decorate,decoration = {brace}, thick, red] (2.2,0.1) --  (4.8,0.1);
\end{tikzpicture}}%
\caption*{\footnotesize (c) Swapping mapping for embedding bit 1.}
\end{minipage}

\caption{Comparison of three binary-token symbol mappings. The horizontal axis denotes the uniform randomness variable $\mathbf{u}_t^b\in[0,1)$. Each mapping partitions $[0,1)$ into blue and red regions that are mapped to token `$0$' or token `$1$' in the binary vocabulary $\mathcal{V}^b$. All three mappings preserve the interval lengths and hence preserve the binary-token distribution, but they induce different relationships between matched and mismatched symbol hypotheses.}
\label{fig:binary_mapping_rules}
\vspace{-1em}
\end{figure}

We use the per-token gap $\Delta_t^b(M^\star)$ to measure the expected score advantage of the matched message over the mismatched message at step $t$. Since all three mappings preserve the same token-generation rule under the matched message, they have the same expected matched score $\mathbb{E}\!\left[s(x_t^b,u_{t,M^\star}^b)\mid M^\star\right]
=
1+\ln(2)H_{\mathrm{b}}(p_t^b)$, where $H_{\mathrm{b}}(p)=-p\log p-(1-p)\log(1-p)$. However, their mismatched scores differ as shown in Table~\ref{tab:binary_mapping_gap}. Derivations are provided in Appendix~\ref{app:binary_mapping_gap}. Therefore, naive Multi-Key Christ leaves mismatched hypotheses as independent noise; DISC shares randomness but retains shifted and non-negligible mismatched scores; swapping turns matched evidence into mismatched anti-evidence and yields the largest gap. This binary analysis motivates the continuous mod-1 mirroring mapping rule introduced in the next section. 

\begin{table}[htbp]
\centering
\vspace{-0.5em}
\caption{Expected per-token matched--mismatched score gaps. $H_{\mathrm b}$ uses base-2 logarithm and 
$\widetilde{H}_{\mathrm{DISC}}(p)=\ln(2)(H_{\mathrm b}(|p-\frac12|)-1)$.}
\label{tab:binary_mapping_gap}
\footnotesize
\setlength{\tabcolsep}{3pt}
\begin{tabular}{lcc}
\toprule
Mapping & Mismatched score & $\Delta^b_t(M^\star)$ \\
\midrule
Multi-Key Christ 
& $1$ 
& $\ln(2)H_{\mathrm b}(p_t^b)$ \\
DISC Shift 
& $1+\widetilde{H}_{\mathrm{DISC}}(p_t^b)$ 
& $\ln(2)H_{\mathrm b}(p_t^b)-\widetilde{H}_{\mathrm{DISC}}(p_t^b)$ \\
Swapping 
& $1-\ln(2)H_{\mathrm b}(p_t^b)$ 
& $2\ln(2)H_{\mathrm b}(p_t^b)$ \\
\bottomrule
\end{tabular}
\vspace{-0.4em}
\vspace{-1em}
\end{table}

\section{MirrorMark}
\label{sec:mirrormark}

Section~\ref{sec:binary_tokenizer_mapping} shows that stronger matched--mismatched separation improves multi-bit decoding. MirrorMark extends this insight to real-tokenizer watermarking through three components: mod-1 mirroring creates separated symbol hypotheses, CABS assigns tokens to message positions, and the decoder replays CABS to recover the payload and aggregate detection evidence. Since mod-1 mirroring is measure-preserving, MirrorMark remains distortion-free when composed with a distortion-free base sampler.

\vspace{-0.5em}
\subsection{Mod-1 Mirroring}
\label{sec:mod1}

The binary-tokenizer analysis in Section~\ref{sec:binary_tokenizer_mapping} serves only as a motivating example, which shows that complementary mappings can create a larger matched--mismatched score gap. MirrorMark applies this principle to real-tokenizer watermarking by transforming the sampler randomness.

For a one-bit message $M\in\{0,1\}$, inspired by the swapping rule, we use the complementary mirroring rule as follows, where $\Psi(u;M)$ is the transformed random value of $u$ with respect to $M$,
\begin{equation}
\begin{aligned}
\Psi(u;M)
=
\begin{cases}
1-u, & M=0,\\
u, & M=1.
\end{cases}
\end{aligned}
\label{eq:mod1_m1}
\end{equation}
This mapping preserves uniformity and satisfies $\Psi(u;0)+\Psi(u;1)=1$.

Thus, a random value that gives strong evidence for one message gives weak evidence for the other.

For an $m$-bit ($m>1$) symbol $M\in\mathcal{M}=\{0,1,\ldots,2^m-1\}$, there are more than two symbol hypotheses. In this case, exact complementarity with every incorrect message is impossible. Because if all incorrect messages were complements of the matched message, they would collapse to the same hypothesis. Therefore, MirrorMark spreads the symbol hypotheses evenly over the unit interval. We assign each message a pivot
\begin{equation}
\begin{aligned}
\psi_M=\frac{M}{2^{m+1}},
\end{aligned}
\label{eq:mirror_pivot}
\end{equation}
and mirror the sampler randomness around this pivot:
\begin{equation}
\begin{aligned}
\Psi(u;\psi_M)=(2\psi_M-u)\bmod 1.
\end{aligned}
\label{eq:mod1_reflect_compact}
\end{equation}
Equivalently, the effective message centers $2\psi_M=\frac{M}{2^m}$ are uniformly spaced on $[0,1)$. This evenly spaces the symbol hypotheses and maximizes the worst-case circular separation within the reflection-based family. The motivation and proof are provided in Appendix~\ref{motivation_mod1}.

Specifically, for every fixed $\psi_M$, the map $u\mapsto(2\psi_M-u)\bmod 1$ is bijective and measure-preserving on $[0,1)$. Hence, replacing $u$ with $\Psi(u;\psi_M)$ preserves the uniform randomness required by the base sampler. Therefore, when the base sampler is distortion-free under uniform randomness, MirrorMark preserves the output token distribution. The proof is given in Appendix~\ref{Distortion-freeness}.

\paragraph{Remark: one-bit specialization.}
For $m=1$, we use equation~\ref{eq:mod1_m1} in the main construction, which follows the same mirroring principle and gives the convenient identity as $\Psi(u;0)+\Psi(u;1)=1$, which simplifies the theoretical analysis of tournament-based MirrorMark in Appendix~\ref{tour_eer}. Appendix~\ref{1_bit_justification} compares this specialization with the general formula in equation~\ref{eq:mod1_reflect_compact}.

\subsection{Context-Anchored Balanced Scheduler (CABS)}
\label{sec:cabs}

To embed a payload of $b=m\cdot H$ bits, MirrorMark represents it as $\texttt{MsgSeq}\in\{0,\ldots,2^m-1\}^H$, where each of the $H$ positions carries one $m$-bit symbol. Each eligible generation step is assigned to a position and embeds the corresponding symbol via mod-1 mirroring.

Existing pseudorandom position schedulers in MPAC and StealthInk may allocate too few tokens to some positions under limited token budgets, making those symbols hard to decode. An intuitive alternative is to preferentially assign tokens to positions that are underrepresented. Yet, this strategy is fragile: even a few token insertions or deletions can desynchronize the assigned positions from those used at generation, thereby destroying the watermark. CABS addresses both issues by combining balanced position allocation with context-anchored framing.

\begin{figure*}[htbp]
    \centering
    \vspace{-0.5em}
    \includegraphics[width=0.78\linewidth]{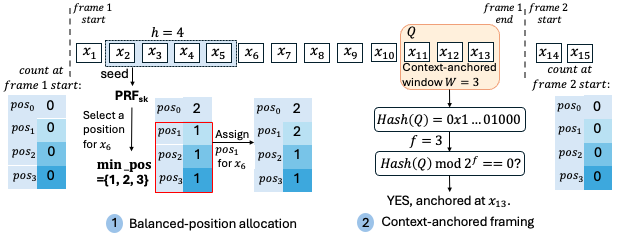}
    \vspace{-0.5em}
    \caption{Overview of CABS with $H=4$ message positions.}
    \label{cabs_figure}
    \vspace{-1em}
\end{figure*}

As shown in~Fig.~\ref{cabs_figure}, within each frame, CABS assigns an eligible token to one of the currently least-populated positions, with tie-breaking determined by the secret key and the context. This encourages every position to receive sufficient evidence for decoding while avoiding a deterministic sequential schedule. Frame boundaries are determined by a context-anchored window $Q$ of $W$ tokens, i.e., a new frame is anchored when the $f$ least significant bits of $\mathrm{Hash}(Q)$ are all zero, or when the frame reaches a maximum length. At the start of a new frame, the position counts are reset.

This framing localizes the effect of edits. A token insertion or deletion may affect the current frame and nearby boundary decisions, but it does not desynchronize the entire sequence. With the maximum frame length $\texttt{max\_len}$, the affected region is bounded by $\texttt{max\_len}+W$. Therefore, edits mainly reduce the amount of usable evidence in local regions, leading to gradual degradation rather than catastrophic failure. Algorithm~\ref{alg:CABS} in Appendix~\ref{cabs_scheduling_algo} gives the full procedure. The encoder is summarized in Algorithm~\ref{alg:encoder} in Appendix~\ref{cabs_encoder}.

\vspace{-0.5em}

\subsection{Decoding and Detection}
\label{sec:decoding_detection}

Given a generated text, the detector first replays CABS using the secret key and observed context tokens to assign each eligible token to a position in $\texttt{MsgSeq}$. It then reconstructs the pseudorandom values used by the base sampler. For each position $pos\in\{1,\ldots,H\}$ and candidate symbol $M\in\mathcal{M}$, the detector mirrors the random values assigned to $pos$ and computes a sampler-specific score. The decoded symbol is selected as the candidate with the largest score:
\begin{equation}
\begin{aligned}
C_{pos,M}
&=
\mathsf{Score}_{\mathrm{base}}
\left(
\{\Psi(u;\psi_M):u\in\mathcal{U}_{pos}\}
\right),
\qquad
\widehat{M}_{pos}
=
\arg\max_{M\in\mathcal{M}} C_{pos,M}.
\end{aligned}
\label{eq:position_score_and_decoder_general}
\end{equation}
Here $\mathcal{U}_{pos}$ denotes the reconstructed random values assigned to position $pos$, and $\mathsf{Score}_{\mathrm{base}}$ is the score function induced by the chosen base sampler.

After decoding all positions, the detector mirrors each token using the decoded symbol of its assigned position and aggregates the resulting evidence into a global watermark score:
\begin{equation}
\begin{aligned}
C_{\mathrm{global}}
=
\mathsf{Score}_{\mathrm{base}}
\left(
\{\Psi(u;\psi_{\widehat{M}_{pos}}): u\in\mathcal{U}_{pos},\ pos=1,\ldots,H\}
\right).
\end{aligned}
\label{eq:global_score_general}
\end{equation}
The text is declared watermarked if $C_{\mathrm{global}}$ exceeds a predefined threshold.

In our evaluations, we instantiate MirrorMark with two representative distortion-free zero-bit samplers: Gumbel-max sampling used by AA~\citep{Aaronson2023} and tournament sampling used by SynthID~\citep{Dathathri2024}. These samplers are chosen because they are distortion-free and therefore do not degrade text quality, while also exposing reproducible random values that can be mirrored during generation and detection. Appendix~\ref{app:warmup} reviews the two base samplers, and gives the concrete decoding and detection formulas for the Gumbel-max and tournament instantiations. Algorithm~\ref{alg:decoder} in Appendix~\ref{cabs_decoder} summarizes the full procedure of decoding and detection.

\vspace{-0.5em}
\section{Theoretical EER of MirrorMark Applied on AA and SynthID}\label{sec:theoretical_eer}

In this section, we analyze the theoretical equal error rate (EER) when MirrorMark is instantiated with two representative samplers, Gumbel-max sampling and tournament sampling, in the single-position setting ($H=1$), and compare the asymptotic estimates with empirical EERs. See proof details in Appendix~\ref{app:theoretical_eer}.

As shown in~\eqref{eq:EER-gumbel-N-asym}, the EER of Gumbel-max-based MirrorMark is governed by the next-token entropy. Since we use top-$k$ sampling with $k=100$, we denote the entropy of the truncated distribution as $\mathcal{H}_{top100}$. For tournament-based MirrorMark,~\eqref{eq:eer-tail} shows that the EER depends on the layer-wise SynthID collision probability $C_{wm}^{\ell}$, i.e., the probability that two independently sampled candidates collide at layer $\ell$.

To validate the analysis, we randomly select 500 prompts from C4~\citep{raffel2020exploring} and generate $T=200$ tokens using LLaMA-2-7B~\citep{touvron2023llama} with temperature $\tau=1.0$. We embed an $m=1$-bit message using both Gumbel-max-based and tournament-based MirrorMark. For Gumbel-max-based MirrorMark, we collect $\mathcal{H}_{top100}$ at each generation step and compute the per-sequence alternative-hypothesis parameters $(\mu_{\mathcal{H}_1}^j,\sigma_{\mathcal{H}_1}^j)$ using~\eqref{H1_distribution}. For tournament-based MirrorMark, we compute the averaged layer-wise collision probabilities $\bar{C}_{wm,j}^{\ell}=\frac{1}{T}\sum_{t=1}^{T}C_{wm,j,t}^{\ell}$ for each sequence $j$, plug them into~\eqref{mu_S_and_v_S}, and obtain $(\mu_{\mathcal{H}_1}^{j},\sigma_{\mathcal{H}_1}^{j})$ using~\eqref{eq:tour-H1}.

For both variants, we approximate the resulting mixture of per-sequence Gaussians with a single Gaussian $\mathcal{N}(\bar{\mu}_{\mathcal{H}_1},\bar{\sigma}_{\mathcal{H}_1}^{2})$ by moment matching:
\begin{equation}
\begin{aligned}
\bar{\mu}_{\mathcal{H}_1}
=
\frac{1}{n}\sum_j \mu_{\mathcal{H}_1}^{j},
\qquad
\bar{\sigma}_{\mathcal{H}_1}^{2}
=
\frac{1}{n}\sum_j(\sigma_{\mathcal{H}_1}^{j})^2
+
\frac{1}{n}\sum_j(\mu_{\mathcal{H}_1}^{j}-\bar{\mu}_{\mathcal{H}_1})^2.
\end{aligned}
\end{equation}
We then compute the asymptotic EER using~\eqref{eq:EER-gumbel-N-asym} for Gumbel-max and~\eqref{eq:eer-compact} for tournament sampling. Table~\ref{theoretical_eer} compares the empirical and asymptotic EERs. The estimates are close to the empirical values for both sampler instantiations, showing that the theoretical analysis captures the practical detection behavior.

\begin{table}[htbp]
\centering
\vspace{-0.5em}
\caption{Empirical and asymptotic EERs of MirrorMark.}
\label{theoretical_eer}
\footnotesize
\setlength{\tabcolsep}{5pt}
\begin{tabular}{lcc}
\toprule
 & Gumbel-max & Tournament \\
\midrule
Empirical & 0.0010 & 0.0020 \\
Asymptotic & 0.0008 & 0.0028 \\
\bottomrule
\end{tabular}
\vspace{-1em}
\end{table}

\vspace{-0.5em}
\section{Evaluations}
\label{sec:evaluation}

We evaluate MirrorMark from three perspectives. First, we conduct controlled comparisons under the Gumbel-max sampler and DiPmark sampler to isolate the effect of the mapping rule. Second, we compare MirrorMark with state-of-the-art multi-bit watermarking schemes in terms of detectability and text quality. Third, we evaluate robustness under editing attacks and ablate the key design choices of CABS. We use LLaMA2-7B~\citep{touvron2023llama} with temperature $1.0$ and top-$100$ sampling on 500 prompts sampled from the RealNewsLike subset of C4~\citep{raffel2020exploring}. The experimental setup and additional details are provided in the Appendix~\ref{setup}. %

\subsection{Controlled Comparison: Mirroring vs. Location-Based Mapping}
\label{sec:controlled_mapping_comparison}

To isolate the effect of the message mapping rule, we conduct controlled comparisons under two base generation mechanisms with $H=1$. First, under Gumbel-max sampling, we compare MirrorMark with the Gumbel-max-based multi-bit extension in ThreeBricks~\citep{fernandez2023three}. Both methods use the same base sampler and differ only in how messages are mapped to sampling randomness. Second, under permutation-based reweighting, we compare a mirroring-based DiPmark variant with StealthInk~\citep{jiang2025stealthink}. Both methods use context-seeded permutations and reweighting, but StealthInk maps each message to a contiguous rank interval, whereas the mirroring-based variant transforms the permutation ranks through mod-1 mirroring. Thus, the two comparisons test whether mirroring improves over location-based mappings in both sampling-randomness and permutation-rank settings.
\begin{figure*}[htbp]
  \centering
  \begin{minipage}[t]{0.32\textwidth}
    \centering
    \includegraphics[width=\linewidth]{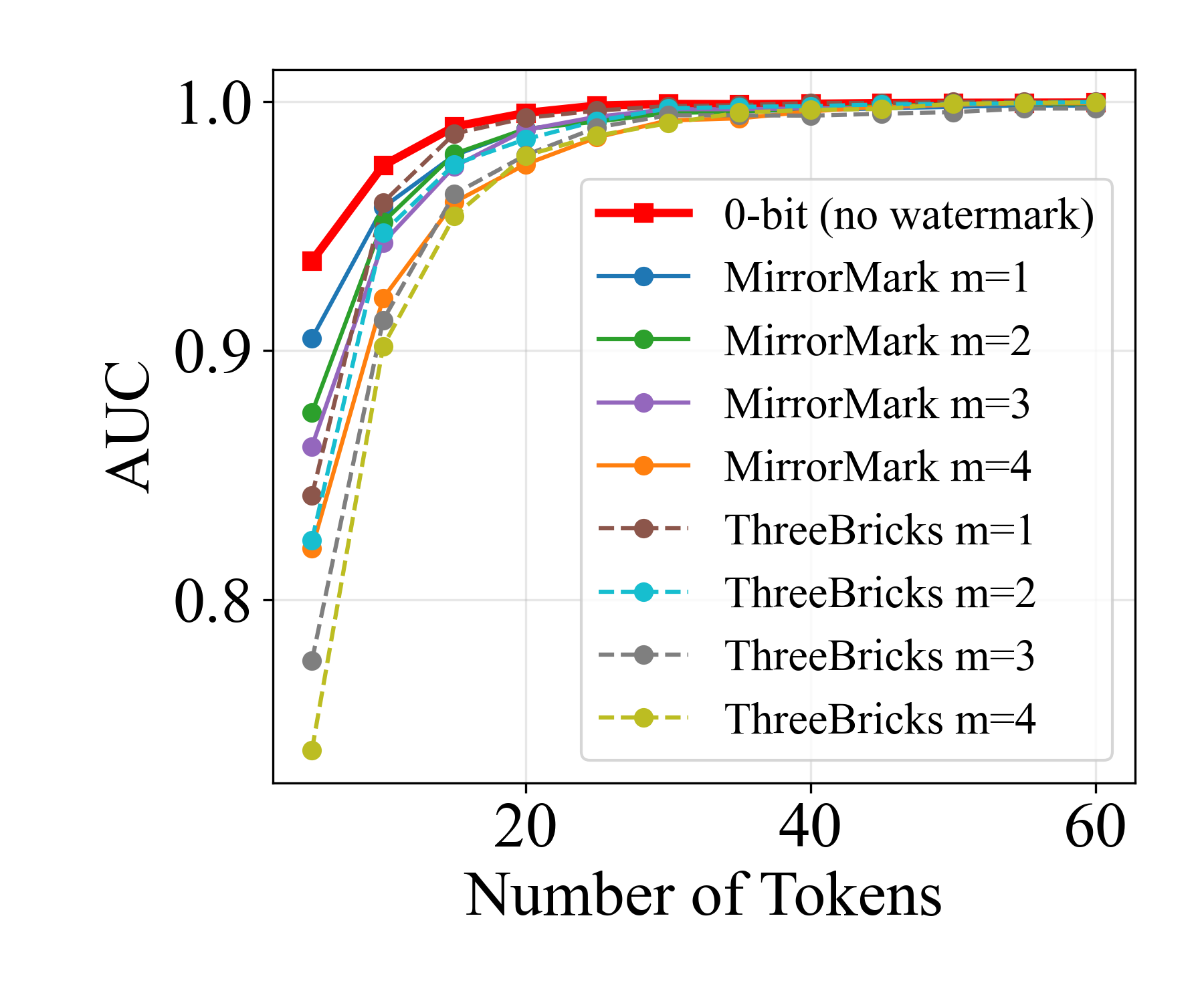}
  \end{minipage}\hspace{-0.2em}
  \begin{minipage}[t]{0.32\textwidth}
    \centering
    \includegraphics[width=\linewidth]{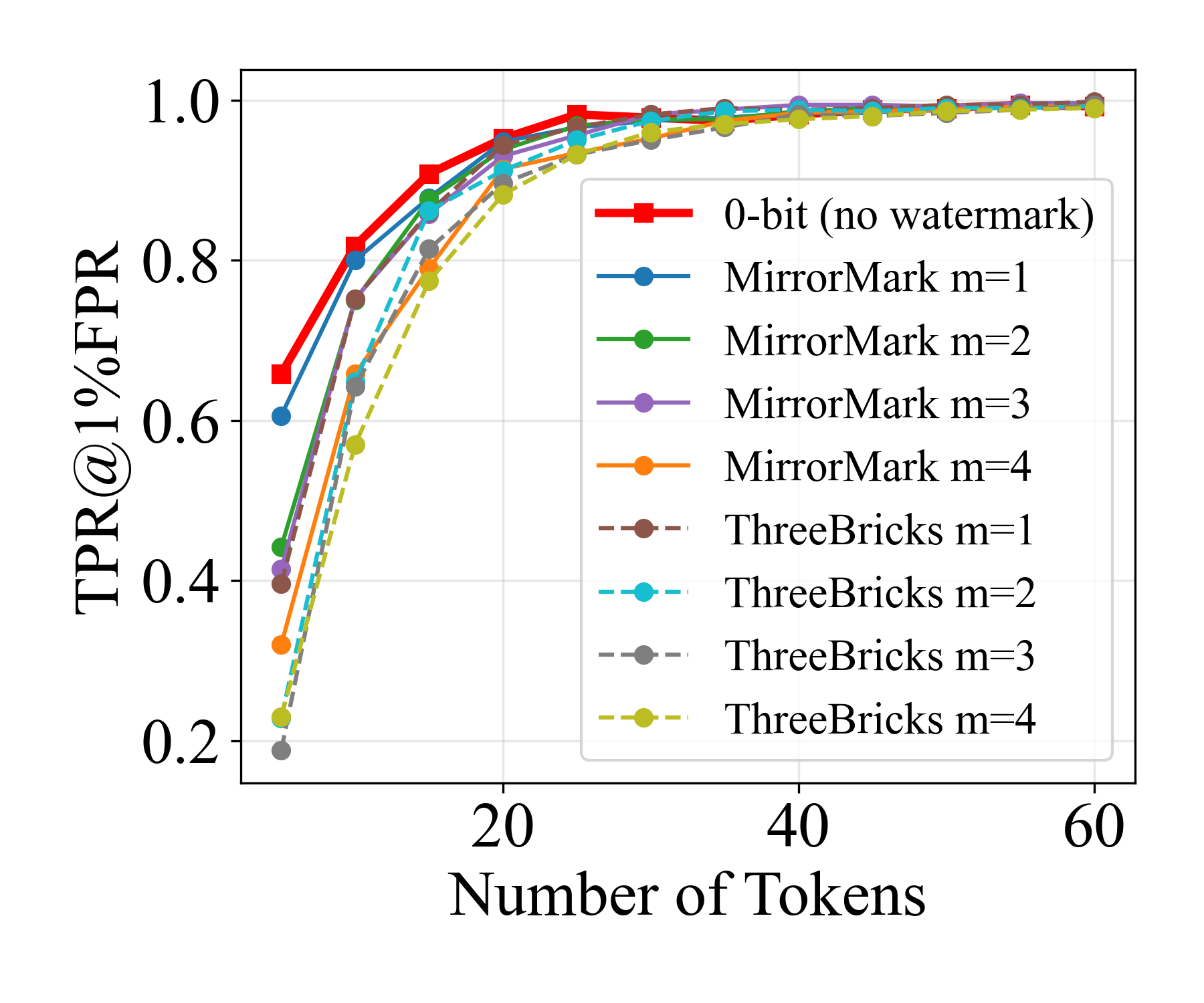}
  \end{minipage}\hspace{-0.2em}
  \begin{minipage}[t]{0.32\textwidth}
    \centering
    \includegraphics[width=\linewidth]{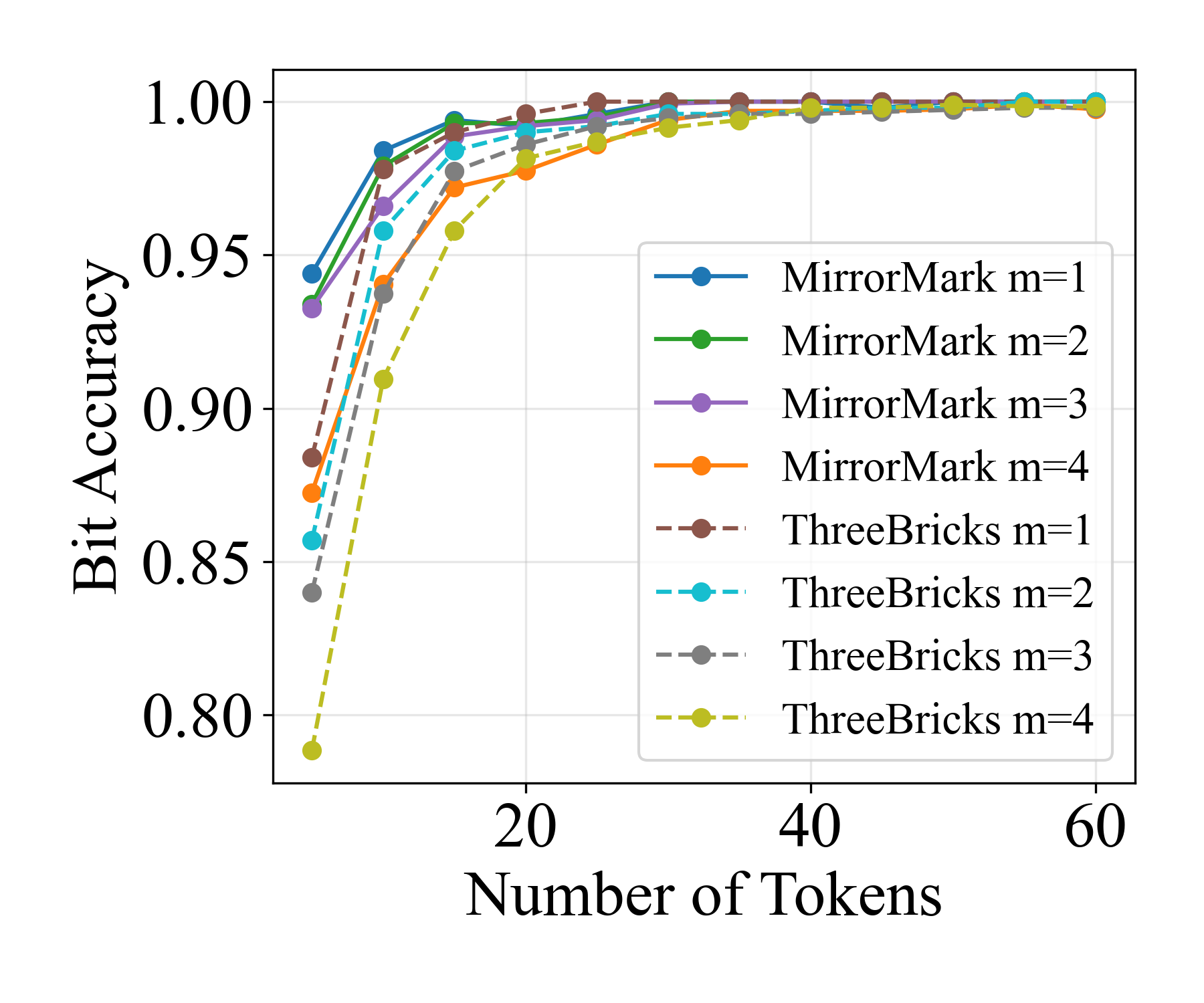}
  \end{minipage}
  \vspace{-0.7em}
  \caption{Controlled comparison between Gumbel-max-based MirrorMark and the Gumbel-max-based multi-bit extension in ThreeBricks under $H=1$.}
  \label{fig:gumbel_control}
  \vspace{-1.0em}
\end{figure*}

As shown in Fig.~\ref{fig:gumbel_control}, MirrorMark consistently improves bit accuracy over ThreeBricks under the same Gumbel-max sampler, especially as the symbol size $m$ increases. This supports the binary-tokenizer analysis in Section~\ref{sec:binary_tokenizer_mapping}, i.e., cyclic shifts place symbol hypotheses at different locations but do not make mismatched hypotheses complementary to the matched one. Appendix~\ref{app:dip_control} provides a similar controlled comparison under permutation-based reweighting, where mirroring-based DiPmark outperforms the interval-based mapping used by StealthInk. These results suggest that the advantage of mirroring is not tied to a specific sampler; rather, it comes from creating stronger matched--mismatched separation than location-based mappings.

\subsection{Trade-off between Text Quality and Detectability}

\begin{table*}[htbp]
\centering
\small
\setlength{\tabcolsep}{2.0pt}       
\renewcommand{\arraystretch}{0.96}   
\vspace{-0.3em}

\caption{Mean perplexity and detectability for different approaches on 300 tokens.
Each perplexity is reported with a 90\% confidence interval based on bootstrapping.}
\label{ppl_tradeoff_300tokens}

\begin{tabular}{
@{}L
C C C S[table-format=2.4]
!{\vrule width 0.5pt}
C C C S[table-format=2.4]
@{}}
\toprule
\multirow{2}{*}{Method}
& \multicolumn{4}{c}{36 Bits}
& \multicolumn{4}{c}{54 Bits} \\
\cmidrule(lr){2-5}\cmidrule(lr){6-9}
& {\footnotesize AUC}
& {\footnotesize TPR@1\%FPR}
& {\footnotesize Bit Acc.}
& {\footnotesize Perplexity}
& {\footnotesize AUC}
& {\footnotesize TPR@1\%FPR}
& {\footnotesize Bit Acc.}
& {\footnotesize Perplexity} \\
\midrule

Non Watermark
& -- & -- & -- &
\shortstack[c]{\meanppl{7.2784}\\[-2pt]\smallci{7.1294,\,7.4296}}
& -- & -- & -- &
\shortstack[c]{\meanppl{7.2784}\\[-2pt]\smallci{7.1294,\,7.4296}} \\

\cmidrule(lr){1-5}\cmidrule(lr){6-9}

MPAC
& 0.9949 & 0.9800 & 0.9347 &
\shortstack[c]{\meanppl{9.1951}\\[-2pt]\smallci{9.0404,\,9.3516}}
& 0.9962 & 0.9840 & 0.8928 &
\shortstack[c]{\meanppl{9.3457}\\[-2pt]\smallci{9.1704,\,9.5224}} \\

RSBH
& \textbf{0.9998} & \textbf{0.9980} & \textbf{1.0000} &
\shortstack[c]{\meanppl{32.8955}\\[-2pt]\smallci{31.4973,\,34.3369}}
& 0.9989 & 0.9980 & \textbf{0.9928} &
\shortstack[c]{\meanppl{32.8184}\\[-2pt]\smallci{31.3574,\,34.3446}} \\

StealthInk
& 0.9892 & 0.8520 & 0.8896 &
\shortstack[c]{\meanppl{7.8241}\\[-2pt]\smallci{7.6260,\,8.0223}}
& 0.9890 & 0.8900 & 0.8415 &
\shortstack[c]{\meanppl{7.8950}\\[-2pt]\smallci{7.6933,\,8.0974}} \\

Gumbel-max
& \textbf{1.0} & \textbf{1.0} & 0.9835 &
\shortstack[c]{\textbf{\meanppl{7.0486}}\\[-2pt]\smallci{6.8991,\,7.1997}}
& \textbf{1.0} & \textbf{1.0} & 0.9683 &
\shortstack[c]{\textbf{\meanppl{7.1751}}\\[-2pt]\smallci{7.0195,\,7.3383}} \\

Tour-Wmean
& 0.9990 & 0.9920 & 0.9732 &
\shortstack[c]{\meanppl{7.3706}\\[-2pt]\smallci{7.2265,\,7.5202}}
& \textbf{1.0} & \textbf{1.0} & 0.9491 &
\shortstack[c]{\meanppl{7.3295}\\[-2pt]\smallci{7.1828,\,7.4792}} \\

Tour-Bayes
& 0.9992 & 0.9960 & 0.9811 &
\shortstack[c]{\meanppl{7.3706}\\[-2pt]\smallci{7.2265,\,7.5202}}
& \textbf{1.0} & 0.9960 & 0.9576 &
\shortstack[c]{\meanppl{7.3295}\\[-2pt]\smallci{7.1828,\,7.4792}} \\

\bottomrule
\end{tabular}
\end{table*}

We evaluate moderate payloads $b\in\{36,54\}$ with 300 generated tokens. Table~\ref{ppl_tradeoff_300tokens} reports detectability and perplexity. Distortion-based baselines achieve strong detection only with substantial perplexity degradation, while StealthInk preserves text quality but has weaker detection and bit accuracy. MirrorMark achieves strong AUC, TPR@1\%FPR, and bit accuracy while keeping perplexity close to non-watermarked text. For completeness, we also provide results on shorter sequences of 200 tokens and longer sequences of 400 tokens respectively in Table~\ref{ppl_tradeoff_200tokens} and Table~\ref{ppl_tradeoff_400tokens} of Appendix~\ref{200_400tokens}. Besides, we evaluate the GPT4o score and repetition rate across these approaches as in Fig.~\ref{gpt4o} and Table~\ref{tab:repetition_rate} of Appendix~\ref{extra_text_quality}, which demonstrates the superior text quality of MirrorMark. Given that MirrorMark maintains a favorable trade-off in these settings, 
we further stress-test it under larger payload sizes ($b \in \{72,90\}$). The results in Fig.~\ref{fig:detectability_of_MirrorMark} of Appendix~\ref{72bits_90bits} demonstrate that MirrorMark continues to provide competitive detectability, highlighting its scalability beyond what baseline approaches can achieve. 

Note that we report results for both Tour-Bayes and Tour-Wmean for completeness, although their performance is largely comparable. In SynthID, the Bayesian detector demonstrates a clearer advantage when using Bernoulli-distributed $u$ values, which maximize diversity but are inherently binary and thus support only one bit per position. In contrast, our focus is on mod-1 mirroring for multi-bit embedding, where this advantage does not directly carry over. Bayesian detection requires training a detector for each configuration, and we illustrate its training setup and cost in Appendix~\ref{setup}.


\vspace{-0.5em}
\subsection{Robustness}

 To evaluate robustness, we consider copy-paste and paraphrasing attacks. In the copy-paste attack, a fraction $\epsilon$ of non-watermarked text is mixed into watermarked text while preserving the total length; for paraphrasing, we rewrite the watermarked text using the model from~\cite{zhang2020pegasus}. Table~\ref{tab:detect_400tokens_copy_paste} reports copy-paste robustness on 400-token texts with 36 embedded bits, where $\epsilon=0$ denotes clean samples. For MirrorMark, we report the default setting $m=2$, which consistently gives strong robustness. Additional copy-paste results for $\epsilon\in\{0.1,0.3,0.5\}$ and different symbol sizes $m\in\{2,3,4,6\}$ are provided in Appendix~\ref{cp_appendix}. Overall, MirrorMark is more robust than the baselines. 

 \begin{table*}[htbp]
\centering
\small
\setlength{\tabcolsep}{2.5pt}     
\renewcommand{\arraystretch}{1.06} 

\caption{Detectability under copy-paste attacks on 400-token texts with 36 embedded bits ($\epsilon$ denotes the edit fraction).}
\label{tab:detect_400tokens_copy_paste}

\begin{tabular}{@{}L
                C C C !{\vrule width 0.5pt}
                C C C !{\vrule width 0.5pt}
                C C C@{}}
\toprule
\multirow{2}{*}{Method}
& \multicolumn{3}{c}{$\epsilon=0$ (No attack)}
& \multicolumn{3}{c}{$\epsilon=0.2$}
& \multicolumn{3}{c}{$\epsilon=0.4$} \\
\cmidrule(lr){2-4}\cmidrule(lr){5-7}\cmidrule(lr){8-10}
& {\footnotesize AUC} & {\footnotesize TPR@1\%FPR} & {\footnotesize Bit Acc.}
& {\footnotesize AUC} & {\footnotesize TPR@1\%FPR} & {\footnotesize Bit Acc.}
& {\footnotesize AUC} & {\footnotesize TPR@1\%FPR} & {\footnotesize Bit Acc.} \\
\midrule
MPAC
& 0.9970 & 0.9820 & 0.9599
& 0.9753 & 0.8975 & 0.8997
& 0.9593 & 0.7675 & 0.8397 \\

RSBH
& 0.9999 & \textbf{1.0} & \textbf{1.0}
& 0.9697 & 0.0850 & 0.6138
& 0.8455 & 0.0050 & 0.6038 \\

StealthInk
& 0.9941 & 0.9500 & 0.9204
& 0.9705 & 0.8175 & 0.8448
& 0.9172 & 0.4750 & 0.7716 \\

Tour-Wmean
& 0.9997 & \textbf{1.0} & 0.9681
& 0.9981 & 0.9900 & 0.9106
& 0.9825 & 0.8980 & 0.8323 \\

Tour-Bayes
& 0.9996 & \textbf{1.0} & 0.9681
& 0.9978 & 0.9840 & 0.9106
& 0.9900 & 0.9220 & 0.8323 \\

Gumbel-max
& \textbf{1.0} & \textbf{1.0} & 0.9891
& \textbf{1.0} & \textbf{1.0} & \textbf{0.9690}
& \textbf{1.0} & \textbf{1.0} & \textbf{0.9328} \\
\bottomrule
\end{tabular}
\end{table*}

Table~\ref{tab:wm_paraphrasing_comparison} reports paraphrasing robustness. MirrorMark maintains stronger watermarked/non-watermarked separability than prior multi-bit schemes and often its zero-bit baselines. This advantage mainly comes from mod-1 mirroring, which creates asymmetric symbol hypotheses: the correct symbol tends to produce larger mirrored values, while incorrect symbols are pushed toward smaller values. Thus, even when paraphrasing weakens token-level evidence, the maximum score still separates watermarked from non-watermarked text more effectively than cyclic-shift or reweighting-based mappings.

\begin{table*}[htbp]
\centering
\caption{Detectability under paraphrasing attacks for different schemes. TB denotes Tour-Bayes and G-max denotes Gumbel-max. Multi-bit methods embed 36 bits in 400 tokens.}
\resizebox{\textwidth}{!}{%
\begin{tabular}{lccccccccccccc}
\toprule
 &
MPAC & RSBH & StealthInk &
TB (0 bit) &
\makecell{TB\\(m=2)} &
\makecell{TB\\(m=3)} &
\makecell{TB\\(m=4)} &
\makecell{TB\\(m=6)} &
\makecell{G-max\\(0 bit)} &
\makecell{G-max\\(m=2)} &
\makecell{G-max\\(m=3)} &
\makecell{G-max\\(m=4)} &
\makecell{G-max\\(m=6)} \\
\midrule
AUC &
0.5743 & 0.3414 & 0.5188 &
0.7925 & 0.8139 & 0.9001 & 0.8938 & 0.8140 &
0.8245 & 0.9306 & 0.9091 & 0.9109 & 0.9025 \\
TPR@1\%FPR &
0.0100 & 0.0000 & 0.0050 &
0.2630 & 0.2220 & 0.3200 & 0.3480 & 0.2300 &
0.2800 & 0.5780 & 0.4860 & 0.4620 & 0.3840 \\
Bit Accuracy &
0.5734 & 0.6220 & 0.5673 &
-- & 0.5216 & 0.5152 & 0.5152 & 0.5123 &
-- & 0.5434 & 0.5398 & 0.5378& 0.5333 \\
\bottomrule
\end{tabular}%
}
\label{tab:wm_paraphrasing_comparison}
\end{table*}

However, all methods show poor bit accuracy after paraphrasing. Payload recovery requires the true symbol to be the top-scoring hypothesis, $\widehat{M}=\arg\max_M C_M$, which is stricter than detection. Paraphrasing changes tokens and their corresponding $u$ values, disrupting this fine-grained ordering. Thus, MirrorMark retains partial evidence for tamper-evident provenance detection, while reliable payload recovery under strong rewriting remains open.

\section{Ablation Study}

We compare CABS with the position schedulers used in MPAC and RSBH, denoted as \emph{NaiveHash} and \emph{DPHash}. As shown in Appendix~\ref{cabs_ablation_scheduler}, CABS achieves near-uniform token allocation with a near-zero Gini score and substantially improves MirrorMark under both Gumbel-max and Tour-Bayes. In contrast, MPAC is less sensitive to the scheduler, with only marginal bit-accuracy gains. These results show that balanced allocation is especially important for MirrorMark, which aggregates evidence across positions.

We then ablate the key CABS parameters, including frame size $f$, context window $W$, and maximum frame expansion factor $\texttt{max\_factor}$. Full results under insertion, deletion, and substitution attacks are reported in Tables~\ref{tab:cabs-insert-final}--\ref{tab:cabs-sub-final} of Appendix~\ref{ablation_cabs_parameters}. Across attacks and edit ratios, $f=3$, $W=4$, and $\texttt{max\_factor}=1.5$ provide the best overall trade-off between robustness and bit recovery, and are used as the default configuration.

\section{Conclusion}
We presented MirrorMark, a generalizable mapping-centric approach to multi-bit LLM watermarking. MirrorMark separates the symbol mapping rule from the base watermarking sampler and uses mod-1 mirroring to transform detector-reproducible pseudorandom objects, including sampling values and permutation ranks. This measure-preserving rule creates stronger matched--mismatched separation while preserving the relevant randomness distribution, enabling distortion-free sampling and natural extension to permutation-based reweighting. We further introduced CABS to balance token allocation and localize edit effects. With Gumbel-max and tournament sampling, we derived theoretical EER estimates that align with empirical results. Experiments show that MirrorMark achieves strong detectability and bit accuracy while maintaining text quality comparable to non-watermarked generation. Future work includes improving payload recovery under strong paraphrasing and strengthening post-edit provenance detection.
\bibliography{cite}
\bibliographystyle{icml2025}

\appendix
\clearpage
\section*{Appendix Contents}
\addcontentsline{toc}{section}{Appendix Contents}

\begin{itemize}
  \item[\textbf{A}] Related Work
  \begin{itemize}
    \item[A.1] \hyperref[related-zero-bit]{Zero-bit Watermarking}
    \item[A.2] \hyperref[related-multi-bit]{Multi-bit Watermarking}
  \end{itemize}

  \item[\textbf{B}] \hyperref[app:binary_mapping_gap]{Derivation of Score Gaps for Binary-Tokenizer Mapping Rules}

  \item[\textbf{C}] \hyperref[motivation_mod1]{Motivation for Mod-1 Mirroring}

  \item[\textbf{D}] \hyperref[app:warmup]{Representative Base Samplers and MirrorMark Instantiations}
  \begin{itemize}
    \item[D.1] \hyperref[app:base_samplers]{Base Samplers}
    \item[D.2] \hyperref[sec_detection]{Zero-Bit Detection Scores}
    \item[D.3] \hyperref[app:mirror_decoder_details]{MirrorMark Decoding and Detection Instantiations}
  \end{itemize}

  \item[\textbf{E}] \hyperref[Distortion-freeness]{Distortion-Freeness of MirrorMark}

  \item[\textbf{F}] \hyperref[cabs_scheduling_algo]{CABS Scheduling Algorithm}

  \item[\textbf{G}] \hyperref[cabs_encoder]{CABS-Based Encoder Algorithm}

  \item[\textbf{H}] \hyperref[cabs_decoder]{CABS-Based Decoding and Detection Algorithm}

  \item[\textbf{I}] Theoretical EER of MirrorMark
  \begin{itemize}
    \item[I.1] \hyperref[gumbel_eer]{Gumbel-Max-Based MirrorMark}
    \item[I.2] \hyperref[tour_eer]{Tournament-Sampling-Based MirrorMark}
    \item[I.3] \hyperref[app:token_independence]{Empirical Validation of Token-Level Dependence}
  \end{itemize}

  \item[\textbf{J}] \hyperref[lemmas]{Lemmas}

  \item[\textbf{K}] \hyperref[setup]{Experimental Setup}

  \item[\textbf{L}] Additional Results
  \begin{itemize}
    \item[L.1] \hyperref[app:dip_control]{Controlled Comparison under Permutation-Based Reweighting}
    \item[L.2] \hyperref[200_400tokens]{Performance Comparison over 200 and 400 Tokens}
    \item[L.3] \hyperref[extra_text_quality]{Repetition Score and LLM-as-Judge Evaluation}
    \item[L.4] \hyperref[72bits_90bits]{Performance with 72-Bit and 90-Bit Payloads}
    \item[L.5] \hyperref[calibration]{Threshold Calibration under Argmax Decoding}
    \item[L.6] \hyperref[cross_lang]{Cross-Language Adaptation}
    \item[L.7] \hyperref[sec:alignment]{Empirical EER Comparison}
    \item[L.8] \hyperref[instruction-task]{Detectability and Text Quality on Instruction Tasks}
    \item[L.9] \hyperref[cp_appendix]{Copy-Paste Robustness with Different Symbol Sizes}
    \item[L.10] \hyperref[1_bit_justification]{One-Bit Specialization of Mod-1 Mirroring}
    \item[L.11] Ablation Study for CABS
    \begin{itemize}
      \item[L.11.1] \hyperref[ablation_cabs_parameters]{Sensitivity to CABS Parameters}
      \item[L.11.2] \hyperref[cabs_ablation_scheduler]{Effect of Position Allocation Schedulers}
    \end{itemize}
  \end{itemize}
\end{itemize}

\clearpage
\newpage
\onecolumn
\section*{Appendix}

\section{Related Work}

\subsection{Zero-bit watermarking}\label{related-zero-bit}

Due to the discrete linguistic nature of text, designing effective watermarking schemes for digital text remains a challenging problem~\citep{shih2017digital}. Early approaches were primarily rule-based, including paraphrasing~\citep{atallah2002natural}, syntactic restructuring~\citep{atallah2001natural}, and synonym substitution~\citep{topkara2006hiding}. However, these methods relied on handcrafted transformations and were limited in scalability, naturalness, and robustness. The emergence of LLMs created new opportunities for watermarking because they are generative by nature, producing text token by token under probabilistic distributions. This generative process allows watermarking to be embedded directly in the sampling procedure rather than through post hoc text modifications. For example, KGW~\citep{kirchenbauer2023watermark} introduces the first in-generation watermarking scheme for LLMs and highlights a key property of reweighting-based watermarking: the watermark can be detected algorithmically without knowledge of the model parameters or access to the LLM API. Their method partitions the vocabulary into red and green token lists using a hash function seeded with the preceding context tokens, and then applies a small bias to the logits of green-list tokens. As a result, the watermarked LLM is more likely to generate green-list tokens. Detection is achieved by reconstructing the same lists and conducting hypothesis testing to evaluate whether a text was generated under the reweighted distribution. Subsequent works strengthened this family along several largely complementary dimensions. In particular, \cite{zhao2023provable} and~\cite{kirchenbauer2023reliability} provide formal robustness guarantees against distortion-bounded editing attacks, including insertion, deletion, and substitution. Focusing on deployment where detection must be publicly accessible, \cite{liu2024an} propose an unforgeable publicly verifiable watermark that decouples watermark generation and detection so that verification can proceed without revealing the generation key, mitigating counterfeiting risks in public detection settings. \cite{zhao2025can} explore whether decoding-time watermarks can support model-level IP infringement detection “for free” (i.e., without modifying training pipelines), and identifies practical failure modes such as query sensitivity and hash-key instability, motivating more reliable infringement-oriented detection procedures. Specifically, by retaining the configurations of red-green list, \cite{hu2023unbiased} and  DiPmark~\citep{wu2023dipmark} introduce an evolved family of permutation-based reweighting strategies for watermarking which maintains the expected distribution of the text; i.e., they proposed a stealthy or unbiased reweighting strategy for LLM watermarking. However, the detector in~\cite{hu2023unbiased} necessitates access to both the prompt and the output distribution provided by the LLM for
a given prompt, which requires the detector to know the prompt used to generate the text.

In contrast to distortion-based watermarking, which embeds signals by perturbing token probability distributions, recent works have explored distortion-free approaches based on inverse sampling. For example, \cite{christ2023undetectable} and \cite{kuditipudi2023robust} propose generating watermarked text without modifying the underlying distribution. However, the method in~\cite{christ2023undetectable} leaves open the challenge of resilience against text corruption. The scheme of~\cite{kuditipudi2023robust}, although tailored for robust detection, depends on hundreds of resampling steps during detection, which is computationally expensive for long texts. Beyond inverse sampling, other distortion-free watermarking techniques have also been proposed. \cite{Aaronson2023} introduce a Gumbel-max–based watermark, which is further extended by~\cite{fugumbelsoft} to improve generation diversity. SynthID~\citep{Dathathri2024} embeds watermarks through tournament sampling. Furthermore,~\cite{he2025theoretically} provides a theoretical framework that characterizes fundamental trade-offs between detectability, distortion, and false positive rate in LLM watermarking, and highlights the need for distribution-adaptive designs to achieve reliable detection under strict constraints.

\subsection{Multi-bit watermarking}\label{related-multi-bit}

 ThreeBricks~\citep{fernandez2023three} extends the schemes of KGW~\citep{kirchenbauer2023watermark} and AA~\citep{Aaronson2023}, respectively, by encoding multi-bit messages through cyclic shifts of the vocabulary permutation or sampling randomness according to the target message, enabling efficient multi-bit identification. Similarly, RSBH~\citep{qu2024provably} construct symbol-dependent cyclic shifts of the vocabulary permutation based on KGW and bias tokens in the corresponding green list to facilitate multi-bit decoding, further incorporating error-correcting codes to improve robustness. However, these approaches fundamentally formulate multi-bit recovery as a multi-class identification problem, where different messages correspond to shifted but otherwise symmetric hypotheses. As a result, under a constrained token budget, cyclic shifts alone do not explicitly maximize the statistical separation between the true message and incorrect alternatives, leading to limited decoding reliability as the message space grows. To compensate for this effect, strong bias is often required to amplify the signal. For example, RSBH sets the bias parameter to $\delta = 6.0$ to achieve high bit accuracy, at the cost of substantially increased text distortion. In contrast, MirrorMark uses a structured mirroring construction that increases matched--mismatched contrast while preserving the relevant pseudorandom distribution.
 
 MPAC~\citep{yoo2023advancing} introduces a multi-color technique. In this scheme, the pseudorandom vocabulary permutation (seeded by prior tokens) is partitioned into multiple equal-length segments, each represented by a distinct color. Message bits are then encoded by selecting color segments. For example, dividing the vocabulary into four colors requires a 2-bit message to specify a segment. During generation, the logits of tokens within the chosen color segment corresponding to the message are boosted by a fixed bias, steering the next token toward that segment.

Beyond color-based methods, other approaches focus on reducing or eliminating distortion. StealthInk~\citep{jiang2025stealthink} perturbs the distribution at each generation step but designs the watermark such that the overall distribution is preserved in expectation, maintaining fluency and text quality. However, its interval-based message mapping can provide weaker matched--mismatched separation under limited token budgets. DISC~\citep{boroujeny2024multi} and \cite{zamir2024excuse} propose fully distortion-free multi-bit schemes. These works demonstrate the feasibility of distribution-preserving payload embedding. However, DISC relies on cyclic-shift-based symbol mapping, and both lines of work are built on~\cite{christ2023undetectable} where practical robustness to text edits remains challenging.


\section{Derivation of Score Gaps for Mapping Rules in Section~\ref{sec:binary_tokenizer_mapping} with Binary Tokenizer}
\label{app:binary_mapping_gap}

In this section, we derive the expected per-token matched--mismatched score gaps in Table~\ref{tab:binary_mapping_gap}. Throughout this section, the binary entropy is defined as
\begin{equation}
\begin{aligned}
H_{\mathrm b}(p)
=
-p\log p-(1-p)\log(1-p).
\end{aligned}
\label{eq:binary_entropy_base2}
\end{equation}

Let $p=p_t^b$ for brevity. The binary sampler outputs token `$1$' if $0\leq u<p$ and token `$0$' if $p\leq u<1$. The binary score is
\begin{equation}
\begin{aligned}
s(x,u)
=
\begin{cases}
\ln \frac{1}{u}, & x=1,\\
\ln \frac{1}{1-u}, & x=0.
\end{cases}
\end{aligned}
\label{eq:app_binary_score}
\end{equation}
For the binary case, there is only one mismatched hypothesis. Thus, the per-token matched--mismatched margin is
\begin{equation}
\begin{aligned}
\Delta_t^b(M^\star)
&=
\mathbb{E}\!\left[
s_{M^\star}(x_t^b,u_t^b)-s_M(x_t^b,u_t^b)
\mid M^\star
\right] \\
&=
\mathbb{E}\!\left[
s_{M^\star}(x_t^b,u_t^b)
\mid M^\star
\right]
-
\mathbb{E}\!\left[
s_M(x_t^b,u_t^b)
\mid M^\star
\right],
\end{aligned}
\label{eq:app_binary_gap_def}
\end{equation}

where $M\neq M^\star$.

\paragraph{Matched score.}
Under the matched hypothesis, the aligned randomness is exactly the randomness used to generate the token. Therefore,
\begin{equation}
\begin{aligned}
\mathbb{E}\!\left[s_{M^\star}(x_t^b,u_t^b)\mid M^\star\right]
&=
\int_0^p \ln\frac{1}{u}\,du
+
\int_p^1 \ln\frac{1}{1-u}\,du\\
&=
\int_0^p \ln\frac{1}{u}\,du
+
\int_0^{1-p} \ln\frac{1}{v}\,dv\\
&=
\big(p-p\ln p\big)
+
\big((1-p)-(1-p)\ln(1-p)\big)\\
&=
1-p\ln p-(1-p)\ln(1-p)\\
&=
1+\ln(2)H_{\mathrm b}(p).
\end{aligned}
\label{eq:app_matched_score}
\end{equation}
Thus, all three mappings have the same expected matched score, because all of them preserve the same token-generation rule under the matched message.

\paragraph{Naive Multi-Key Christ Mapping.}
In the Naive Multi-Key Christ mapping, each symbol hypothesis uses an independent pseudorandom value. For any mismatched message $M\neq M^\star$, the reconstructed randomness is independent of the generated token. Therefore,
\begin{equation}
\begin{aligned}
\mathbb{E}\!\left[s_M(x_t^b,u_t^b)\mid M\neq M^\star\right]
&=
\Pr[x_t^b=1]\int_0^1 \ln\frac{1}{u}\,du
+
\Pr[x_t^b=0]\int_0^1 \ln\frac{1}{1-u}\,du\\
&=
p\cdot 1+(1-p)\cdot 1\\
&=
1.
\end{aligned}
\label{eq:app_mk_mismatch}
\end{equation}
Combining equations~\ref{eq:app_matched_score} and~\ref{eq:app_mk_mismatch}, we obtain
\begin{equation}
\begin{aligned}
\Delta_{t,\mathrm{Multi-Key}}^b(M^\star)
=
\ln(2)H_{\mathrm b}(p_t^b).
\end{aligned}
\label{eq:app_mk_gap}
\end{equation}

\paragraph{DISC Shift Mapping.}
For the DISC shift mapping, the two binary messages use shifts $\delta_0=0$ and $\delta_1=\frac{1}{2}$. Without loss of generality, consider $M^\star=0$ and the mismatched hypothesis $M=1$. The mismatched aligned randomness is
\begin{equation}
\begin{aligned}
\widetilde{u}
=
\left(u-\frac{1}{2}\right)\bmod 1.
\end{aligned}
\label{eq:app_disc_shifted_u}
\end{equation}
Therefore, the expected mismatched score is
\begin{equation}
\begin{aligned}
\mathbb{E}\!\left[s_M(x_t^b,u_t^b)\mid M\neq M^\star\right]
=
\int_0^p
\ln\frac{1}{(u-\frac{1}{2})\bmod 1}\,du
+
\int_p^1
\ln\frac{1}{1-\left((u-\frac{1}{2})\bmod 1\right)}\,du.
\end{aligned}
\label{eq:app_disc_integral}
\end{equation}
Let
\begin{equation}
\begin{aligned}
q=\left|p-\frac{1}{2}\right|.
\end{aligned}
\label{eq:app_disc_q}
\end{equation}
Splitting the integral at $u=\frac{1}{2}$ and simplifying gives
\begin{equation}
\begin{aligned}
\mathbb{E}\!\left[s_M(x_t^b,u_t^b)\mid M\neq M^\star\right]
&=
1-\ln(2)-q\ln q-(1-q)\ln(1-q)\\
&=
1+\ln(2)\left(H_{\mathrm b}(q)-1\right).
\end{aligned}
\label{eq:app_disc_mismatch}
\end{equation}
Define
\begin{equation}
\begin{aligned}
\widetilde{H}_{\mathrm{DISC}}(p)
=
\ln(2)\left[
H_{\mathrm b}\!\left(\left|p-\frac{1}{2}\right|\right)-1
\right].
\end{aligned}
\label{eq:app_h_disc_def}
\end{equation}
Then equation~\ref{eq:app_disc_mismatch} can be written as
\begin{equation}
\begin{aligned}
\mathbb{E}\!\left[s_M(x_t^b,u_t^b)\mid M\neq M^\star\right]
=
1+\widetilde{H}_{\mathrm{DISC}}(p).
\end{aligned}
\label{eq:app_disc_mismatch_short}
\end{equation}
Combining equations~\ref{eq:app_matched_score} and~\ref{eq:app_disc_mismatch_short}, we obtain
\begin{equation}
\begin{aligned}
\Delta_{t,\mathrm{DISC}}^b(M^\star)
=
\ln(2)H_{\mathrm b}(p_t^b)
-
\widetilde{H}_{\mathrm{DISC}}(p_t^b).
\end{aligned}
\label{eq:app_disc_gap}
\end{equation}

\paragraph{Swapping Mapping.}
For the swapping mapping, the two symbol hypotheses are complementary. Without loss of generality, consider $M^\star=0$. Then token `$1$' is generated when $u\in[0,p)$ and token `$0$' is generated when $u\in[p,1)$. The mismatched hypothesis $M=1$ uses the complementary randomness $1-u$. Hence,
\begin{equation}
\begin{aligned}
\mathbb{E}\!\left[s_M(x_t^b,u_t^b)\mid M\neq M^\star\right]
&=
\int_0^p \ln\frac{1}{1-u}\,du
+
\int_p^1 \ln\frac{1}{u}\,du\\
&=
1+p\ln p+(1-p)\ln(1-p)\\
&=
1-\ln(2)H_{\mathrm b}(p).
\end{aligned}
\label{eq:app_swap_mismatch}
\end{equation}
Therefore,
\begin{equation}
\begin{aligned}
\Delta_{t,\mathrm{Swap}}^b(M^\star)
&=
\left(1+\ln(2)H_{\mathrm b}(p_t^b)\right)
-
\left(1-\ln(2)H_{\mathrm b}(p_t^b)\right)\\
&=
2\ln(2)H_{\mathrm b}(p_t^b).
\end{aligned}
\label{eq:app_swap_gap}
\end{equation}

The same derivations hold when $M^\star=1$ by symmetry. Thus, under the same binary sampler, the three mappings preserve the same matched score but induce different mismatched scores. In particular, the swapping mapping doubles the margin of Naive Multi-Key Christ and is no smaller than the DISC shift margin. To see the latter, let $q=|p-\frac{1}{2}|$, and $p\in[0, 1]$. From equations~\ref{eq:app_disc_gap} and~\ref{eq:app_swap_gap},
\begin{equation}
\begin{aligned}
\Delta_{t,\mathrm{Swap}}^b(M^\star)
-
\Delta_{t,\mathrm{DISC}}^b(M^\star)
=
\ln(2)\left[
H_{\mathrm b}(p)+H_{\mathrm b}(q)-1
\right]\geq 0,
\end{aligned}
\end{equation}
Therefore, among the three mappings, swapping gives the largest expected per-token matched--mismatched margin by making the mismatched evidence complementary to the matched evidence.

\section{Motivation for Mod-1 Mirroring: Optimal Message-Hypothesis Separation}\label{motivation_mod1}

In this section, we provide the motivation for the mod-1 mirroring design from the perspective of symbol-hypothesis separation. We first show that, for binary messages, mirroring gives the optimal pairwise complementarity between the matched and mismatched hypotheses under the distortion-free constraint. That is, when the matched mirrored randomness $\Psi(u, \psi_{M^\star})$ is close to 1, the mismatched mirrored randomness $\Psi(u, \psi_{M})$ ($M\neq M^\star$) is close to 0. We then show that when the message space contains more than two hypotheses, such binary complementarity cannot be achieved for all pairs simultaneously. This motivates distributing the reflection centers as evenly as possible on the unit circle, which maximizes the worst-case pairwise separation.

We first consider the binary-message case. Let $Z_{M^\star}$ be the mirrored randomness under the matched hypothesis and $Z_M$ be that under the mismatched hypothesis. The decoding scores used in MirrorMark are monotone increasing in the mirrored value. For example, the Gumbel-max decoder uses $-\log(1-\Psi(u,\psi_M))$, and the tournament weighted-mean decoder uses $\Psi(u,\psi_M)$. Therefore, a larger mirrored value gives higher score for the corresponding symbol hypothesis.

For binary messages, the best situation is that the two symbol hypotheses give opposite evidence for the same randomness $u$, and thus whenever the matched value is large, the mismatched value should be small, and vice versa. This is achieved by
\begin{equation}
\begin{aligned}
Z_M=1-Z_{M^\star}.
\end{aligned}
\end{equation}
This relation is optimal for binary messages because both $Z_{M^\star}$ and $Z_M$ must still be uniformly distributed on $[0,1)$ to preserve distortion freeness. The mapping $Z_M=1-Z_{M^\star}$ keeps this marginal distribution unchanged, while making the two hypotheses maximally opposite point by point. Thus, it gives the largest pairwise contrast between the matched and mismatched hypotheses for any score that increases with the mirrored value.

This binary complementarity motivates a reflection-based construction. It suggests that different symbol hypotheses should be arranged as reflected versions of the same sampling randomness. To extend this idea beyond two messages, we assign each message $M$ a reflection center $\psi_M$ on the unit circle. Reflecting $u$ around the center $\psi_M$ gives a point $Z_M$ whose midpoint with $u$ is $\psi_M$, i.e.,
\begin{equation}
\begin{aligned}
\psi_M=\frac{u+Z_M}{2}.
\end{aligned}
\end{equation}
Solving for $Z_M$ yields
\begin{equation}
\begin{aligned}
Z_M=2\psi_M-u.
\end{aligned}
\end{equation}
Since the reflected value may fall outside $[0,1)$, we wrap it back to the unit interval and define the mod-1 mirroring map as
\begin{equation}
\begin{aligned}
Z_M=\Psi(u,\psi_M)=(2\psi_M-u)\bmod 1 .
\end{aligned}
\end{equation}

For two messages $M_a$ and $M_b$, this gives
\begin{equation}
\begin{aligned}
Z_{M_b}
&=\Psi(u,\psi_{M_b}) \\
&=(\Psi(u,\psi_{M_a})+2(\psi_{M_b}-\psi_{M_a}))\bmod 1 \\
&=(Z_{M_a}+\Delta_{b,a})\bmod 1,
\end{aligned}
\end{equation}
where
\begin{equation}
\begin{aligned}
\Delta_{b,a}=2(\psi_{M_b}-\psi_{M_a})\bmod 1 .
\end{aligned}
\end{equation}
Thus, the relative separation between two symbol hypotheses is determined by the circular offset between their effective centers.

When the message space contains more than two hypotheses, exact binary complementarity cannot hold for all pairs simultaneously. Suppose, for contradiction, that for a matched message $M^\star$, every incorrect message could satisfy the binary complement relation. Then for two distinct incorrect messages $M_1\neq M_2$, we would have
\begin{equation}
\begin{aligned}
Z_{M_1}=1-Z_{M^\star},\qquad
Z_{M_2}=1-Z_{M^\star}.
\end{aligned}
\end{equation}
Hence,
\begin{equation}
\begin{aligned}
Z_{M_1}=Z_{M_2},
\end{aligned}
\end{equation}
which makes these two incorrect hypotheses indistinguishable. Therefore, for $m>1$, the design objective cannot be exact complementarity with all incorrect hypotheses. Instead, the goal is to spread all symbol hypotheses as evenly as possible on the unit circle.

Let
\begin{equation}
\begin{aligned}
q_M=2\psi_M\bmod 1
\end{aligned}
\end{equation}
be the effective center of message $M$. The worst-case pairwise separation is
\begin{equation}
\begin{aligned}
\min_{M_a\neq M_b} d_{\mathrm{circ}}(q_{M_a},q_{M_b}),
\end{aligned}
\end{equation}
where
\begin{equation}
\begin{aligned}
d_{\mathrm{circ}}(x,y)=\min\{|x-y|,1-|x-y|\}.
\end{aligned}
\end{equation}
We now show that this quantity is maximized by uniformly spaced effective centers. Sort the $2^m$ effective centers as
\begin{equation}
\begin{aligned}
0\le q_0<q_1<\cdots<q_{2^m-1}<1 .
\end{aligned}
\end{equation}
Define the circular gaps
\begin{equation}
\begin{aligned}
g_i=q_{i+1}-q_i,\quad i=0,\ldots,2^m-2,
\end{aligned}
\end{equation}
and
\begin{equation}
\begin{aligned}
g_{2^m-1}=1-q_{2^m-1}+q_0 .
\end{aligned}
\end{equation}
Since the gaps sum to one,
\begin{equation}
\begin{aligned}
\sum_{i=0}^{2^m-1}g_i=1 .
\end{aligned}
\end{equation}
Therefore,
\begin{equation}
\begin{aligned}
\min_i g_i\le \frac{1}{2^m}.
\end{aligned}
\end{equation}
The minimum pairwise circular distance cannot exceed the minimum adjacent gap. Hence,
\begin{equation}
\begin{aligned}
\min_{M_a\neq M_b} d_{\mathrm{circ}}(q_{M_a},q_{M_b})
\le \frac{1}{2^m}.
\end{aligned}
\end{equation}
This upper bound is achieved by choosing
\begin{equation}
\begin{aligned}
q_M=\frac{M}{2^m},\qquad M=0,1,\ldots,2^m-1 .
\end{aligned}
\end{equation}
Since $q_M=2\psi_M\bmod 1$, this corresponds to
\begin{equation}
\begin{aligned}
\psi_M=\frac{M}{2^{m+1}},\qquad M=0,1,\ldots,2^m-1 .
\end{aligned}
\end{equation}
This is exactly the center assignment used in MirrorMark.

Finally, this separation-maximizing assignment is still distortion free. For every fixed $\psi_M$, the map $\Psi(\cdot,\psi_M)$ is a bijective and measure-preserving map on $[0,1)$. Therefore, if $U\sim\mathrm{Unif}(0,1)$, then $\Psi(U,\psi_M)\sim\mathrm{Unif}(0,1)$ for every message $M$. See Appendix~\ref{Distortion-freeness} for proof details. Thus, among mappings of the form $\Psi(u,\psi)=(2\psi-u)\bmod 1$, uniformly spaced centers, i.e., $\psi=\psi_M=\frac{M}{2^{m+1}}$, maximize the worst-case pairwise circular separation while preserving the marginal randomness required by the underlying distortion-free sampler.

\section{Representative Base Samplers and MirrorMark Instantiations}
\label{app:warmup}

In this appendix, we review the two representative zero-bit samplers used in our experiments, AA and SynthID, and then describe how MirrorMark instantiates decoding and detection on top of their score functions. Let $p(x_1), \dots, p(x_V)$ denote the probability distribution over the $V$-token vocabulary at generation step $t$, given by the LLM as $p_{\mathrm{LM}}(\cdot \mid x_{<t})$.\footnote{In our paper, uppercase characters such as $G$ and $U$ denote random variables, lowercase characters such as $g$ and $u$ denote their realizations, and bold characters such as $\boldsymbol{u}$ denote vectors.}

\subsection{Base Samplers}
\label{app:base_samplers}

\subsubsection{Gumbel-max Sampling}
\label{sec_gumbel}

The classical Gumbel trick~\citep{gumbel1954statistical} samples from the distribution by adding i.i.d.\ $\mathrm{Gumbel}(0,1)$ noise to the log-probabilities:
\vspace{-0.5em}
\begin{equation}
\begin{aligned}
x^\ast
=
\arg\max_i \left\{\log p(x_i)+G_i\right\},
\end{aligned}
\end{equation}

which guarantees $\Pr(x^\ast=x_i)=p(x_i)$. Using the representation $G_i=-\log(-\log U_i)$ with $U_i\sim\mathrm{Uniform}(0,1)$, this is equivalently
\vspace{-0.5em}
\begin{equation}
\label{eq:gumbel-form3}
\begin{aligned}
x^\ast
&=
\arg\max_{1\leq i\leq V}
\left[\log p(x_i)-\log(-\log U_i)\right] \\
&=
\arg\max_{1\leq i\leq V}
U_i^{1/p(x_i)} .
\end{aligned}
\end{equation}

To embed the watermark by Gumbel-max sampling, AA uses the watermark key and context tokens as the seed $r_t$ at step $t$, and sets $u_i=g(x_i,r_t)$ for each token $x_i$, where $g(\cdot,r_t)$ is a PRF with range $\mathrm{Uniform}(0,1)$. This construction embeds the watermark in the sampled token and allows detection by reproducing the same pseudorandom values. Since Gumbel-max sampling preserves the target distribution, the sampling process is distortion-free. However, as shown in~\eqref{eq:gumbel-form3}, the token with the largest $U_i^{1/p(x_i)}$ is always selected, so the generated response is deterministic for the same prompt and key.

\subsubsection{Tournament Sampling}
\label{sec_tournament}

SynthID proposes tournament sampling to embed a zero-bit watermark. Tournament sampling proceeds in $L$ layers. At layer $\ell$, a PRF $g^\ell(\cdot,r_t):\mathcal{V}\rightarrow[0,1]$ assigns each token a value $u^\ell$ using a seed $r_t$ derived from the watermark key and context tokens. Before the tournament starts, $n_0$ candidate tokens $\{c_1,\dots,c_{n_0}\}$ are sampled from the original distribution $p_{\mathrm{LM}}(\cdot\mid x_{<t})$, where $n_0=2^L$. At the first layer, the $n_0$ candidates are randomly paired. At each subsequent layer, the surviving candidates are paired according to the tournament structure. In each match, the token with the larger PRF value wins. After $L$ layers, the final surviving token $x_t$ is chosen as the output. Compared with AA, which deterministically samples the token for the same prompt and key, SynthID is probabilistic and therefore provides more generation diversity.

\subsection{Zero-Bit Detection Scores}
\label{sec_detection}

Given a text $x_1,\dots,x_T$, the detector in AA recomputes $u_t=g(x_t,r_t)$ for $t=1,\dots,T$. If the text is not watermarked, the $u_t$ values follow $\mathrm{Uniform}(0,1)$ i.i.d.; if watermarked, they are skewed toward larger values. AA uses the following score:
\vspace{-0.5em}
\begin{equation}
\label{eq:logscore}
\begin{aligned}
\mathrm{LogScore}(x)
=
-\sum_{t=1}^{T}\log(1-u_t).
\end{aligned}
\end{equation}

Following SynthID, let $u_{t,\ell}:=g^\ell(x_t,r_t)$ denote the value produced by the $\ell$-th tournament layer at step $t$, and let $\alpha_\ell$ be the corresponding layer weight. The weighted mean score is
\vspace{-0.5em}
\begin{equation}
\label{eq:wmeanscore}
\begin{aligned}
\mathrm{WeightedMeanScore}
=
\frac{1}{TL}
\sum_{t=1}^{T}
\sum_{\ell=1}^{L}
\alpha_\ell u_{t,\ell}.
\end{aligned}
\end{equation}

SynthID also uses a Bayesian score that aggregates evidence across tokens and layers:

\begin{equation}
\label{eq:bayes-score}
\begin{aligned}
\mathrm{BayesianScore}(x)
&=
P(w\mid\boldsymbol{u}) \\
&=
\sigma\!\left(
\log\frac{P(w\mid\boldsymbol{u})}{P(\neg w\mid\boldsymbol{u})}
\right) \\
&=
\sigma\!\left(
\log\frac{P(\boldsymbol{u}\mid w)}{P(\boldsymbol{u}\mid\neg w)}
+
\log\frac{P(w)}{1-P(w)}
\right),
\end{aligned}
\end{equation}
where $w$ and $\neg w$ denote the watermarked and non-watermarked hypotheses, respectively, $P(w)$ and $P(\neg w)$ are the priors, and $\sigma(\cdot)$ is the logistic sigmoid. Since $u_{t,\ell}$ is generated by a PRF and follows $\mathrm{Uniform}(0,1)$ under the non-watermarked hypothesis,
\vspace{-0.5em}
\begin{equation}
\label{nonwatermark_likelihood}
\begin{aligned}
P(\boldsymbol{u}\mid\neg w)
&=
\prod_{t=1}^{T}
\prod_{\ell=1}^{L}
P(u_{t,\ell}\mid\neg w) \\
&=
\prod_{t=1}^{T}
\prod_{\ell=1}^{L}
1
=
1.
\end{aligned}
\end{equation}
For the watermarked hypothesis, SynthID models the likelihood as
\vspace{-0.5em}
\begin{equation}
\label{watermark_likelihood}
\begin{aligned}
P(\boldsymbol{u}\mid w)
&=
\prod_{t=1}^{T}
\prod_{\ell=1}^{L}
P(u_{t,\ell}\mid w,u_{t,<\ell}) \\
&=
\prod_{t=1}^{T}
\prod_{\ell=1}^{L}
\sum_{c=1}^{2}
P(u_{t,\ell}\mid \pi_{t,\ell}=c)
P(\pi_{t,\ell}=c\mid w,u_{t,<\ell}),
\end{aligned}
\end{equation}
where $\pi_{t,\ell}\in\{1,2\}$ denotes the number of distinct $u$ values in the pairwise tournament at layer $\ell$ for the $t$-th token. SynthID derives $P(u_{t,\ell}\mid \pi_{t,\ell}=c)$ from the distribution of watermarked $u_{t,\ell}$ conditioned on $\pi_{t,\ell}$, and learns $P(\pi_{t,\ell}=c\mid w,u_{t,<\ell})$ with a logistic regression model.

\subsection{MirrorMark Decoding and Detection Instantiations}
\label{app:mirror_decoder_details}

This subsection gives the concrete decoders used when MirrorMark is instantiated with the two base samplers reviewed above. In all cases, the detector first replays CABS to assign eligible tokens to message positions, reconstructs the corresponding pseudorandom values, and evaluates each candidate message $M\in\mathcal{M}$ by applying the mirrored transformation $\Psi(\cdot;\psi_M)$ before using the base-sampler score.

\subsubsection{Gumbel-max-based MirrorMark}
\label{app:gumbel_decoder}

For Gumbel-max-based MirrorMark, suppose a position receives $K$ tokens with reconstructed values $\{u_i\}_{i=1}^{K}$. The detector decodes the symbol at this position by applying the LogScore in~\eqref{eq:logscore} to the mirrored values:
\vspace{-0.5em}
\begin{equation}
\begin{aligned}
\widehat{M}
=
\arg\max_{M\in\mathcal{M}}
-\sum_{i=1}^{K}
\log\left(1-\Psi(u_i;\psi_M)\right).
\end{aligned}
\label{eq:gumbeldecoder}
\end{equation}

\subsubsection{Tournament-based MirrorMark}
\label{app:tournament_decoder}

For tournament-based MirrorMark, suppose a position receives $K$ tokens and each token has $L$ tournament-layer values $\{u_{t,\ell}\}_{t=1,\ldots,K;\ell=1,\ldots,L}$. With the weighted mean score in~\eqref{eq:wmeanscore}, the decoded symbol is
\vspace{-0.5em}
\begin{equation}
\begin{aligned}
\widehat{M}
=
\arg\max_{M\in\mathcal{M}}
\frac{1}{K}
\sum_{t=1}^{K}
\frac{1}{L}
\sum_{\ell=1}^{L}
\alpha_\ell\Psi(u_{t,\ell};\psi_M).
\end{aligned}
\label{eq:wmeandecoder}
\end{equation}

For the Bayesian decoder, let $U=\{u_{i,\ell}\}_{i=1,\ldots,K;\ell=1,\ldots,L}$ be the reconstructed values assigned to one position. The decoder selects

\begin{equation}
\begin{aligned}
\widehat{M}
=
\arg\max_{M\in\mathcal{M}}
\left\{
\log P(M)+\log P(U\mid M,w)
\right\},
\end{aligned}
\label{eq:bayesiandecoder_compact}
\end{equation}
where $P(U\mid M,w)$ is computed by applying the SynthID likelihood model to the mirrored values $\Psi(u_{i,\ell};\psi_M)$. Then

\begin{equation}
\begin{aligned}
P(U\mid M,w)
=
\prod_{i=1}^{K}
\prod_{\ell=1}^{L}
\sum_{c=1}^{2}
P\big(\Psi(u_{i,\ell};\psi_M)\mid \pi_{i,\ell}=c\big)
P\big(\pi_{i,\ell}=c\mid w,\Psi(u_{i,<\ell};\psi_M)\big).
\end{aligned}
\label{eq:bayesian_likelihood_decoder}
\end{equation}

\subsubsection{Global Detection}
\label{app:global_detection}

After decoding all positions, the detector mirrors each token using the symbol decoded for its assigned position. The mirrored values are aggregated using the corresponding base-sampler score: LogScore for Gumbel-max, WeightedMeanScore for tournament weighted mean, or BayesianScore for tournament Bayesian detection. The text is declared watermarked if the resulting global score exceeds a predefined threshold.

\section{Distortion-freeness of Gumbel-max and tournament-based MirrorMark}\label{Distortion-freeness}

\begin{theorem}[Distortion-freeness of MirrorMark]
Suppose the underlying zero-bit watermarking sampler is distortion-free whenever its randomness follows $U\sim \mathrm{Uniform}(0,1)$. In particular, this holds for AA and SynthID, which are the zero-bit watermarking baselines, where distortion-freeness is established in Appendix B.1.1 and Appendix G of SynthID, respectively. If $U$ is replaced by the mod-1 mirroring transformation $\Psi(U;\psi_M)$ as in~\eqref{eq:mod1_reflect_compact}, where $\psi_M$ is determined by the embedded message symbol $M$, then the resulting sampler remains distortion-free. In particular, both AA-based MirrorMark and SynthID-based MirrorMark are distortion-free.
\end{theorem}

\begin{proof}
Since AA and SynthID are already distortion-free under uniform randomness, it suffices to show that for any fixed $M$,
\begin{equation}
\begin{aligned}
U\sim \mathrm{Uniform}(0,1)
\;\Rightarrow\;
\Psi(U;\psi_M)\sim \mathrm{Uniform}(0,1).
\end{aligned}
\end{equation}

Fix $\psi:=\psi_M= \frac{M}{2^{m+1}}\in[0,1)$. The mod-1 mirroring in~\eqref{eq:mod1_reflect_compact} can be written equivalently as
\begin{equation}\label{mod-1-mirroring-proof}
\begin{aligned}
\Psi(U;\psi)=
\begin{cases}
2\psi-U, & 0\le U\le 2\psi,\\
2\psi-U+1, & 2\psi<U<1.
\end{cases}
\end{aligned}
\end{equation}

\paragraph{To prove $\boldsymbol{\Psi(U;\psi)\sim \mathrm{Uniform}(0,1)}$,}
let $[a,b)\subseteq[0,1)$ be any interval. It suffices to show
\begin{equation}
\begin{aligned}
\Pr(\Psi(U;\psi)\in[a,b)) = b-a.
\end{aligned}
\end{equation}

Since $U\sim \mathrm{Uniform}(0,1)$, this is equivalent to showing that the preimage $\Psi^{-1}([a,b))$ has length $b-a$. Let $u$ denote a realization of $U$. We consider three cases.

\paragraph{Case 1: $[a,b)\subseteq[0,2\psi]$.}
From~\eqref{mod-1-mirroring-proof},
\begin{equation}
\begin{aligned}
a \le 2\psi - u < b
\;\Rightarrow\;
2\psi - b < u \le 2\psi - a.
\end{aligned}
\end{equation}
Thus
\begin{equation}
\begin{aligned}
\Psi^{-1}([a,b)) = (2\psi-b,\;2\psi-a],
\end{aligned}
\end{equation}
whose length is
\begin{equation}
\begin{aligned}
(2\psi-a)-(2\psi-b)=b-a.
\end{aligned}
\end{equation}

\paragraph{Case 2: $[a,b)\subseteq[2\psi,1)$.}
From~\eqref{mod-1-mirroring-proof},
\begin{equation}
\begin{aligned}
a \le 2\psi - u + 1 < b
\;\Rightarrow\;
2\psi + 1 - b < u \le 2\psi + 1 - a.
\end{aligned}
\end{equation}
Thus
\begin{equation}
\begin{aligned}
\Psi^{-1}([a,b)) = (2\psi+1-b,\;2\psi+1-a],
\end{aligned}
\end{equation}
whose length is
\begin{equation}
\begin{aligned}
(2\psi+1-a)-(2\psi+1-b)=b-a.
\end{aligned}
\end{equation}

\paragraph{Case 3: $a<2\psi<b$.}
The interval crosses the split point $2\psi$, so the preimage consists of two parts, where the length from first branch is $2\psi - a$, and the length from second branch is $b - 2\psi$.
Hence
\begin{equation}
\begin{aligned}
|\Psi^{-1}([a,b))|
=
(2\psi-a)+(b-2\psi)
=
b-a.
\end{aligned}
\end{equation}

Therefore, for all intervals $[a,b)\subseteq[0,1)$,
\begin{equation}
\begin{aligned}
\Pr(\Psi(U;\psi)\in[a,b)) = b-a,
\end{aligned}
\end{equation}
which implies
\begin{equation}
\begin{aligned}
\Psi(U;\psi)\sim \mathrm{Uniform}(0,1).
\end{aligned}
\end{equation}

Thus, mod-1 mirroring preserves the uniform distribution. Since AA and SynthID are distortion-free under uniform randomness, replacing $U$ with $\Psi(U;\psi_M)$ does not alter the randomness distribution seen by the sampler, and therefore does not change the resulting token distribution. Hence MirrorMark remains distortion-free.
\end{proof}

\section{CABS Scheduling Algorithm~\ref{alg:CABS}}\label{cabs_scheduling_algo}
Specifically, the function $\textsf{Elig}(\cdot)$ in Algorithm~\ref{alg:CABS} specifies the eligibility condition for watermarking a token, i.e., the context $h$ tokens are not repeated for the current generation step. Since each random value is generated by a PRF seeding the context $h$ tokens, this restriction avoids correlations between consecutive watermarking decisions and helps maintain the statistical independence of the resulting pseudorandom draws.
\begin{algorithm}[htbp]
\caption{CABS Scheduling}\label{alg:CABS}
\begin{algorithmic}[1]
\REQUIRE Eligibility function $\mathsf{Elig}(\cdot)$, secret key $\mathsf{sk}$, message length $H$, counter vector $\boldsymbol{c}\gets \mathbf{0}^H$, queue $Q\gets[\,]$, window size $W$, $f$, count of tokens within a frame $\ell\gets 0$, context length $h$, $\texttt{min\_len}$, $\texttt{max\_factor}$, sequence $\boldsymbol{x}_{0:T-1}$
\ENSURE Position assignment for each eligible token
\STATE $\texttt{max\_len}=\texttt{max\_factor}\times H$
\FOR{$i=h,\dots,T-1$}
    \IF{not $\mathsf{Elig}(\boldsymbol{x}_{i-h:i-1})$}
        \STATE \textbf{continue}
    \ELSE
        \STATE $F \gets \mathsf{Hash}(Q)$
        \STATE $Q.\mathsf{enqueue}(\boldsymbol{x}_i)$
        \IF{$|Q|>W$}
            \STATE $Q.\mathsf{dequeue}(\boldsymbol{x}_{i-W})$
        \ENDIF
        \STATE $min\_pos=\argmin(\boldsymbol{c})$ \%\% Select the positions with the fewest tokens from counter vector for tokens-to-positions mapping 
        \STATE $pos \sim \mathsf{Unif}\big(min\_pos\big)$ \%\% Randomly select a position, seeded by $\mathsf{PRF}_{\mathsf{sk}}(\boldsymbol{x}_{i-h:i-1})$
        \STATE $\boldsymbol{c}_{pos} \gets \boldsymbol{c}_{pos}+1$ \%\% Increment the count for the assigned position
        \STATE $\ell \gets \ell+1$ \%\% Increment the count of tokens for the current frame
        \STATE $\textit{cut } (\mathsf{true} \text{ or } \mathsf{false}) \gets \big(\ell \ge \texttt{min\_len} \ \wedge\ (F \bmod 2^f==0)\big)\ \vee\ (\ell \ge \texttt{max\_len})$ \%\% Whether to end the current frame and start a new one
        \IF{$\textit{cut}$}
            \STATE $\boldsymbol{c}\gets \mathbf{0}^H$, \quad $Q\gets[\,]$, \quad $\ell\gets 0$
        \ENDIF
    \ENDIF
\ENDFOR
\end{algorithmic}
\end{algorithm}

\section{CABS-based Encoder Algorithm~\ref{alg:encoder}}\label{cabs_encoder}

\begin{algorithm}[htbp]
\caption{CABS-based Encoder}\label{alg:encoder}
\begin{algorithmic}[1]
\REQUIRE CABS parameters $\big(\mathsf{sk}, H, W, f, h, \texttt{min\_len}, \texttt{max\_factor}\big)$, prompt $\boldsymbol{a}$, length $T$, message sequence with $H$ positions $\texttt{MsgSeq}\in\{0,\dots,2^m-1\}^H$, original distribution $p_{LM}$, watermarked distribution $p_{wm}$
\ENSURE Generated sequence $\boldsymbol{x}_{0:T-1}$
\STATE $\texttt{max\_len}=\texttt{max\_factor}\times H$
\STATE $\mathsf{cabs}\gets \mathsf{CABS}\big(\mathsf{sk}, H, W, f, h, \texttt{min\_len}, \texttt{max\_len}\big)$
\FOR{$t=0$ to $T-1$}
    \IF{$t<h$}
        \STATE Sample $x_t \sim p_{LM}\big(\cdot \mid \boldsymbol{a}, \boldsymbol{x}_{:t-1}\big)$
    \ELSE
        \STATE $pos \gets \mathsf{cabs}(\boldsymbol{x}_{:t-1})$
        \STATE Sample $x_t \sim p_{wm}\big(\cdot \mid \boldsymbol{a}, \boldsymbol{x}_{:t-1}, \texttt{MsgSeq}[pos]\big)$
    \ENDIF
\ENDFOR
\end{algorithmic}
\end{algorithm}

\section{CABS-based Decoding and Detection Algorithm~\ref{alg:decoder}}\label{cabs_decoder}
\begin{algorithm}[htbp]
\caption{CABS-based Decoding \& Detection}\label{alg:decoder}
\begin{algorithmic}[1]
\REQUIRE Sequence $\boldsymbol{x}_{0:T-1}$, secret key $\mathsf{sk}$, message length $H$, context length $h$, CABS params $\big(\mathsf{Elig}(\cdot), W, f, \texttt{min\_len}, \texttt{max\_factor}\big)$, decoder choice $\texttt{DEC}\in\{\texttt{gumbel},\texttt{wmean},\texttt{bayes}\}$, scorer choice $\texttt{SCORER}\in\{\texttt{gumbel},\texttt{wmean},\texttt{bayes}\}$, threshold $\texttt{thres}$
\ENSURE Message sequence $\texttt{MsgSeq}\in\{0,\dots,2^m-1\}^H$ and a decision $\in\{\mathsf{true},\mathsf{false}\}$ on whether $\boldsymbol{x}_{0:T-1}$ is watermarked
\STATE $\texttt{max\_len}=\texttt{max\_factor}\times H$
\STATE $\mathsf{cabs}\gets\mathsf{CABS}\big(\mathsf{Elig}(\cdot),\mathsf{sk},H,W,f,h,\texttt{min\_len},\texttt{max\_len}\big)$
\STATE Initialize $\mathcal{U} \gets \{\, \textit{pos} : [\,] \mid \textit{pos}=1,\dots,H \,\}$, \quad $\mathcal{U}_{\text{mirror}}\gets[\,]$
\FOR{$t=h,\dots,T-1$}
    \STATE $\textit{pos}\gets \mathsf{cabs}(\boldsymbol{x}_{:t})$
    \STATE Generate random value $u_t$ seeding $\mathsf{sk}$ and $\boldsymbol{x}_{t-h:t}$
    \STATE $\mathcal{U}[\textit{pos}].\mathsf{append}(u_t)$
\ENDFOR
\FOR{$\textit{pos}=1,\dots,H$}
    \STATE $\texttt{MsgSeq}[\textit{pos}]\gets \mathsf{SymbolDecoder}\big(\mathcal{U}[\textit{pos}];\,\texttt{DEC}\big)$ \%\% If $\texttt{DEC}=\texttt{gumbel}$ use~\eqref{eq:gumbeldecoder}; if $\texttt{DEC}=\texttt{wmean}$ use~\eqref{eq:wmeandecoder}; if $\texttt{DEC}=\texttt{bayes}$ use~\eqref{eq:bayesiandecoder_compact}.
    \FOR{each $u\in\mathcal{U}[\textit{pos}]$}
        \STATE $u_{\text{mir}}\gets \Psi\big(u,\,\psi_{\texttt{MsgSeq}[\textit{pos}]})\big)$
        \STATE $\mathcal{U}_{\text{mirror}}.\mathsf{append}(u_{\text{mir}})$
    \ENDFOR
\ENDFOR
\STATE $\textit{score}\gets \mathsf{Score}\big(\mathcal{U}_{\text{mirror}};\,\texttt{SCORER}\big)$ \%\% If $\texttt{SCORER}=\texttt{gumbel}$ use~\eqref{eq:logscore}; if $\texttt{SCORER}=\texttt{wmean}$ use~\eqref{eq:wmeanscore}; if $\texttt{SCORER}=\texttt{bayes}$ use~\eqref{eq:bayes-score}.
\STATE \textbf{return} $\mathsf{\textbf{true}}$ if $\textit{score}>\texttt{thres}$, else $\mathsf{\textbf{false}}$
\end{algorithmic}
\end{algorithm}

\section{Theoretical EER of MirrorMark}\label{app:theoretical_eer}

In the following, we analyze the theoretical EER of Gumbel-max and tournament-based MirrorMark with the number of positions $H$
=1.

\subsection{Gumbel-max-based multibit watermarking}\label{gumbel_eer}

Recall the sequence-level score of text $W$ for message $M$  is derived as follows, where $\mathsf{sk}$ is the watermark key and $u_t$ is the random value seeded by $\mathsf{sk}$, and $h$ context tokens from $W_{t-h:t}$, 
\begin{equation}
    \begin{aligned}
S_{M}(W_t,\mathsf{sk})&=\ln\frac{1}{1-\Psi(u_{t},\psi(M))},\\
C_{M}(W,\mathsf{sk})
&= \frac{1}{T}\sum_{t=1}^T S_{M}(W_t,\mathsf{sk}).
    \end{aligned}
\end{equation}

Under the null hypothesis $\mathcal{H}_{0}$, all $C_M$ share the same non-watermarked distribution. Under the alternative hypothesis $\mathcal{H}_{1}$,
exactly one index $M^\star$ is ``signal'', representing the message embedded by the encoder.

Under $\mathcal{H}_{0}$, $\Psi\sim\mathrm{Uniform}(0,1)$ and hence
$S_{M}(W_t,\mathsf{sk})\stackrel{d}{=}\mathrm{Exp}(1)$. Therefore,
\begin{equation}
\label{eq:gumbel-H0-moments}
\begin{aligned}
    \mathbb{E}[C_{M}(W,\mathsf{sk})\mid\mathcal{H}_{0}]&=\mu_{\mathcal{H}_{0}}=1,\\
\mathrm{Var}(C_{M}(W,\mathsf{sk})\mid\mathcal{H}_{0})&=\sigma_{\mathcal{H}_{0}}^2=\frac{1}{T}.
\end{aligned}
\end{equation}

 Under $\mathcal{H}_{1}$, referring to equation (14) in ThreeBricks~\citep{fernandez2023three}, for the $t$-th watermarked token with bias $p_t\in(0,1]$,
$\Psi(u_t,\psi(M^\star))\sim\mathrm{Beta}\!\big(\tfrac{1}{p_t},1\big)$ so that
$1-\Psi\sim\mathrm{Beta}\!\big(1,\tfrac{1}{p_t}\big)$. According to the digamma function $\psi_{0}$ and trigamma function $\psi_{1}$ defined in~Lemma~\ref{lem:beta-log-digamma},
\begin{equation}
    \begin{aligned}
        \mathbb{E}[S_{M^\star}(W_t,\mathsf{sk})]
&= \psi_{0} \Big(1+\tfrac{1}{p_t}\Big)-\psi_{0}(1)
:= H_{1/p_t}, \\
 \mathrm{Var}\!\big(S_{M^\star}(W_t,\mathsf{sk})\big)
&= \psi_1(1)-\psi_1\!\Big(1+\tfrac{1}{p_t}\Big).
    \end{aligned}
\end{equation}

\paragraph{Relating $p_t$ to the vocabulary size.} In Gumbel-max sampling for an LLM with the vocabulary size of $V$, suppose $\kappa V$ candidates enter a uniform competition, which means each candidate receives an i.i.d. PRF value $U\sim\mathrm{Uniform}(0,1)$ and the winner achieves $U_{(\kappa V)}=\max\{U_1,\ldots,U_{\kappa V}\}\sim\mathrm{Beta}(\kappa V,1)$. Intuitively, $\kappa$ increases with the entropy. Therefore, we can identify an effective pool size $\kappa V\simeq 1/p$. Equivalently, $\kappa$ can be characterized via the entropy of the next-token
distribution. Let $p_t(\cdot)=p_{\mathrm{LM}}(\cdot \mid x_{<t})$ denote the next-token probability distribution at step $t$, and define the entropy as
\begin{equation}
\begin{aligned}\label{entropy_eq}
    \mathcal{H}(p_t) \;\triangleq\; -\sum_{i=1}^{V} p_t(i)\log p_t(i).
\end{aligned}
\end{equation}

The effective number of competing tokens is then given by $\exp\big(\mathcal{H}(p_t)\big)$, which corresponds to the size of a uniform distribution with the same
uncertainty. Accordingly, we have $\kappa \approx \exp\big(\mathcal{H}(p_t)\big)/V$. Furthermore, we set $\frac{1}{p_t}=\kappa V=\exp\big(\mathcal{H}(p_t)\big)$ as estimated in a development set. For large $V$, using the expansions in Lemma~\ref{lem:beta-log-digamma}, we have
\begin{equation}\label{large_v_approx}
    \begin{aligned}
        &H_{1/p_t}
        = \ln(\exp\big(\mathcal{H}(p_t)\big)) + \gamma
        = \mathcal{H}(p_t) + \gamma, \\
        &\psi_1(1)-\psi_1\!\big(1+\exp\big(\mathcal{H}(p_t)\big)\big)
        = \frac{\pi^2}{6},
    \end{aligned}
\end{equation}
where $\gamma$ is Euler's constant defined in
Lemma~\ref{lem:beta-log-digamma}.

Therefore, for the true message $M^\star$,
\begin{equation}
\begin{aligned}\label{H1_distribution}
   \mathbb{E}[C_{M^\star}\mid\mathcal{H}_1]
   &= \mu_{\mathcal{H}_1}
   = \frac{1}{T}\sum_{t=1}^T H_{1/p_t}= \frac{1}{T}\sum_{t=1}^T \mathcal{H}(p_t)+\gamma, \\
   \mathrm{Var}[C_{M^\star}\mid\mathcal{H}_1]
   &= \sigma_{\mathcal{H}_1}^2
   = \frac{\pi^2}{6T^2}.
\end{aligned}
\end{equation}

Let $Z=\max_{M\in\{0, \dots, 2^m-1\}}\{C_{M}\}$. Since the sequence-level score $C_{M}(W,\mathsf{sk})$ averages over
$T$ tokens,  the Central Limit Theorem (CLT) suggests that, as 
$T$ grows, $C_{M}(W,\mathsf{sk})\sim \mathcal{N}(\mu_{\mathcal{H}_0}, \sigma^2_{\mathcal{H}_0})$. Besides, although the statistics $\{C_M\}$ are not strictly independent since they are calculated on the same text, 
each $C_M$ is an average of $T$ per-token scores with variance $O(1/T)$. 
As $T$ grows, the variance of each $C_M$ shrinks. Therefore, the event $\{Z>\tau\}$ is potentially caused by one candidate $C_M$ exhibiting an unusually large deviation, 
rather than by simultaneous moderate deviations of many correlated $C_M$. Although $\{C_M\}$ for different candidate messages $M$ are not exactly independent, by Lemma~\ref{lem:maxK}, we obtain
\begin{equation}\label{fpr}
    \begin{aligned}
        \mathrm{FPR}(\tau)
= \Pr\big(Z>\tau\mid \mathcal{H}_{0}\big)
&= 1-\Big[\Phi\big(\tfrac{\tau-\mu_{\mathcal{H}_{0}}}{\sigma_{\mathcal{H}_{0}}}\big)\Big]^{2^m}\approx 2^m Q\!\big(\tfrac{\tau-\mu_{\mathcal{H}_{0}}}{\sigma_{\mathcal{H}_{0}}}\big),
    \end{aligned}
\end{equation}
where as defined in Lemma~\ref{lem:maxK}, $\Phi(\cdot)$ denotes the cumulative distribution function of the standard normal distribution while $Q(\cdot)$ is the gaussian tail probability. Here, we justify the accuracy of the approximation in \eqref{fpr}, where $\Pr\big(Z>\tau\mid \mathcal{H}_{0}\big)$ is empirically derived by collecting a bunch of $Z$, while $\mu_{H_0}$ and $\sigma_{H_0}$ is calculated as in~\eqref{eq:gumbel-H0-moments}. As shown in Fig.~\ref{fig:analysis_gumbel1}, we observe the all three curves closely match near the decision region ($\boldsymbol{\tau \approx 1.33}$ with $m=3$ and $T=200$), indicating that \eqref{fpr} provides an accurate approximation in practice despite the independence simplification.

\begin{figure}[htbp]
  \centering
  \captionsetup[subfigure]{justification=centering}

  \begin{subfigure}[t]{0.49\linewidth}
    \centering
    \includegraphics[width=\linewidth]{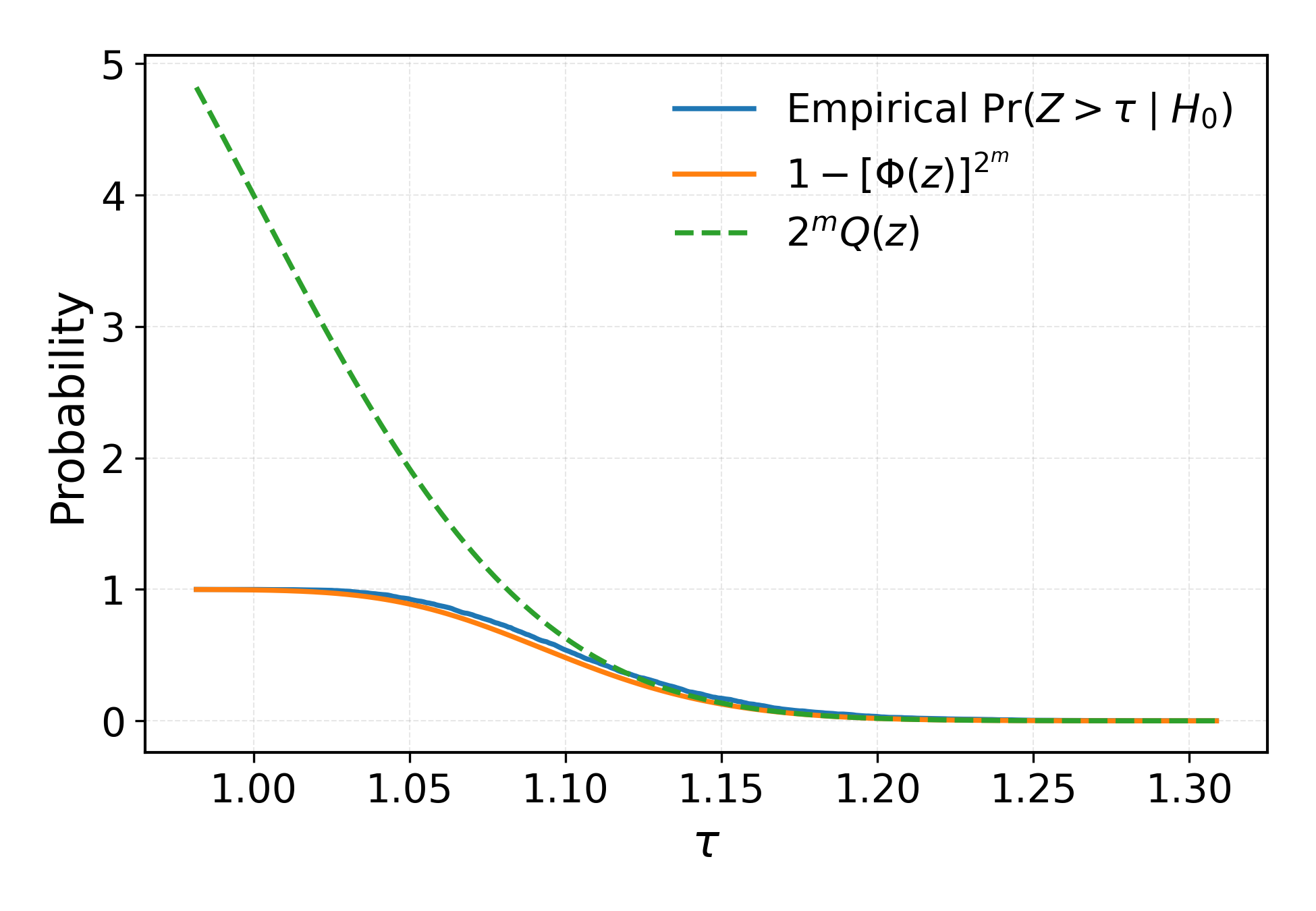}
    \caption{$\Pr(Z>\tau \mid \mathcal{H}_0)$ approximation}
    \label{fig:analysis_gumbel1}
  \end{subfigure}
  \hfill
  \begin{subfigure}[t]{0.49\linewidth}
    \centering
    \includegraphics[width=\linewidth]{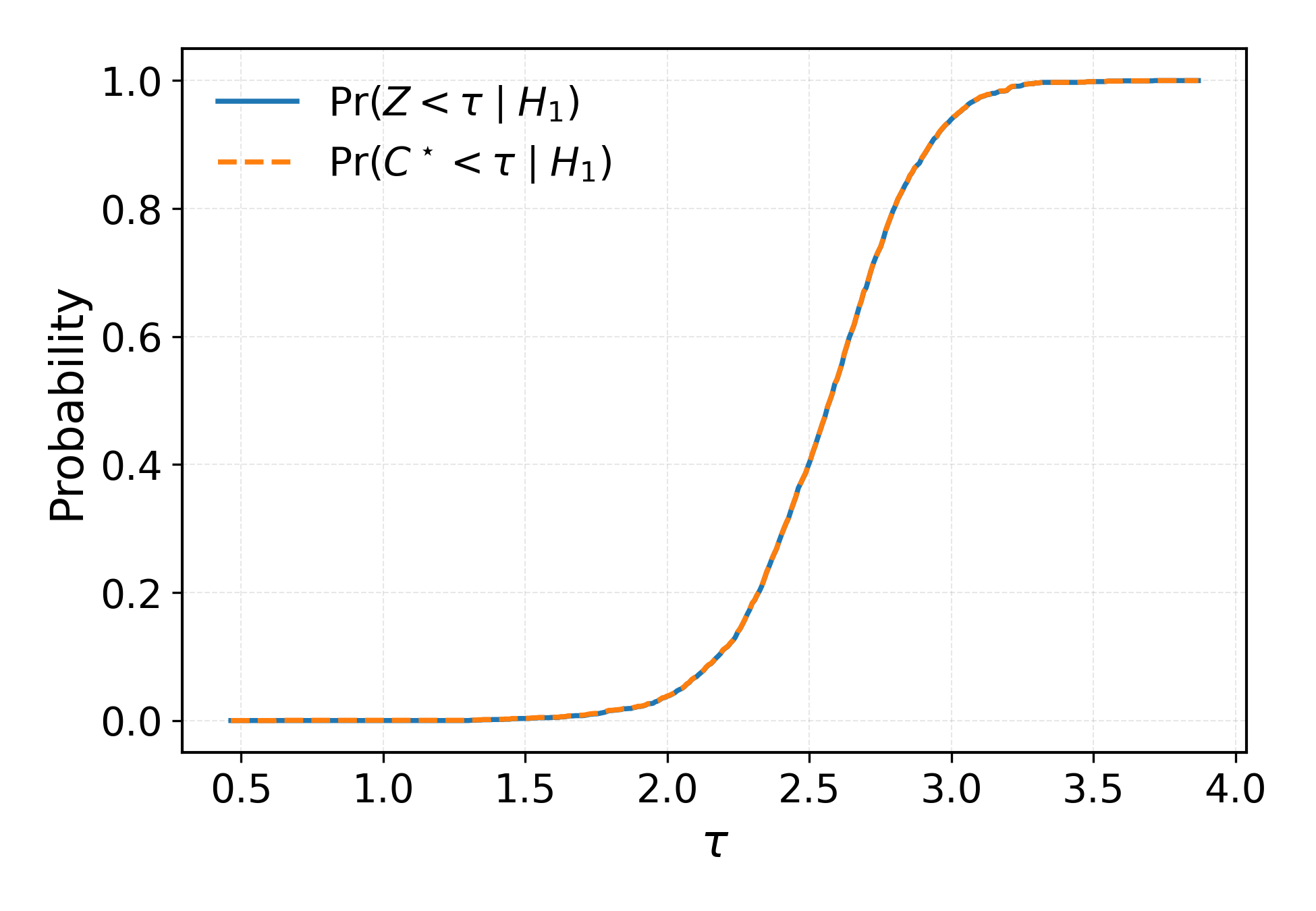}
    \caption{$\Pr(Z<\tau \mid \mathcal{H}_1)$ approximation}
    \label{fig:analysis_gumbel2}
  \end{subfigure}

  \caption{Validation of Gaussian approximations under $\mathcal{H}_0$ and $\mathcal{H}_1$.}
  \label{fig:analysis_gumbel}
\end{figure}

Similarly, under $\mathcal{H}_{1}$, we can approximate $C_{M^{\star}}(W,\mathsf{sk})\sim \mathcal{N}(\mu_{\mathcal{H}_1}, \sigma^2_{\mathcal{H}_1})$ and calculate FNR as

\begin{equation}\label{fnr}
    \begin{aligned}
        \mathrm{FNR}(\tau)
&=\Pr(Z < \tau \mid H_1) = \Pr(C_{M^*} < \tau \;\text{and}\; C_M < \tau,\ \forall M \neq M^* \mid H_1)\approx\Pr\big(C_{M^{\star}}<\tau\mid \mathcal{H}_{1}\big)\\&=\Phi\big(\frac{\tau-\mu_{\mathcal{H}_1}}{\sigma_{\mathcal{H}_1}}\big)=Q\big(\frac{\mu_{\mathcal{H}_1}-\tau}{\sigma_{\mathcal{H}_1}}\big).
    \end{aligned}
\end{equation}

Similarly, to justify the approximation of $\Pr\big(Z<\tau\mid \mathcal{H}_{1}\big)\approx\Pr\big(C_{M^{\star}}<\tau\mid \mathcal{H}_{1}\big)$, we use 2500 watermarked samples ($m = 3$ and $T=200$), compute all $C_M$, and compare $\Pr(Z < \tau)$ and $\Pr(C_{M^*} < \tau)$. As shown in Fig.~\ref{fig:analysis_gumbel2}, the two curves almost coincide across the full range, including near the decision threshold. This indicates that errors are dominated by $C_{M^*} < \tau$, while failures caused by competing messages are rare.

To solve the EER threshold, let $\mathrm{FPR}(\tau^{\mathrm{eer}})=\mathrm{FNR}(\tau^{\mathrm{eer}})$. Let 
\begin{equation}\label{z_brevity}
    \begin{aligned}
        z_0(\tau^{\mathrm{eer}})&=\frac{\tau^{\mathrm{eer}}-\mu_{\mathcal{H}_0}}{\sigma_{\mathcal{H}_0}},\\ z_1(\tau^{\mathrm{eer}})&=\frac{\mu_{\mathcal{H}_1}-\tau^{\mathrm{eer}}}{\sigma_{\mathcal{H}_1}},
    \end{aligned}
\end{equation}

we write $z_0=z_0(\tau^{\mathrm{eer}})$and $z_1=z_1(\tau^{\mathrm{eer}})$ for brevity, combining~\eqref{fpr} and~\eqref{fnr}, then

\begin{equation}
    \begin{aligned}
        z_1^2 = z_0^2 - 2m\ln 2-2\ln\Big(\frac{z_1}{z_0}\Big).
    \end{aligned}
\end{equation}
Since $z_0$ and $z_1$ are of the same order as the EER operating points, $2\ln\Big(\frac{z_1}{z_0}\Big)$ is lower-order. Thus, we obtain
\begin{equation}\label{eqstar}
    \begin{aligned}
        z_1^2 \approx z_0^2 -2m\ln 2.
    \end{aligned}
\end{equation}

Let $\Delta\mu=\mu_{\mathcal{H}_1}-\mu_{\mathcal{H}_0}$, and thus,
\begin{equation}
    \begin{aligned}
        \sigma_0 z_0+\sigma_1 z_1=\Delta\mu
    \end{aligned}
\end{equation}

We first take $m=0$. Therefore,

\begin{equation}\label{eer_m0}
    \begin{aligned}
        \tau^{\mathrm{eer}}_{m=0}=\frac{\mu_{\mathcal{H}_0}\sigma_{\mathcal{H}_1}+\mu_{\mathcal{H}_1}\sigma_{\mathcal{H}_0}}{\sigma_{\mathcal{H}_0}+\sigma_{\mathcal{H}_1}}
    \end{aligned}
\end{equation}

    Thus, plug~\eqref{eer_m0} into~\eqref{z_brevity}, at $m=0$,
\begin{equation}\label{forz}
    \begin{aligned}
        z_0=z_1=z_{m=0}=\frac{\Delta\mu}{\sigma_{\mathcal{H}_0}+\sigma_{\mathcal{H}_1}}.
    \end{aligned}
\end{equation}

Now we take a first order perturbation for $m>0$. Let
\begin{equation}
    \begin{aligned}
        z_0=z_{m=0}+\varepsilon_0,\qquad
z_1=z_{m=0}+\varepsilon_1,
    \end{aligned}
    \label{pertubation_z}
\end{equation}

since $z_1^2-z_0^2=(z_{m=0}+\varepsilon_1)^2-(z_{m=0}+\varepsilon_0)^2\approx 2z(\varepsilon_1-\varepsilon_0)$, from~\eqref{eqstar},
\begin{equation}
    \begin{aligned}
        2z(\varepsilon_1-\varepsilon_0)=-2m\ln 2.
    \end{aligned}
\end{equation}
Therefore, 

\begin{equation}\label{eq11}
    \begin{aligned}
        \varepsilon_1-\varepsilon_0=-\frac{m\ln 2}{z_{m=0}}
    \end{aligned}
\end{equation}

Combining the identity $\mu_{\mathcal{H}_0}+\sigma_{\mathcal{H}_0} z_0=\tau^{\mathrm{eer}}=\mu_{\mathcal{H}_1}-\sigma_{\mathcal{H}_1} z_1$, we obtain

\begin{equation}\label{eq22}
    \begin{aligned}
        \sigma_{\mathcal{H}_0}\varepsilon_0+\sigma_{\mathcal{H}_1}\varepsilon_1=0
    \end{aligned}
\end{equation}
Furthermore, combining~\eqref{forz},~\eqref{eq11} and~\eqref{eq22}, we obatin

\begin{equation}
    \begin{aligned}
        \varepsilon_1&=-\frac{\sigma_{\mathcal{H}_0}}{\Delta\mu}\,m\ln 2\\
\varepsilon_0&=\frac{\sigma_{\mathcal{H}_1}}{\Delta\mu}\,m\ln 2.
    \end{aligned}
\end{equation}
Therefore, from~\eqref{pertubation_z}, we obtain
\begin{equation}
    \begin{aligned}
        z_1\approx z_{m=0}-\frac{\sigma_{\mathcal{H}_0}}{\Delta\mu}\,m\ln 2
    \end{aligned}
\end{equation}
Substituting~\eqref{forz} gives the EER approximation
\begin{equation}\label{eergeneral}
    \begin{aligned}
        \mathrm{EER}_{\text{Gumbel}}\approx
Q\left(
\frac{\Delta\mu}{\sigma_{\mathcal{H}_0}+\sigma_{\mathcal{H}_1}}
-\frac{\sigma_{\mathcal{H}_0}}{\Delta\mu}\,m\ln 2
\right).
    \end{aligned}
\end{equation}

To obtain a closed-form expression and visualization of dependence of EER on $H$, $m$, and $T$, we collapse token-level heterogeneous statistics into a single global entropy averaged over all tokens. Denote $\mathcal{H}=\frac{1}{T}\sum_{t=1}^{T}\mathcal{H}(p_t)$, we obtain 
\begin{equation}
\label{eq:EER-gumbel-N-asym}
\mathrm{EER}_{\mathrm{Gumbel}}
\approx Q\big(z_{\mathcal H}\big),
\end{equation}
where
\begin{equation}\label{eq:z_V}
    \begin{aligned}
        z_{\mathcal H}
        &=
        \frac{(\mathcal{H}+\gamma-1)\sqrt{T}}{1+\pi/\sqrt{6}}
        -
        \frac{m\ln 2}{(\mathcal{H}+\gamma-1)\sqrt{T}} .
    \end{aligned}
\end{equation}

For large
$z_{\mathcal H}$, we approximate
\begin{equation}\label{log_eer_gumbel}
\begin{aligned}
\log \mathrm{EER}_{\mathrm{Gumbel}}
&=
-\frac{z_{\mathcal H}^2}{2}
-\log\!\big(z_{\mathcal H}\sqrt{2\pi}\big) \\
&=
- c_1\, T \big(\mathcal{H}+\gamma-1\big)^2 + c_2\, m
-\frac{(m\ln 2)^2}{2T\big(\mathcal{H}+\gamma-1\big)^2}
-\log\!\big(z_{\mathcal H}\sqrt{2\pi}\big),
\end{aligned}
\end{equation}
where the constants $c_1=\frac{1}{2\big(1+\pi/\sqrt{6}\big)^2}$ and $c_2=\frac{\ln 2}{1+\pi/\sqrt{6}}$.

The dominant term in~\eqref{log_eer_gumbel} scales quadratically with
$\mathcal{H}$, while the dependence on the symbol size $m$ appears as a linear
correction in the exponent. The remaining terms are strictly lower order in
$\mathcal{H}$. Consequently, 
\begin{equation}
\begin{aligned}
    \log \mathrm{EER}_{\mathrm{Gumbel}}
=
- c_1\, T \mathcal{H}^2
+ c_2\, m
+ o\!\big(T\mathcal{H}^2\big),
\end{aligned}
\end{equation}
which implies an exponential decay of $\mathrm{EER}_{\mathrm{Gumbel}}$ at a
quadratic rate in $\mathcal{H}$. Furthermore, increasing the symbol size $m$
leads to a larger $\mathrm{EER}$ through a linear shift in the exponent.

\subsection{Tournament sampling based multi-bit watermarking}\label{tour_eer}

Recall in~\eqref{eq:wmeanscore} the score of $t$-th token for message $M$
\begin{equation}
    \begin{aligned}
        S_{M}(\mathsf{sk},W_t)=\frac{1}{L}\sum_{\ell=1}^{L}\alpha_{\ell}\,\Psi(u_{t, l},\psi(M)),
    \end{aligned}
\end{equation}
and the sequence-level statistic for message $M$ as the per-token average
\begin{equation}
\label{eq:Cmax-def}
C_{M}(\mathsf{sk},W)=\frac{1}{T}\sum_{t=1}^{T} S_{M}(\mathsf{sk},W_t).
\end{equation}

In this derivation, we treat $m=1$, where the construction enforces $S_0(\mathsf{sk}, W_t)+S_1(\mathsf{sk}, W_t)\equiv 1$ per token from the property of mirroring. Define
\begin{equation}
    \begin{aligned}
        Z_t&=\max\{S_0(\mathsf{sk}, W_t),S_1(\mathsf{sk}, W_t)\}\\&=\frac{1}{2}+\left|S_0(\mathsf{sk}, W_t)-\frac{1}{2}\right|,
    \end{aligned}
\end{equation}
and detect with $C_{\max}=\frac{1}{T}\sum_{t=1}^T Z_t$. For easier analysis, we assume $\frac{1}{L}\sum_{\ell=1}^L\alpha_\ell=1$.


Under null hypothesis, at each layer $\ell$, $\Psi\sim\mathrm{Uniform}(0,1)$, hence
\begin{equation}
    \begin{aligned}
        \mathbb{E}\big[S_0(\mathsf{sk}, W_t)\mid\mathcal H_0\big]&=\frac{1}{L}\sum_{\ell=1}^L\frac{\alpha_\ell}{2}
=\frac{1}{2},\\
\mathrm{Var}\big[S_0(\mathsf{sk}, W_t)\mid\mathcal H_0\big]
&=\frac{1}{L^2}\sum_{\ell=1}^L\alpha_\ell^2\,\mathrm{Var}(\Psi)
=\frac{A}{12\,L^2},
    \end{aligned}
\end{equation}
where $A=\sum_{\ell=1}^L \alpha_\ell^2$. Approximating $S_0(\mathsf{sk}, W_t)-\tfrac12$ by $\mathcal N(0,\,\frac{A}{12\,L^2})$ and using the
Lemma~\ref{lem:folded-normal}, 
\begin{equation} \label{eq:tour-H0-fold}
\begin{aligned}
    \mathbb{E}\left|S_0(\mathsf{sk}, W_t)-\tfrac12\right|
&= \sqrt{\frac{A}{6\pi\,L^2}},\\
\mathrm{Var}\left|S_0(\mathsf{sk}, W_t)-\tfrac12\right|
&= \frac{A}{12L^2}\left(1-\frac{2}{\pi}\right).
\end{aligned}
\end{equation}
Therefore
\begin{equation}
\begin{aligned}
    \mathbb{E}[Z_t\mid\mathcal H_0]
&=\frac12+\sqrt{\frac{A}{6\pi\,L^2}},\\
\mathrm{Var}[Z_t\mid\mathcal H_0]
&=\frac{A}{12L^2}\!\left(1-\frac{2}{\pi}\right),
\end{aligned}
\end{equation}

By central limit theorem (CLT), for $C_{\max}=\frac{1}{T}\sum_t Z_t$,
\begin{equation}
\label{eq:tour-H0}
\begin{aligned}
    \mu_{\mathcal H_0}&=\mathbb{E}[C_{\max}\mid\mathcal H_0]= \frac12+\sqrt{\frac{A}{6\pi\,L^2}},\\
\sigma_{\mathcal H_0}^2&=\mathrm{Var}[C_{\max}\mid\mathcal H_0]
= \frac{A}{12L^2T}\!\left(1-\frac{2}{\pi}\right).
\end{aligned}
\end{equation}

We can derive the FPR as
\begin{equation}\label{tournament_fpr}
    \begin{aligned}
    \mathrm{FPR}=\Pr[C_{max}>\tau|H_{0}]=Q(\frac{\tau-\mu_{\mathcal H_0}}{\sigma_{\mathcal H_0}})
\end{aligned}
\end{equation}

On the other hand, under the alternative hypothesis $\mathcal H_1$, at layer $\ell$, refer to Corollary 28 in SynthID, the mirrored random variable described by the cumulative density function (CDF) and probability density function (PDF) as follows,
\begin{equation}
    \begin{aligned}
        F_{\Psi_\ell}(x) &= C_{wm}^\ell x+(1-C_{wm}^\ell)x^2,\\
f_{\Psi_\ell}(x) &= C_{wm}^\ell+2(1-C_{wm}^\ell)x,
    \end{aligned}
\end{equation}
where $C_{wm}^\ell\in[0,1)$ represents the collision probability at layer $\ell$ as defined in Definition 22 in SynthID, which is the probability that two samples drawn i.i.d. from the probability distribution of tokens at layer $l$ are the same. Hence,
\begin{equation}
    \begin{aligned}
        \mathbb E[\Psi_\ell]&=\frac{2}{3}-\frac{C_{wm}^\ell}{6},\\
\mathrm{Var}(\Psi_\ell)&=\frac{2+2C_{wm}^\ell-(C_{wm}^\ell)^2}{36}.
    \end{aligned}
\end{equation}
Hence the per-token $S_0(t)=\frac{1}{L}\sum_\ell \alpha_\ell\Psi_\ell$ has
\begin{equation}
\begin{aligned}\label{mu_S_and_v_S}
    \mu_S&=\mathbb{E}\big[S_0(\mathsf{sk}, W_t)\mid\mathcal H_1\big]
=\frac{1}{L}\sum_{\ell=1}^L\alpha_\ell\left(\frac{2}{3}-\frac{C_{wm}^\ell}{6}\right),\\
v_S&=\mathrm{Var}\big[S_0(\mathsf{sk}, W_t)\mid\mathcal H_1\big]=\frac{1}{L^2}\sum_{\ell=1}^L\alpha_{\ell}^2\,\frac{2+2C_{wm}^\ell-(C_{wm}^\ell)^2}{36}.
\end{aligned}
\end{equation}
Let $\mu_\Delta=\mu_S-\tfrac{1}{2}$. Using the Lemma~\ref{lem:folded-normal} again,
\begin{equation}
\label{eq:tour-H1-fold}
\begin{aligned}
    &\mathbb{E}|S_0(\mathsf{sk}, W_t)-\tfrac12|= \sqrt{\frac{2}{\pi}}\,\sqrt{v_S}\;\exp\!\left(-\frac{\mu_\Delta^2}{2v_S}\right)
+ \mu_\Delta\left[1-2\Phi\!\left(-\frac{\mu_\Delta}{\sqrt{v_S}}\right)\right]=: \mathrm{FNmean}(\mu_\Delta,v_S),\\
&\mathrm{Var}(|S_0(\mathsf{sk}, W_t)-\tfrac12|)= \mu_\Delta^2 + v_S - \big(\mathbb{E}|S_0(\mathsf{sk}, W_t)-\tfrac12|\big)^2=: \mathrm{FNvar}(\mu_\Delta,v_S).
\end{aligned}
\end{equation}
Thus for $Z_t=\tfrac12+|S_0(\mathsf{sk}, W_t)-\tfrac12|$,

\begin{equation}
    \begin{aligned}
        \mathbb{E}[Z_{t}|\mathcal{H}_{1}]
&=\frac12+\mathrm{FNmean}(\mu_\Delta,v_S),\\
\textsf{Var}[Z_{t}|\mathcal{H}_{1}]
&=\mathrm{FNvar}(\mu_\Delta,v_S),
    \end{aligned}
\end{equation}
and
\begin{equation}
\label{eq:tour-H1}
\begin{aligned}
    \mu_{\mathcal H_1}&=\mathbb{E}[C_{\max}\mid\mathcal H_1]
= \mathbb{E}[Z_{t}|\mathcal{H}_{1}],\\
\sigma_{\mathcal H_1}^2&=\mathrm{Var}[C_{\max}\mid\mathcal H_1]
= \textsf{Var}[Z_{t}|\mathcal{H}_{1}]/T.
\end{aligned}
\end{equation}

We can derive the FNR as

\begin{equation}\label{tournament_fnr}
    \begin{aligned}
    \mathrm{FNR}=\Pr[C_{max}<\tau|H_{1}]=\Phi(\frac{\tau-\mu_{\mathcal H_1}}{\sigma_{\mathcal H_1}})
\end{aligned}
\end{equation}

Combining~\eqref{tournament_fpr} and~\eqref{tournament_fnr}, let $\mathrm{FPR}=\mathrm{FNR}$, we can derive the $\mathrm{EER}$ is
\begin{equation}
\label{eq:EER}
\begin{aligned}
    \mathrm{EER}_{\text{tour}}
&=\mathrm{FPR}(\tau^{\mathrm{eer}})
=\mathrm{FNR}(\tau^{\mathrm{eer}})=Q\left(\frac{\mu_{\mathcal H_1}-\mu_{\mathcal H_0}}
{\sigma_{\mathcal H_0}+\sigma_{\mathcal H_1}}\right).
\end{aligned}
\end{equation}

For easier analysis, we assume $\frac{1}{L}\sum_{\ell=1}^L\alpha_\ell=1$, define
\begin{equation}\label{c1_c2}
    \begin{aligned}
        C_1 &:=\frac{1}{L}\sum_{\ell=1}^L\alpha_\ell C_{wm}^\ell,
    \\C_2 &:= \frac{1}{L}\sum_{\ell=1}^L\alpha_{\ell}^2\,\frac{2+2C_{wm}^\ell-(C_{wm}^\ell)^2}{36}.
    \end{aligned}
\end{equation}
Therefore,

\begin{equation}
    \begin{aligned}
       \mu_\Delta&=\frac{1-C_1}{6},\\
v_S&=\frac{C_2}{L}.
    \end{aligned}
\end{equation}

Let $z=\frac{\mu_{\Delta}}{\sqrt{v_S}}$, which means $z$ depends on $C_1, C_2$, and $L$. Then the folded-normal mean in~\eqref{eq:tour-H1-fold} satisfies
\begin{equation}
    \begin{aligned}
        \mathrm{FNmean}(\mu_\Delta,v_S)
&=\mu_{\Delta}\Big(\sqrt{\tfrac{2}{\pi}}\frac{e^{-z^2/2}}{z}
+2\Phi(z)-1\Big).
    \end{aligned}
\end{equation}

Let $m(z) = \sqrt{\tfrac{2}{\pi}}\frac{e^{-z^2/2}}{z}
+2\Phi(z)-1$. Therefore, $\mathrm{FNmean}(\mu_\Delta,v_S)
=\mu_{\Delta}m(z)$. Hence,
\begin{equation}
    \begin{aligned}
        \mu_{\mathcal H_1}=\frac{1}{2}+\mu_{\Delta}m(z).
    \end{aligned}
\end{equation}

Meanwhile,
\begin{equation}
    \begin{aligned}
        \sigma_{\mathcal H_1}&=\sqrt{\frac{\mu_\Delta^2 + v_S - \big(\mathrm{FNmean}(\mu_\Delta,v_S)\big)^2}{T}}\\&=\frac{1}{\sqrt{T}}\sqrt{\mu^2_{\Delta}\Big(1-m^2(z)+\frac{1}{z^2}\Big)}.
    \end{aligned}
\end{equation}

Let

\begin{equation}\label{kappa0}
    \kappa_0=\sqrt{\frac{A}{6\pi L^2}},
\end{equation}

\begin{equation}\label{kappa1}
    \begin{aligned}
       \kappa_1=\sqrt{\frac{A(1-\tfrac{2}{\pi})}{12L^2}}, 
    \end{aligned}
\end{equation}

\begin{equation}\label{kappa2}
    \kappa_2=\mu_{\Delta}m(z),
\end{equation}

and

\begin{equation}\label{kappa3}
    \kappa_3=\mu_{\Delta}\sqrt{1-m^2(z)+\frac{1}{z^2}},
\end{equation}
then

Hence, by~\eqref{eq:EER},
\begin{equation}
\label{eq:eer-LT-exact}
\begin{aligned}
    \mathrm{EER}_{\text{tour}}
&=
Q\!\left(\frac{\mu_{\mathcal H_1}-\mu_{\mathcal H_0}}
{\sigma_{\mathcal H_0}+\sigma_{\mathcal H_1}}\right)
\\&=
Q\left(
\frac{\kappa_2-\kappa_0}{\kappa_3+\kappa_1}\sqrt{T}
\right).
\end{aligned}
\end{equation}

Define
\begin{equation}
\label{eq:GammaBeta-def}
\begin{aligned}
    \Gamma=\Gamma(C_1, C_2, L)&=\frac{\kappa_2}{\kappa_3+\kappa_1},\\
\beta=\beta(C_1, C_2, L)
&=\frac{\kappa_0}{\kappa_3+\kappa_1},
\end{aligned}
\end{equation}

Then
\begin{equation}
\label{eq:eer-compact}
\mathrm{EER}_{\text{tour}}
=
Q\Big((\Gamma-\beta)\sqrt{T}\Big).
\end{equation}

Using Lemma~\ref{lem:maxK}, we obtain
\begin{equation}
\label{eq:eer-tail}
\begin{aligned}
    \mathrm{EER}_{\text{tour}}
\approx\frac{
\exp\Big(
-\frac{T}{2}\big(\Gamma-\beta\big)^2
\Big)
}{
\sqrt{2\pi T}\big(\Gamma-\beta\big)
}.
\end{aligned}
\end{equation}

Taking the logarithm of~\eqref{eq:eer-tail} yields 

\begin{equation}
\label{eq:logeer-expanded-exact}
\begin{aligned}
\log \mathrm{EER}_{\mathrm{tour}}
&=
-\frac{T}{2}(\Gamma-\beta)^2-\frac{1}{2}\log2\pi T-\log(\Gamma-\beta).
\end{aligned}
\end{equation}

Since~\eqref{c1_c2} shows that both $C_1$ and $C_2$ depend on the
layer-wise collision probabilities $\{C_{\ell,\mathrm{wm}}\}_{\ell=1}^L$,
we define
\begin{equation}
\zeta\!\left(L, C_{1,\mathrm{wm}}, \dots, C_{L,\mathrm{wm}}\right)
\triangleq \Gamma(C_1, C_2, L) - \beta(C_1, C_2, L).
\end{equation}

For notational convenience, we denote the collection of collision-related
parameters by
\begin{equation}
    \begin{aligned}
        \mathbf{c} \triangleq \left(L, C_{1,\mathrm{wm}}, C_{2,\mathrm{wm}}, \dots, C_{L,\mathrm{wm}}\right).
    \end{aligned}
\end{equation}

 Hence,

\begin{equation}
    \begin{aligned}
        \log \mathrm{EER}_{\mathrm{tour}}
&=
-\frac{T}{2}\zeta^2(c)
- \frac{1}{2}\log2\pi T-\log\zeta(c).
    \end{aligned}
\end{equation}

\subsection{Empirical Validation of Token-Level Dependence}
\label{app:token_independence}

\begin{figure}[htbp]
    \centering
    \begin{subfigure}{0.49\linewidth}
        \centering
        \includegraphics[width=\linewidth]{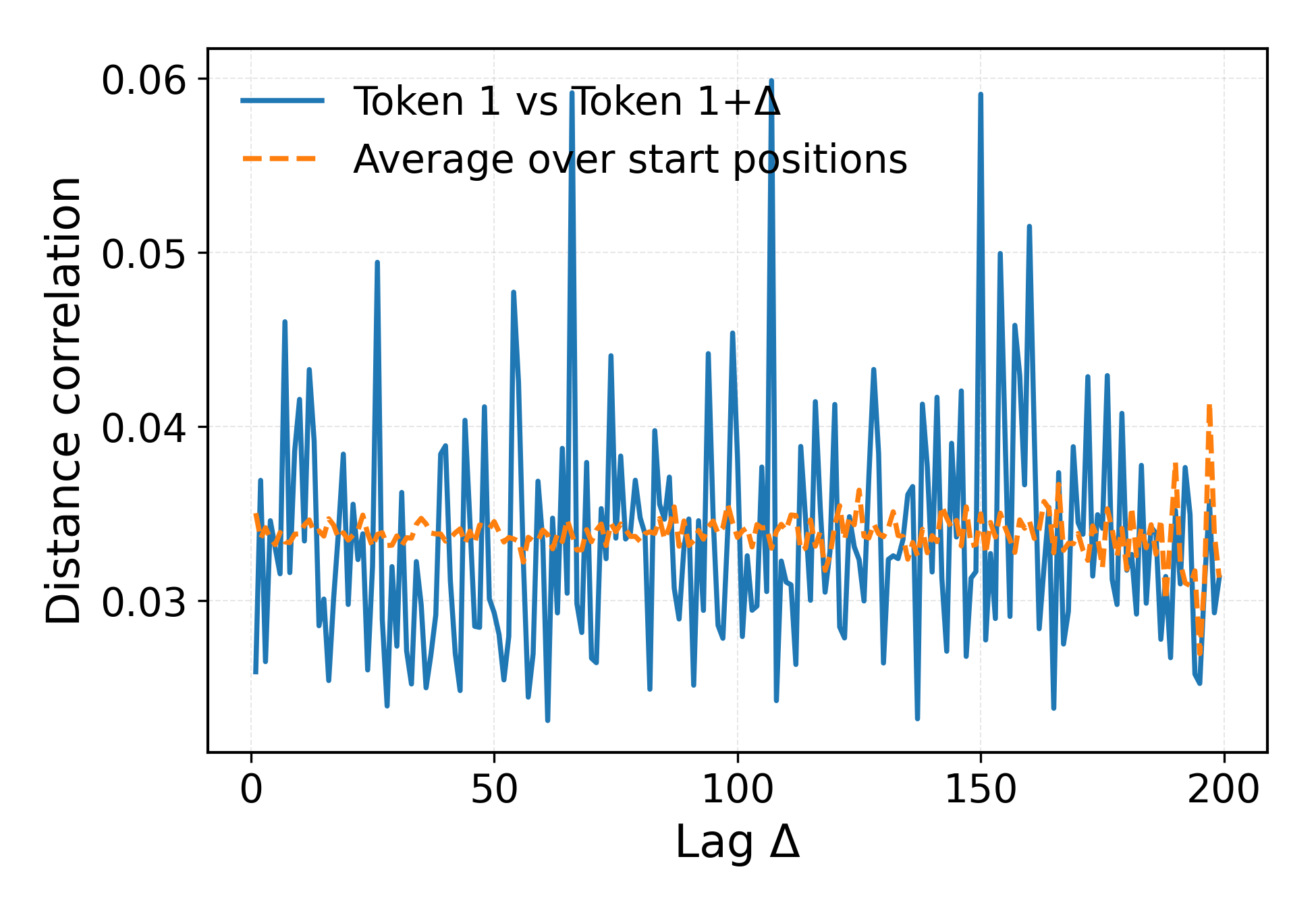}
        \caption{Gumbel-max-based MirrorMark.}
    \end{subfigure}
    \hfill
    \begin{subfigure}{0.49\linewidth}
        \centering
        \includegraphics[width=\linewidth]{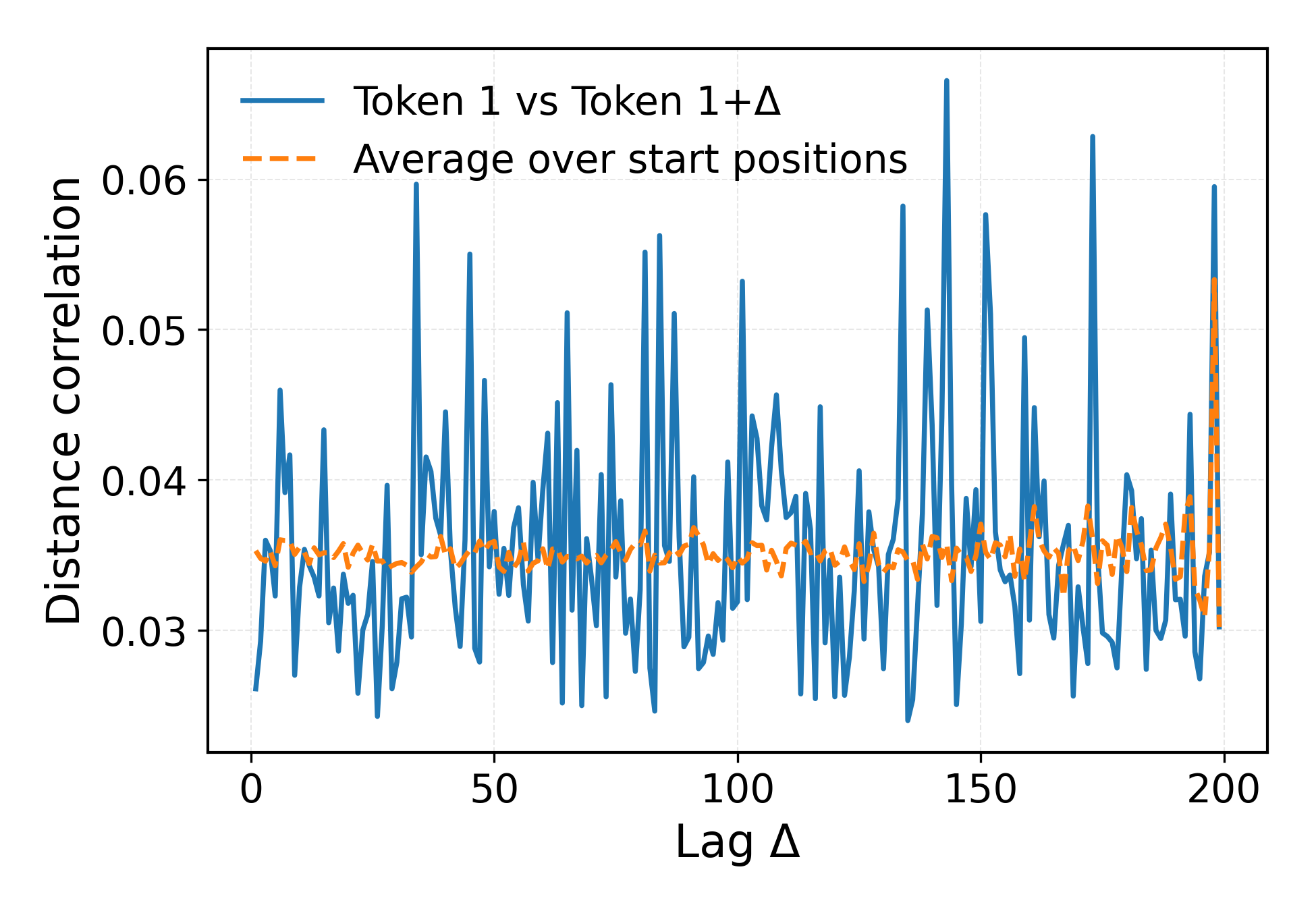}
        \caption{Tournament-based MirrorMark.}
    \end{subfigure}
    \caption{Token-level dependence measured by distance correlation between $S_M(t)$ and $S_M(t+\Delta)$ on 2,500 non-watermarked sequences with $T=200$. The dependence is weak but non-zero, and does not increase with the lag $\Delta$.}
    \label{fig:token_independence}
\end{figure}

The CLT-based EER analysis assumes that token-level score contributions are approximately independent across generation steps. To empirically examine this assumption, we measure the distance correlation between $S_M(t)$ and $S_M(t+\Delta)$ over 2,500 non-watermarked sequences, each with $T=200$ tokens. Fig.~\ref{fig:token_independence} reports the results for both Gumbel-max-based and tournament-based MirrorMark. The observed dependence is weak but non-zero, with distance correlation typically around $0.03$--$0.05$, and it does not increase as the lag $\Delta$ grows. These results suggest that token-level scores exhibit limited short- and long-range dependence in practice. Therefore, while the independence assumption is an approximation rather than an exact property of LLM-generated text, the measured dependence is small, which helps explain why the asymptotic EER estimates remain close to the empirical EERs reported in Table~\ref{theoretical_eer}. 

\section{Lemmas}\label{lemmas}
\begin{lemma}
\label{lem:maxK}
Let $X_1,\dots,X_K\overset{\mathrm{i.i.d.}}{\sim}\mathcal{N}(0,1)$ and $G_K=\max_i X_i$. For large $z$,
\begin{equation}
    \begin{aligned}
        \Pr(G_K>z) = 1-\bigl(1-Q(z)\bigr)^K = K\,Q(z)\,\bigl(1+o(1)\bigr),
    \end{aligned}
\end{equation}
where $\Phi(z)$ denotes the cumulative distribution function of the standard normal distribution\footnote{\url{https://en.wikipedia.org/wiki/Normal_distribution}}
and $Q(z) = 1 - \Phi(z)$ is its Gaussian tail probability.
\end{lemma}

\begin{lemma}
\label{lem:folded-normal}
Let $X\sim\mathcal N(\mu,\sigma^2)$ and $Y=|X|$. Then
\begin{equation}
    \begin{aligned}
        \mathbb E[Y]
= \sigma\sqrt{\tfrac{2}{\pi}}\,\exp\!\left(-\tfrac{\mu^2}{2\sigma^2}\right)
+ \mu\left[1-2\Phi\!\left(-\tfrac{\mu}{\sigma}\right)\right],
    \end{aligned}
\end{equation}
and
\begin{equation}
    \begin{aligned}
        \mathrm{Var}(Y)
= \mu^2+\sigma^2 - \bigl(\mathbb E[Y]\bigr)^2.
    \end{aligned}
\end{equation}
\end{lemma}

\begin{lemma}
\label{lem:beta-log-digamma}
Let $T\sim\mathrm{Beta}(a,b)$. Then
\begin{equation}
    \begin{aligned}
        \mathbb{E}[\ln T] &= \psi_{0}(a)-\psi_{0}(a+b),
    \\
    \mathrm{Var}(\ln T) &= \psi_1(a)-\psi_1(a+b),
    \end{aligned}
\end{equation}
where the digamma function $\psi_{0}(x)$ for $x>0$ is defined as
\begin{equation}
    \begin{aligned}
        \psi_{0}(x) \;=\; -\gamma + \sum_{n=0}^{\infty} \Big(\frac{1}{n+1} - \frac{1}{n+x}\Big),
    \end{aligned}
\end{equation}
    with $\gamma$ the Euler's constant.
The trigamma function $\psi_1(x)$ for $x>0$ is defined as 
\begin{equation}
    \begin{aligned}
        \psi_1(x) = \sum_{n=0}^{\infty} \frac{1}{(n+x)^2}.
    \end{aligned}
\end{equation}
In particular, for $x>0$, let the generalized harmonic number $H_x=\psi_{0}(x+1)-\psi_{0}(1)$. As $x\to\infty$,
\begin{equation}
    H_x = \ln x + \gamma.
\end{equation}
\end{lemma}

\section{Experimental Setup}\label{setup}

Unless otherwise specified, all experiments use the Llama-2-7B model~\citep{touvron2023llama} on a text completion task. We construct prompts from the RealNewsLike subset of C4~\citep{raffel2020exploring}. We randomly select 500 documents, truncate each document to obtain a prefix, and ask the model to generate a continuation conditioned on that prefix. Most results in the main paper are reported on this setting. To assess the generality of MirrorMark beyond this model and task, we additionally evaluate on the Gemma-7B-it~\citep{gemmateam2024gemmaopenmodelsbased} model on an instruction-following task. We randomly sample 500 prompts from the ELI5 dataset~\citep{fan-etal-2019-eli5}, treat them as user instructions, and generate model responses. We report AUC and TPR@1\%FPR for detection, bit accuracy for decoding, and perplexity, GPT-4o score, and repetition rate for text quality. 

Following SynthID, we use top-100 sampling with temperature of $1.0$ for all evaluated watermarking approaches. For CABS, we use the same hyperparameters throughout the experiments, where $h=4$, $f=3$, $W=4$, and $\texttt{max\_len} = \texttt{max\_factor} \cdot H$ with $\texttt{max\_factor}=1.5$ and $H$ denoting the number of positions in the context. Following SynthID, our experiments use a default of $L=30$ tournament layers. For each combination of $m$, $L$, and base model used in tournament-sampling–based MirrorMark, we train a separate Bayesian detector using 10{,}000 watermarked samples and 10{,}000 non-watermarked samples. We randomly split the watermarked and non-watermarked feature files into an 80\% training set and a 20\% validation set. The detector is trained with the Adam optimizer using a learning rate of $3\times 10^{-3}$, a batch size of 64, and up to 100 epochs. This training is relatively lightweight, i.e., approximately 2 hours on a single A100. We select the model that achieves the highest validation TPR at 1\% FPR, and report its performance in the main paper. We evaluate perplexity using the same model that generates the text. Specifically, text generated by LLaMA2-7B is evaluated using LLaMA2-7B, and text generated by Gemma-7B-it is evaluated using Gemma-7B-it.

For all baseline comparisons, we follow the default symbol sizes $m$ specified in the original papers, as these settings are reported to yield their best performance. In particular, MPAC uses $m=2$, StealthInk uses $m=1$, and RSBH uses $m=6$. Therefore, to embed $b$ bits, $H=\frac{b}{m}$ positions are needed.

\section{Additional Results}~\label{extra_experiments}

\subsection{Controlled Comparison under Permutation-Based Reweighting}
\label{app:dip_control}

We further evaluate whether mod-1 mirroring can be applied beyond distortion-free sampling. In permutation-based unbiased reweighting, the watermark randomness is a context-seeded vocabulary permutation rather than a final sampling value. We therefore apply mirroring to permutation ranks. Specifically, for each token, we normalize its position in the permutation to a value in $[0,1)$. For a candidate symbol $M$, we mirror these normalized ranks using the same mod-1 mirroring rule and obtain a symbol-specific effective permutation. The right half of this effective permutation is treated as the green region and receives larger probability under the DiPmark reweighting rule. During detection, we reconstruct the permutation for each context, apply the candidate-symbol mirroring, and decode the symbol with the largest green-token count.

This construction, which we call MirrorDip, uses the same permutation-based reweighting mechanism as the base unbiased watermark but replaces interval-based message assignment with mirroring-based rank transformation. We compare it with StealthInk, which also extends DiPmark to multi-bit watermarking but assigns each message to a contiguous interval in the context-seeded permutation and suppresses that interval during generation. This controlled comparison isolates the effect of the mapping rule under the same class of permutation-based reweighting methods.

\begin{figure*}[htbp]
  \centering
  \begin{minipage}[t]{0.33\textwidth}
    \centering
    \includegraphics[width=\linewidth]{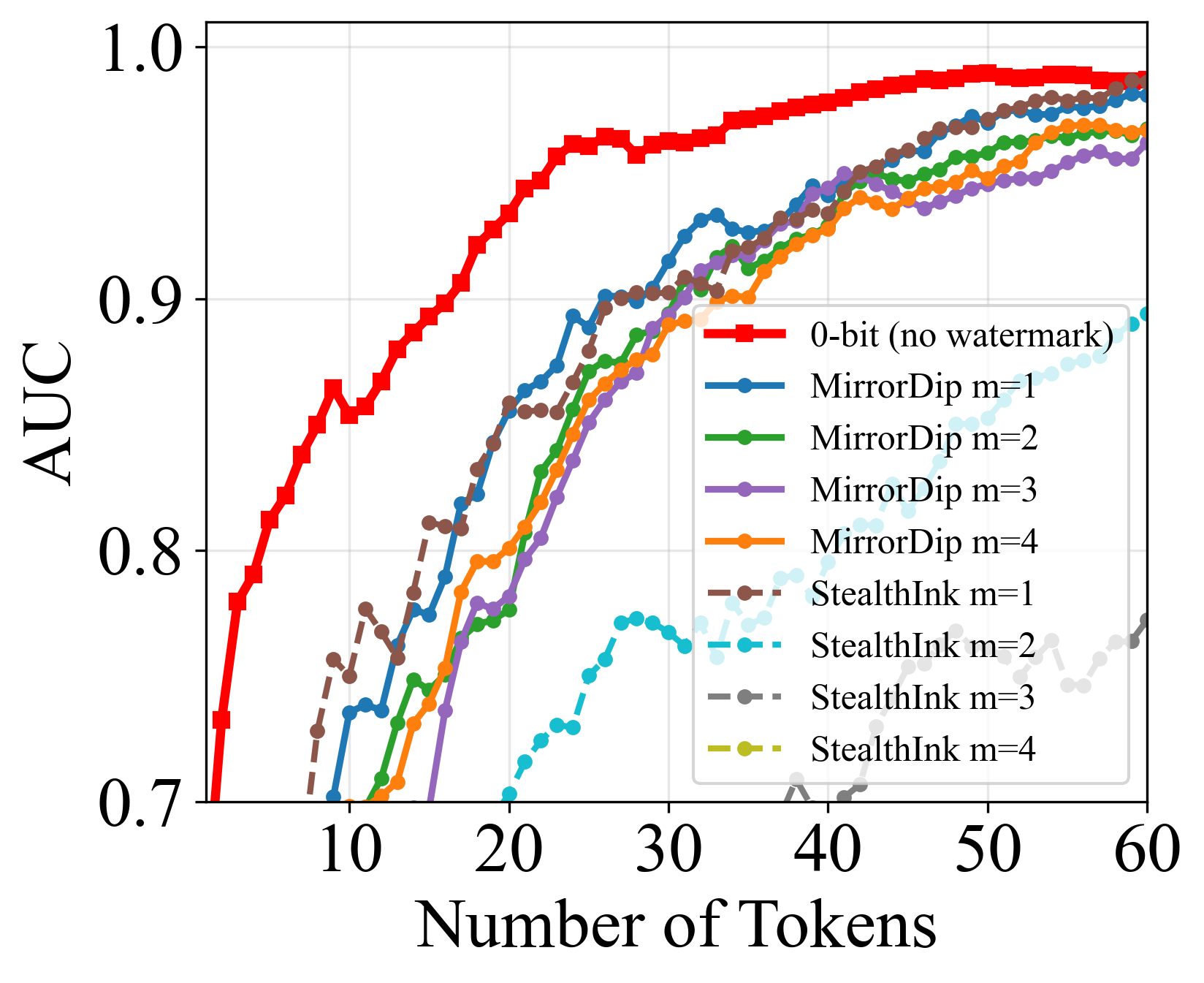}
  \end{minipage}\hfill
  \begin{minipage}[t]{0.33\textwidth}
    \centering
    \includegraphics[width=\linewidth]{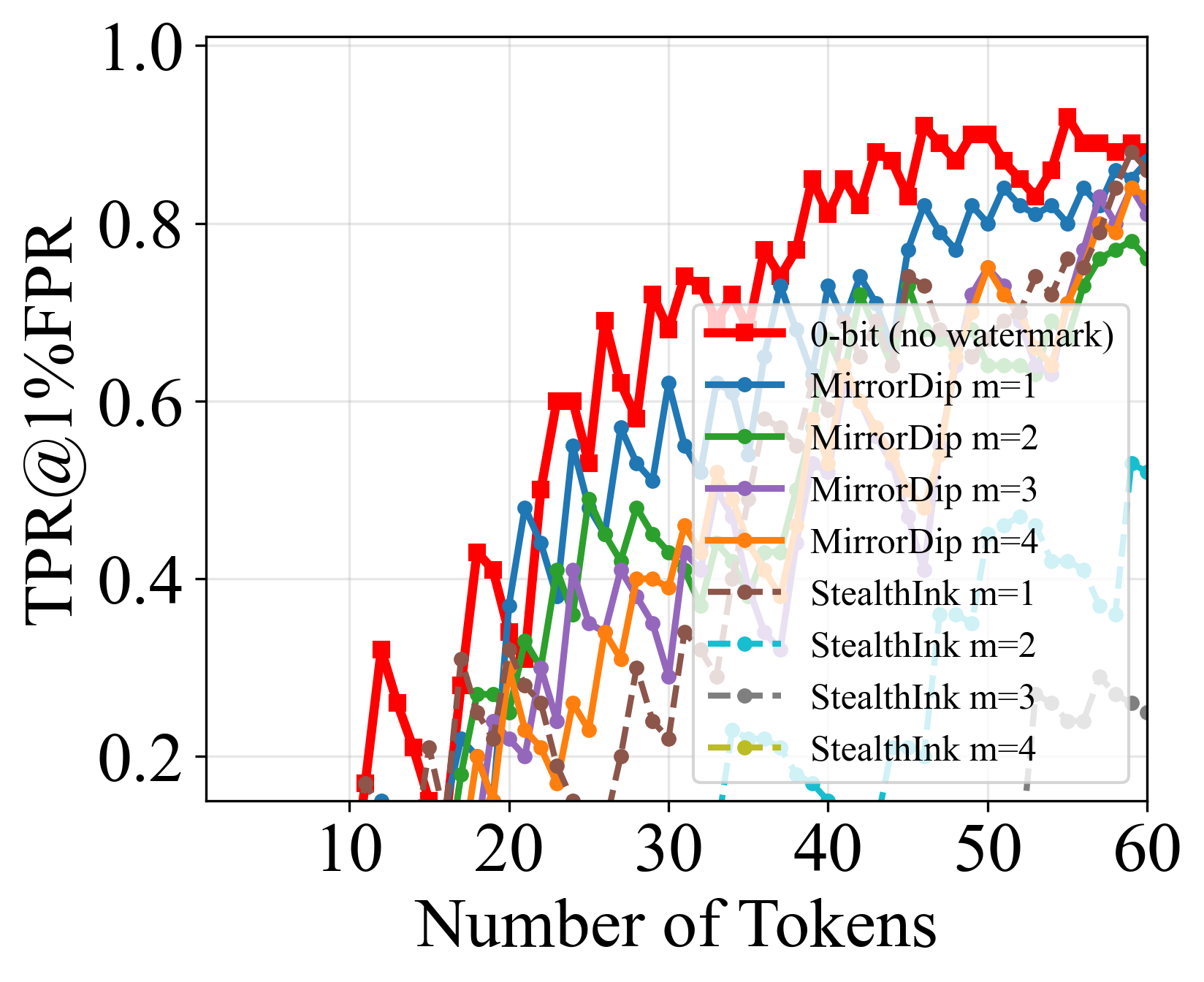}
  \end{minipage}\hfill
  \begin{minipage}[t]{0.33\textwidth}
    \centering
    \includegraphics[width=\linewidth]{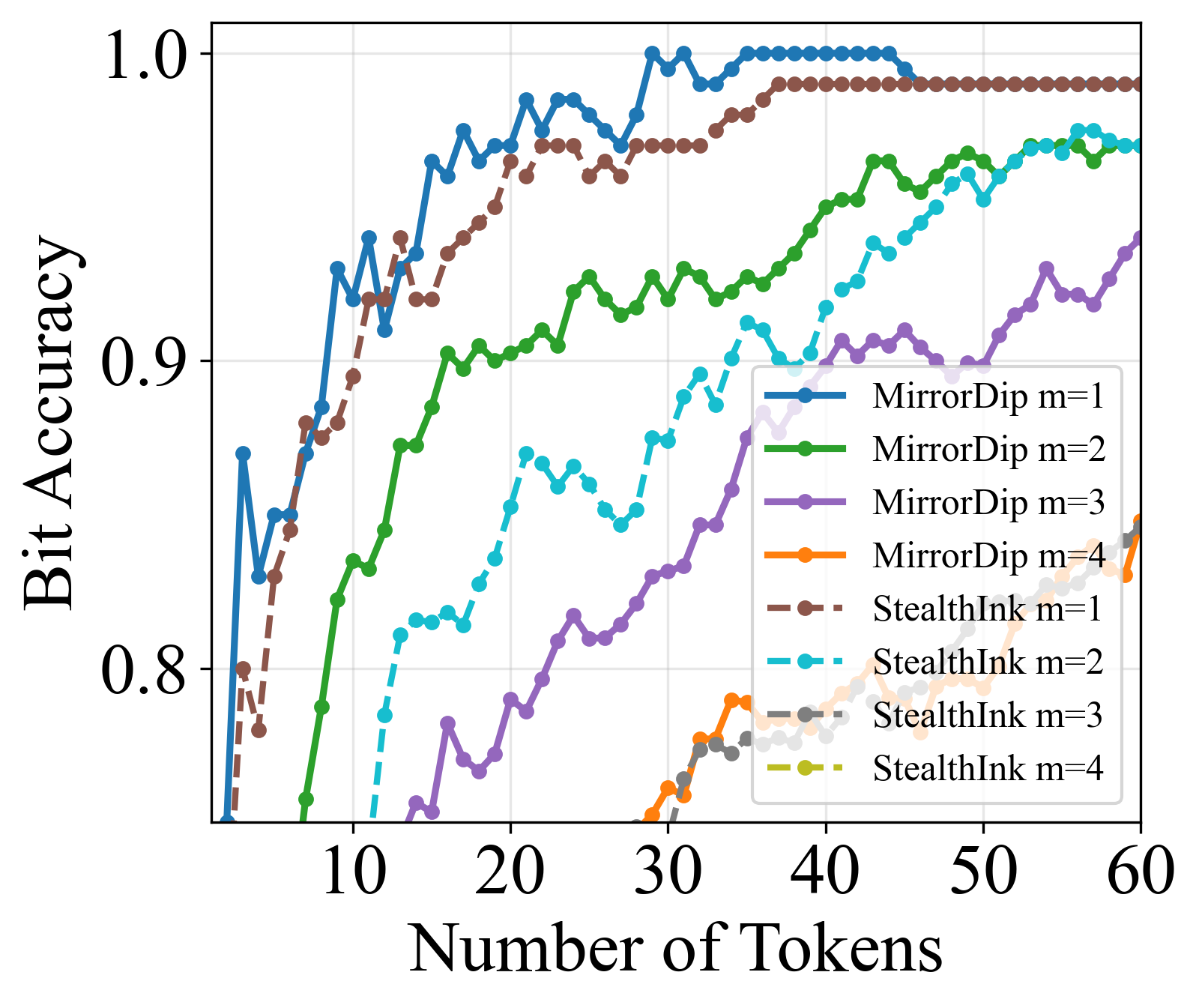}
  \end{minipage}
  \caption{Controlled comparison under permutation-based reweighting with $H=1$. Mirroring-based DiPmark applies mod-1 mirroring to normalized permutation ranks, while StealthInk uses contiguous rank intervals.}
  \label{fig:dip_control}
\end{figure*}

As shown in Fig.~\ref{fig:dip_control}, mirroring-based DiPmark consistently improves AUC, TPR@1\%FPR, and bit accuracy over StealthInk across token budgets and symbol sizes. This result suggests that the advantage of mirroring is not limited to samplers that directly expose sampling values. It also applies to permutation-based reweighting, where mirroring can be used to create symbol-specific effective permutations with stronger matched--mismatched separation than interval-based mappings.

\subsection{Performance comparison over 200 and 400 tokens}\label{200_400tokens}

We present the performance comparison across different approaches over 200 and 400 tokens, respectively as in Table~\ref{ppl_tradeoff_200tokens} and Table~\ref{ppl_tradeoff_400tokens}, where the watermarked text generated by each approach is embedded with 36 bits and 54 bits, respectively.

\begin{table*}[htbp]
\centering
\small
\setlength{\tabcolsep}{2.0pt}        
\renewcommand{\arraystretch}{0.96}   
\vspace{-0.3em}
\caption{Mean perplexity and detectability for different approaches on 200 tokens. Each perplexity is given with a 90\% confidence interval based on bootstrapping.}
\label{ppl_tradeoff_200tokens}

\begin{tabular}{@{}L C C S[table-format=1.4] C !{\vrule width 0.5pt} C C S[table-format=1.4] C@{}}
\toprule
\multirow{2}{*}{Method} & \multicolumn{4}{c}{36 Bits} & \multicolumn{4}{c}{54 Bits} \\
\cmidrule(lr){2-5}\cmidrule(lr){6-9}
& {\footnotesize AUC}
& {\footnotesize TPR@1\%FPR}
& {\footnotesize Bit Acc.}
& {\footnotesize Perplexity}
& {\footnotesize AUC}
& {\footnotesize TPR@1\%FPR}
& {\footnotesize Bit Acc.}
& {\footnotesize Perplexity} \\
\midrule
Non Watermark
& -- & -- & \multicolumn{1}{c}{--} & \makecell{\meanppl{7.7836}\\ \smallci{7.6024,\,7.9665}}
& -- & -- & \multicolumn{1}{c}{--} & \makecell{\meanppl{7.7836}\\ \smallci{7.6024,\,7.9665}} \\
\cmidrule(lr){1-5}\cmidrule(lr){6-9}

MPAC
& 0.9903 & 0.9400 & 0.8893 & \makecell{\meanppl{9.8604}\\ \smallci{9.6782,\, 10.0450}}
& 0.9913 & 0.9180 & 0.8394 & \makecell{\meanppl{10.1388}\\ \smallci{9.9353,\,10.3464}} \\

RSBH
& 0.9983 & \textbf{0.9980} & \textbf{0.9992} & \makecell{\meanppl{32.6466}\\ \smallci{31.2956,\,34.0539}}
& 0.9979 & \textbf{0.9980} & \textbf{0.9928} & \makecell{\meanppl{32.6994}\\ \smallci{31.2430,\,34.2013}} \\

StealthInk
& 0.9787 & 0.6540 & 0.8423 & \makecell{\textbf{\meanppl{7.3038}}\\ \smallci{7.0626,\,7.5421}}
& 0.9654 & 0.4420 & 0.7896 & \makecell{\textbf{\meanppl{7.2339}}\\ \smallci{6.9976,\,7.4662}} \\

Gumbel-max
& \textbf{1.0} & 0.9980 & 0.9613 & \makecell{\meanppl{7.5709}\\ \smallci{7.3951,\,7.7503}}
& 0.9998 & 0.9960 & 0.9338 & \makecell{\meanppl{7.6708}\\ \smallci{7.4902,\,7.8644}} \\

Tour-Wmean
& 0.9955 & 0.9880 & 0.9345 & \makecell{\meanppl{7.7710}\\ \smallci{7.6014,\,7.9373}}
& \textbf{0.9999} & \textbf{0.9980} & 0.8962 & \makecell{\meanppl{7.7592}\\ \smallci{7.5870,\,7.9293}} \\

Tour-Bayes
& 0.9954 & 0.9860 & 0.9495 & \makecell{\meanppl{7.7710}\\ \smallci{7.6014,\,7.9373}}
& 0.9992 & 0.9800 & 0.9051 & \makecell{\meanppl{7.7592}\\ \smallci{7.5870,\,7.9293}} \\

\bottomrule
\end{tabular}
\end{table*}

\begin{table*}[htbp]
\centering
\small
\setlength{\tabcolsep}{2.0pt}        
\renewcommand{\arraystretch}{0.96}   
\vspace{-0.3em}
\caption{Mean perplexity and detectability for different approaches on 400 tokens. Each perplexity is given with a 90\% confidence interval based on bootstrapping.}
\label{ppl_tradeoff_400tokens}

\begin{tabular}{@{}L C C S[table-format=1.4] C !{\vrule width 0.5pt} C C S[table-format=1.4] C@{}}
\toprule
\multirow{2}{*}{Method} & \multicolumn{4}{c}{36 Bits} & \multicolumn{4}{c}{54 Bits} \\
\cmidrule(lr){2-5}\cmidrule(lr){6-9}
& {\footnotesize AUC}
& {\footnotesize TPR@1\%FPR}
& {\footnotesize Bit Acc.}
& {\footnotesize Perplexity}
& {\footnotesize AUC}
& {\footnotesize TPR@1\%FPR}
& {\footnotesize Bit Acc.}
& {\footnotesize Perplexity} \\
\midrule
Non Watermark
& -- & -- & \multicolumn{1}{c}{--} & \makecell{\meanppl{7.0513}\\ \smallci{6.9156,\,7.1849}}
& -- & -- & \multicolumn{1}{c}{--} & \makecell{\meanppl{7.0513}\\ \smallci{6.9156,\,7.1849}} \\
\cmidrule(lr){1-5}\cmidrule(lr){6-9}

MPAC
& 0.9970 & 0.9820 & 0.9599 & \makecell{\meanppl{8.8160}\\ \smallci{8.6754,\,8.9583}}
& 0.9960 & 0.9940 & 0.9227 & \makecell{\meanppl{8.8811}\\ \smallci{8.7232,\,9.0393}} \\

RSBH
& 0.9999 & \textbf{1.0} & \textbf{1.0} & \makecell{\meanppl{32.5108}\\ \smallci{32.1111,\,34.9533}}
& 0.9990 & \textbf{1.0} & \textbf{0.9972} & \makecell{\meanppl{33.6699}\\ \smallci{32.1541,\,35.2284}} \\

StealthInk
& 0.9941 & 0.9500 & 0.9204 & \makecell{\textbf{\meanppl{6.5826}}\\ \smallci{6.4060,\,6.7593}}
& 0.9952 & 0.9400 & 0.8748 & \makecell{\textbf{\meanppl{6.5893}}\\ \smallci{6.4053,\,6.7813}} \\

Gumbel-max
& \textbf{1.0} & \textbf{1.0} & 0.9929 & \makecell{\meanppl{6.8081}\\ \smallci{6.6618,\,6.9545}}
& \textbf{1.0} & \textbf{1.0} & 0.9849 & \makecell{\textbf{\meanppl{6.8855}}\\ \smallci{6.7453,\,7.0332}} \\

Tour-Wmean
& 0.9998 & 0.9960 & 0.9811 & \makecell{\meanppl{7.1759}\\ \smallci{7.0406,\,7.3120}}
& \textbf{1.0} & \textbf{1.0} & 0.9665 & \makecell{\meanppl{7.0888}\\ \smallci{6.9534,\,7.2243}} \\

Tour-Bayes
& 0.9996 & 0.9920 & 0.9819 & \makecell{\meanppl{7.1759}\\ \smallci{7.0406,\,7.3120}}
& 0.9997 & 0.9960 & 0.9706 & \makecell{\meanppl{7.0888}\\ \smallci{6.9534,\,7.2243}} \\

\bottomrule
\end{tabular}
\end{table*}

\subsection{Repeatition score and LLM-as-judge score of the text generated with watermarking scheme}\label{extra_text_quality}

\begin{figure*}[htbp]
  \centering
  \captionsetup[subfigure]{justification=centering, font=small}

  \begin{subfigure}[t]{0.48\textwidth}
    \centering
    \includegraphics[width=0.9\linewidth]{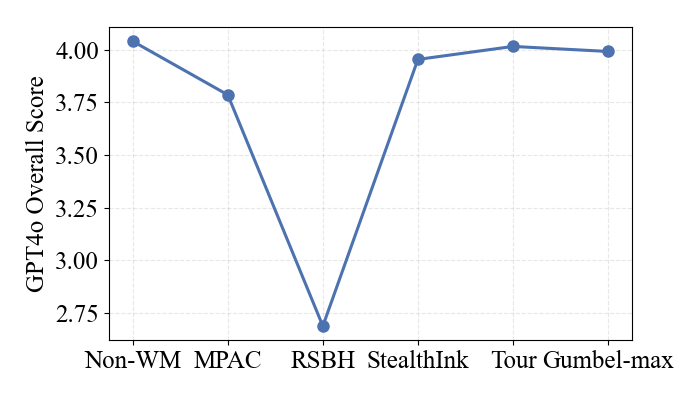}
  \end{subfigure}
  \hfill
  \begin{subfigure}[t]{0.48\textwidth}
    \centering
    \includegraphics[width=0.9\linewidth]{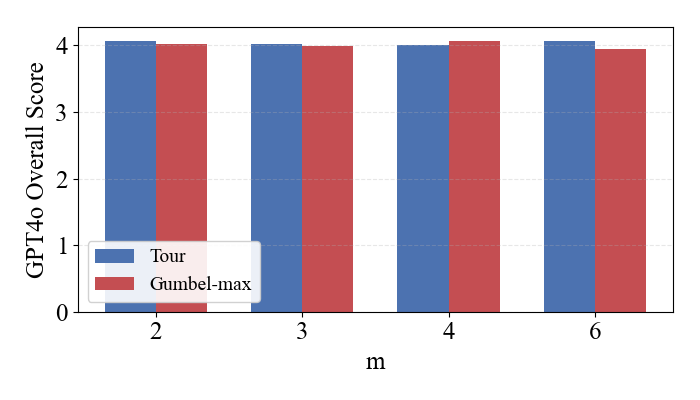}
  \end{subfigure}

  \caption{Text quality scored by GPT4o over 300 tokens, where $m=3$ and $H=12$}
    \label{gpt4o}
\end{figure*}

\begin{table*}[htbp]
\centering
\caption{Text quality scored with distinct-2 and repetition rate across watermarking schemes, 36 bits are embedded in 300 tokens.}
\resizebox{\textwidth}{!}{%
\begin{tabular}{lccccccccccccc}
\toprule
& Non-watermarked &
MPAC & RSBH & StealthInk &
\makecell{TB\\(m=2)} &
\makecell{TB\\(m=3)} &
\makecell{TB\\(m=4)} &
\makecell{TB\\(m=6)} &
\makecell{G-max\\(m=2)} &
\makecell{G-max\\(m=3)} &
\makecell{G-max\\(m=4)} &
\makecell{G-max\\(m=6)} \\
\midrule
Distinct-2 & 0.9471 & 0.9624 & 0.9648 & 0.9498 & 0.9452 & 0.9494 & 0.9475 & 0.9451 & 0.9277 & 0.9269 & 0.9209 & 0.9292\\
Repetition Rate & 0.4542 & 0.4183 & 0.3528 & 0.4410 & 0.4538 & 0.4504 & 0.4509 & 0.4561 & 0.4733 & 0.4761 & 0.4849 & 0.4752\\
\bottomrule
\end{tabular}%
}
\label{tab:repetition_rate}
\end{table*}

We further evaluate the linguistic quality of MirrorMark using two complementary metrics: (1) an LLM-as-a-judge assessment with GPT-4o as in Fig.~\ref{gpt4o}, and (2) a repetition-based analysis using distinct-2 and repetition rate as in Table~\ref{tab:repetition_rate}.

For the LLM-as-a-judge study, GPT-4o scored each text along four dimensions: coherence, clarity, naturalness, and overall quality. Following~\cite{jovanovic2024watermark}, we design the following GPT4o Judge prompt explicitly to ignore truncation effects and focus solely on linguistic fluency. 

Across all configurations, MirrorMark achieves GPT-4o scores that are statistically indistinguishable from the non-watermarked baseline. The overall score difference consistently stays within 0.05–0.10, well inside the natural variance of GPT-4o evaluations. These results confirm that mod-1 mirroring does not degrade linguistic quality, aligning with our theoretical guarantee that MirrorMark is distribution-preserving. In contrast, distortion-based baselines such as MPAC and RSBH exhibit noticeably lower GPT-4o scores, consistent with their higher perplexity and the known side effects of their logit-biasing mechanisms.

The diversity analysis further reinforces these findings. Although MPAC and RSBH report high distinct-2 and low repetition rates, this behavior is driven by artificially skewing the token distribution away from natural language usage, which corresponds to their lower GPT-4o scores. In comparison, MirrorMark, especially the tournament-sampling variant, achieves distinct-2 and repetition rates nearly identical to non-watermarked text, demonstrating that it preserves natural linguistic diversity. While Gumbel-max is inherently more deterministic under top-$k$ sampling and thus yields slightly lower diversity, GPT-4o evaluations confirm that this does not harm fluency or naturalness, as the generated sentences remain coherent and well-structured.

\newtcolorbox{promptbox}{
  breakable,
  enhanced,
  colback=gray!5,
  colframe=teal!60!black,
  title=GPT-4o Judge Prompt,
  fonttitle=\bfseries,
  coltitle=white,
  colbacktitle=teal!70!black,
  sharp corners,
  boxrule=0.7pt,
}

\begin{promptbox}
\small
You are an impartial expert evaluator of linguistic text quality.

The given text is a continuation generated from a truncated C4 sample (15--20 words).
The text may start or end abruptly because the generation length is fixed (e.g., 300 tokens).
Do \textbf{NOT} penalize truncation or incompleteness.

Evaluate \textbf{ONLY} linguistic quality:
\begin{itemize}
  \item Coherence --- logical flow of ideas
  \item Clarity --- easy to understand
  \item Naturalness --- how fluent / human-written the text appears
\end{itemize}

Rate each from 1 to 5.  
Compute ``overall'' as the average of the three.

Return only a JSON object in exactly the following structure:

\begin{verbatim}
{
  "coherence": float,
  "clarity": float,
  "naturalness": float,
  "overall": float
}
\end{verbatim}

Text: \verb|<<<TEXT>>>|
\end{promptbox}

\subsection{Performance of MirrorMark in 72 Bits and 90 Bits}\label{72bits_90bits}

Fig.~\ref{fig:detectability_of_MirrorMark_auc} demonstrates the AUC of MirrorMark in 72 bits and 90 bits across varying number of tokens, respectively.
\begin{figure*}[htbp]
  \centering
  \captionsetup[subfigure]{justification=centering, font=small}

  \begin{subfigure}[t]{0.49\textwidth}
    \centering
    \includegraphics[width=0.7\linewidth]{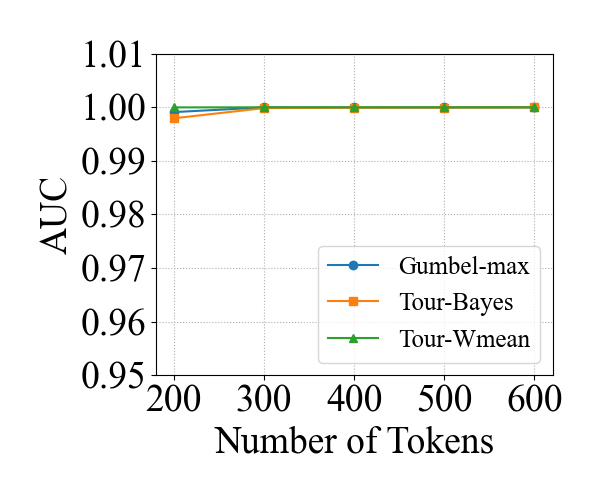}
    \caption{72 Bits}
  \end{subfigure}
  \hfill
  \begin{subfigure}[t]{0.49\textwidth}
    \centering
    \includegraphics[width=0.7\linewidth]{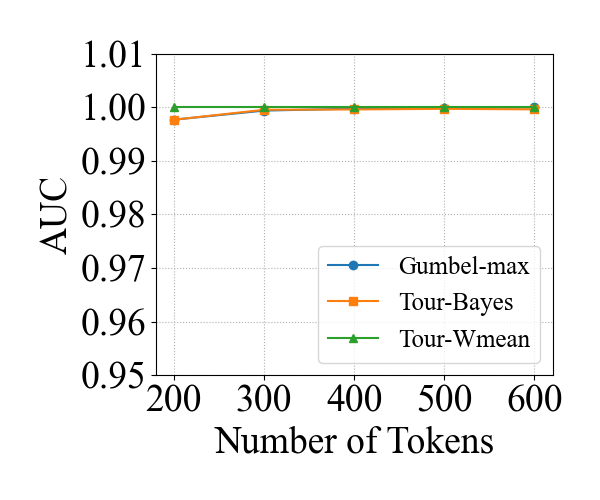}
    \caption{90 Bits}
  \end{subfigure}

  \caption{AUC of MirrorMark across varying number of tokens respectively with 72 and 90 bits embedded.}
  \label{fig:detectability_of_MirrorMark_auc}
\end{figure*}

\begin{figure*}[htbp]
  \centering
  \captionsetup[subfigure]{justification=centering, font=small}

  \begin{subfigure}[t]{0.9\textwidth}
    \centering
    \begin{minipage}[t]{0.49\linewidth}
      \centering
      \includegraphics[width=0.8\linewidth]{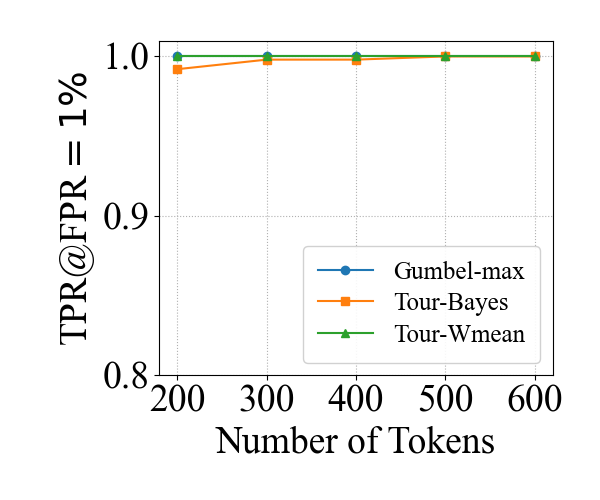}
    \end{minipage}\hfill
    \begin{minipage}[t]{0.49\linewidth}
      \centering
      \includegraphics[width=0.8\linewidth]{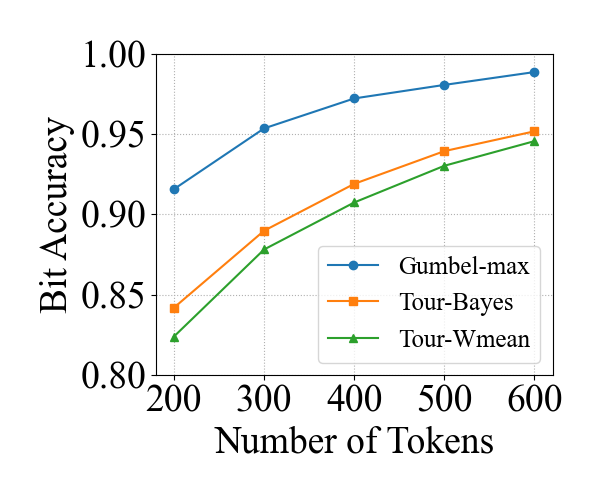}
    \end{minipage}
    \caption{72 Bits}
  \end{subfigure}

  \vspace{0.8em}

  \begin{subfigure}[t]{0.9\textwidth}
    \centering
    \begin{minipage}[t]{0.49\linewidth}
      \centering
      \includegraphics[width=0.8\linewidth]{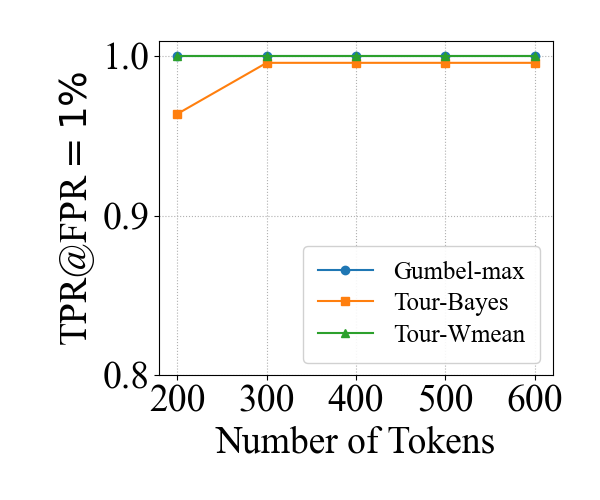}
    \end{minipage}\hfill
    \begin{minipage}[t]{0.49\linewidth}
      \centering
      \includegraphics[width=0.8\linewidth]{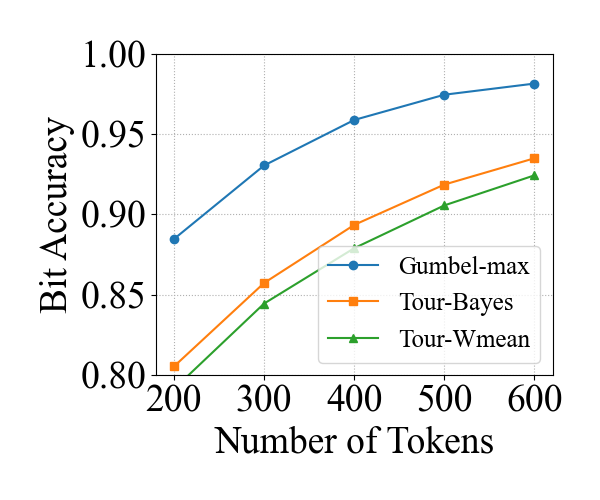}
    \end{minipage}
    \caption{90 Bits}
  \end{subfigure}

  \caption{Detectability of MirrorMark across varying number of tokens with 72 and 90 bits embedded.}
  \label{fig:detectability_of_MirrorMark}
\end{figure*}

\subsection{Threshold Calibration under Argmax Decoding}
\label{calibration}

In practical detection, MirrorMark decodes the embedded message by selecting the message with maximum score. This argmax step complicates the analytical characterization of the final null distribution. Therefore, instead of relying solely on an analytically derived threshold, we empirically calibrate the detection threshold on non-watermarked samples to achieve a target FPR of $1\%$, and evaluate whether the calibrated threshold remains reliable across settings. We first perform held-out calibration within the same setting, where the threshold is calibrated on one split and applied to a held-out split. As shown in Table~\ref{tab:within_setting_calibration}, the held-out FPR remains close to the target $1\%$FPR for both Gumbel-max and Tour-Wmean, while maintaining high TPR, indicating that empirical calibration remains reliable despite the argmax decoding step. We next evaluate cross-length transfer. Raw detection scores exhibit substantial FPR inflation when a threshold calibrated on longer sequences is applied to shorter ones, due to length-dependent scaling. To mitigate this mismatch, we apply z-score normalization. For Gumbel-max, we use $z = (C - 1)\sqrt{T}$, where $C$ is the average score and $T$ is the number of tokens. For Tour-Wmean with $m>1$, as follows, we use a variance-normalized approximation with null mean $1/2$,
\begin{equation}
    \begin{aligned}
        z = \frac{C - 1/2}{\sqrt{\mathrm{Var}(C)}}, \quad
\mathrm{Var}(C) \approx \sum_{\ell=1}^{L} \frac{\alpha_\ell^2}{12L^2T},
    \end{aligned}
\end{equation}
after normalizing the layer weights such that $\frac{1}{L}\sum_{\ell=1}^{L}\alpha_\ell=1$. As shown in Table~\ref{tab:cross_length_calibration}, z-score normalization substantially reduces FPR inflation in cross-length transfer while preserving high TPR. Finally, we evaluate cross-position transfer. As shown in Table~\ref{tab:cross_position_calibration}, z-score normalization does not fully resolve the mismatch and can sometimes worsen it. This indicates that the remaining shift is not solely due to token-count scaling. That is because when fewer tokens are assigned to each position, decoding becomes noisier and more sensitive to random fluctuations, while increasing the token budget per position improves stability.

\begin{table}[htbp]
\centering
\small
\caption{Held-out calibration within the same setting. Each entry reports held-out FPR / TPR after calibrating to target FPR $=0.01$.}
\label{tab:within_setting_calibration}
\begin{tabular}{lcc}
\toprule
Setting & Gumbel-max & Tour-Wmean \\
\midrule
pos=18, 300 tokens & 0.002 / 1.000 & 0.008 / 0.986 \\
pos=24, 300 tokens & 0.004 / 0.996 & 0.004 / 0.994 \\
pos=18, 400 tokens & 0.002 / 1.000 & 0.008 / 0.994 \\
pos=24, 400 tokens & 0.004 / 0.996 & 0.008 / 1.000 \\
\bottomrule
\end{tabular}
\end{table}

\begin{table}[htbp]
\centering
\small
\caption{Cross-length threshold transfer. Raw FPR / TPR $\rightarrow$ z-normalized FPR / TPR.}
\label{tab:cross_length_calibration}
\begin{tabular}{lcc}
\toprule
Transfer & Gumbel-max & Tour-Wmean \\
\midrule
pos=18, 400 $\rightarrow$ 200
& 0.130 / 1.000 $\rightarrow$ 0.004 / 0.998
& 0.772 / 1.000 $\rightarrow$ 0.034 / 0.986 \\

pos=24, 400 $\rightarrow$ 200
& 0.830 / 1.000 $\rightarrow$ 0.006 / 0.994
& 0.910 / 1.000 $\rightarrow$ 0.008 / 0.980 \\

pos=24, 400 $\rightarrow$ 300
& 0.086 / 1.000 $\rightarrow$ 0.006 / 0.994
& 0.176 / 1.000 $\rightarrow$ 0.010 / 0.988 \\
\bottomrule
\end{tabular}
\end{table}

\begin{table}[htbp]
\centering
\small
\caption{Cross-position threshold transfer. Raw FPR / TPR $\rightarrow$ z-normalized FPR / TPR.}
\label{tab:cross_position_calibration}
\begin{tabular}{lcc}
\toprule
Transfer & Gumbel-max & Tour-Wmean \\
\midrule
pos=18, 200 $\rightarrow$ pos=24, 200
& 0.100 / 0.998 $\rightarrow$ 0.154 / 0.996
& 0.060 / 0.994 $\rightarrow$ 0.140 / 0.994 \\

pos=18, 300 $\rightarrow$ pos=24, 300
& 0.006 / 0.996 $\rightarrow$ 0.028 / 0.996
& 0.072 / 1.000 $\rightarrow$ 0.210 / 1.000 \\

pos=18, 200 $\rightarrow$ pos=24, 300
& 0.002 / 0.996 $\rightarrow$ 0.080 / 0.998
& 0.004 / 0.984 $\rightarrow$ 0.150 / 1.000 \\
\bottomrule
\end{tabular}
\end{table}

\subsection{Cross-language adaptation}\label{cross_lang}

To evaluate whether MirrorMark is tied to a specific language or can be reliably applied across languages, we conduct a cross-language experiment using the multilingual XL-Sum dataset~\citep{hasan2021xlsumlargescalemultilingualabstractive} on Gemma-7B-it~\citep{gemmateam2024gemmaopenmodelsbased}. For each language (English, Chinese, and Russian), we sample summaries from XL-Sum and prompt the model to generate full news articles in the corresponding language. During generation, we apply exactly the same MirrorMark watermarking rule as in our main experiments. For each language, we generate 500 paired watermarked and non-watermarked samples of length 200 tokens, and evaluate both the Bayesian detector for tournament sampling (Tour-Bayes) and the analytic detector for Gumbel-max.

Fig.~\ref{roc_cross_lang} shows that a detection threshold~$\tau$ calibrated in one language does not perfectly transfer to another. In particular, when a threshold learned on English is applied to Chinese, the empirical false positive rate (FPR) on Chinese increases, whereas applying the same threshold to Russian yields largely unchanged behavior. Conversely, a threshold calibrated on Chinese becomes overly conservative when applied to English or Russian, reducing both FPR and true positive rate (TPR).

This cross-language threshold mismatch is consistent with a well-established empirical observation~\citep{Montemurro20111UniversalEntropy}: Chinese text exhibits systematically lower next-token entropy than English and Russian, while English and Russian have similar entropy profiles. As a result, both watermarked and non-watermarked scores for Chinese are expected to be shifted toward larger values, even when the separation between the two hypotheses remains comparable. Consequently, a threshold~$\tau$ calibrated on English (where the non-watermarked distribution is farther left) becomes slightly too permissive for Chinese, increasing FPR, whereas a threshold calibrated on Chinese becomes too strict when applied to English or Russian.

Overall, Fig.~\ref{roc_cross_lang} demonstrates that MirrorMark is not tied to English or any particular dataset. Across all three languages, both the tournament-based (Tour-Bayes) and Gumbel-max variants remain reliably detectable, with similar ROC trends. The observed differences are limited to small score-scale shifts induced by language-specific entropy characteristics, which can be addressed through simple threshold recalibration. These results support our claim that MirrorMark is a data-agnostic generative watermark whose detectability is primarily governed by sequence length and entropy, rather than by language or domain.

\begin{figure*}[htbp]
  \centering
  \captionsetup[subfigure]{justification=centering, font=small}

  \begin{subfigure}[t]{0.49\textwidth}
    \centering
    \includegraphics[width=0.9\linewidth]{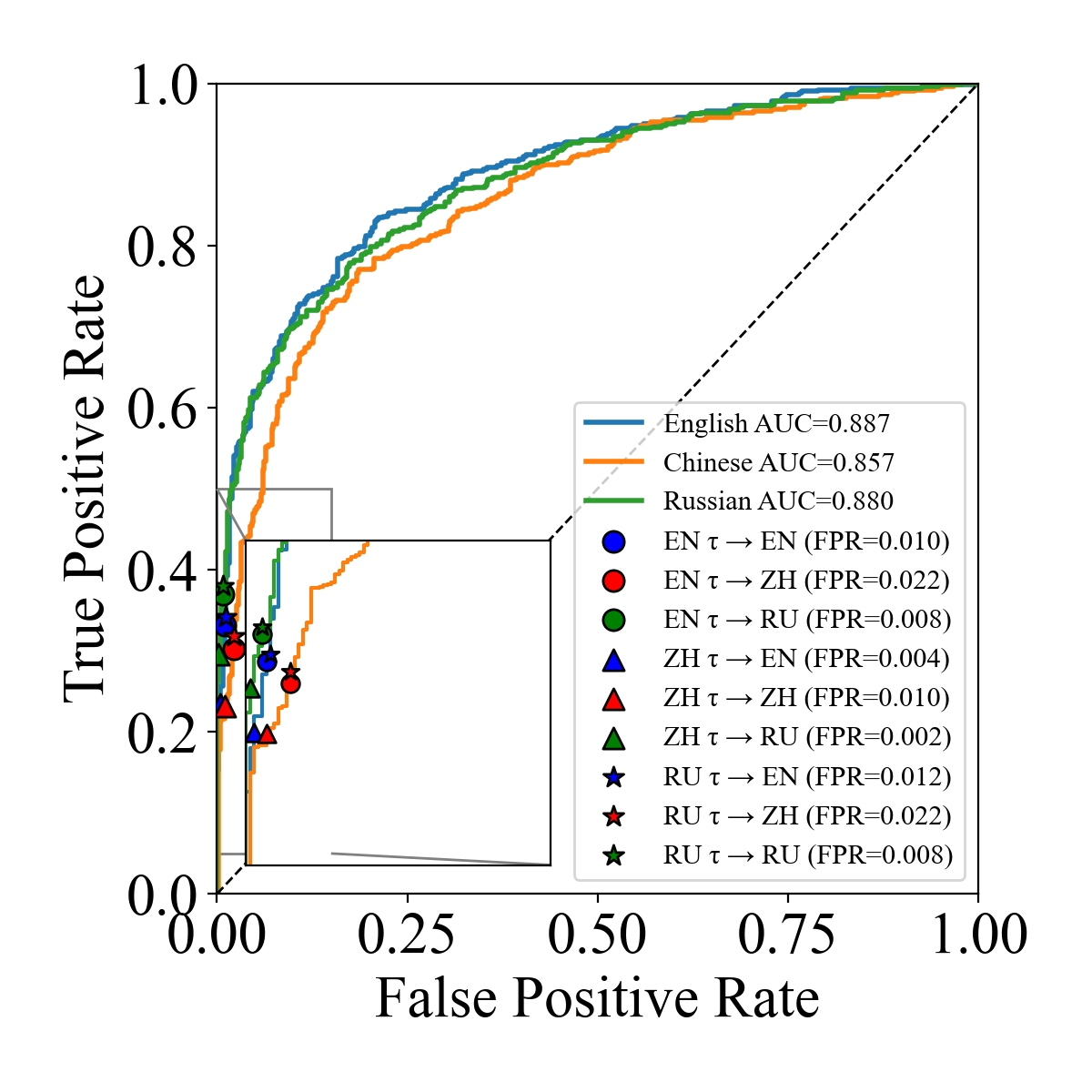}
    \caption{Tour-Bayes}
  \end{subfigure}
  \hfill
  \begin{subfigure}[t]{0.49\textwidth}
    \centering
    \includegraphics[width=0.9\linewidth]{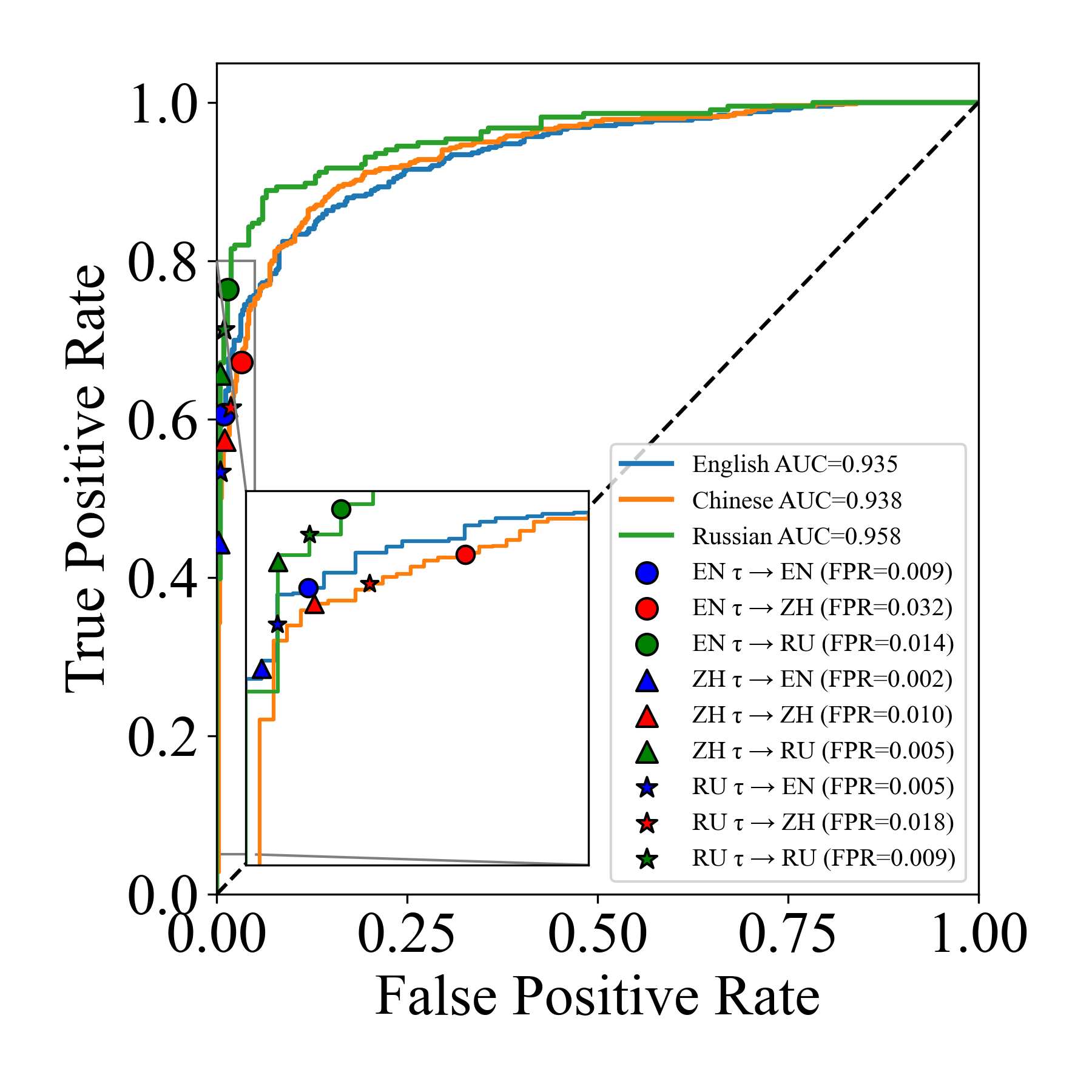}
    \caption{Gumbel-max}
  \end{subfigure}

  \caption{ROC across three languages over 200 tokens, where $m=3$ and the number of positions $H$=12.}\label{roc_cross_lang}
\end{figure*}

\subsection{Empirical EER comparison}\label{sec:alignment}
\vspace{-0.5em}
\begin{figure}[htbp]
  \centering

  \includegraphics[width=0.6\columnwidth]{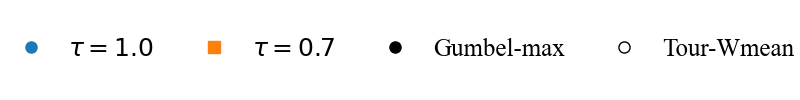}

  \vspace{-1em}

  \begin{minipage}[htbp]{0.48\columnwidth}
    \centering
    \includegraphics[width=0.7\linewidth]{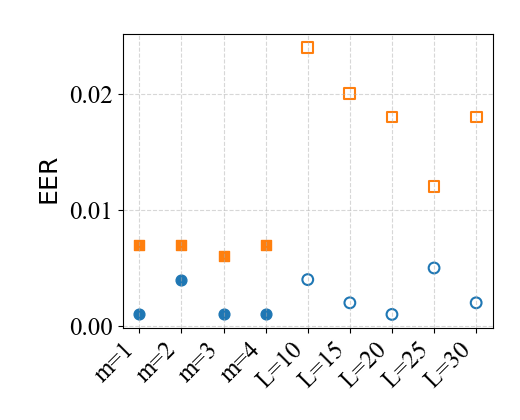}
  \end{minipage}
  \hfill
  \begin{minipage}[htbp]{0.48\columnwidth}
    \centering
    \includegraphics[width=0.7\linewidth]{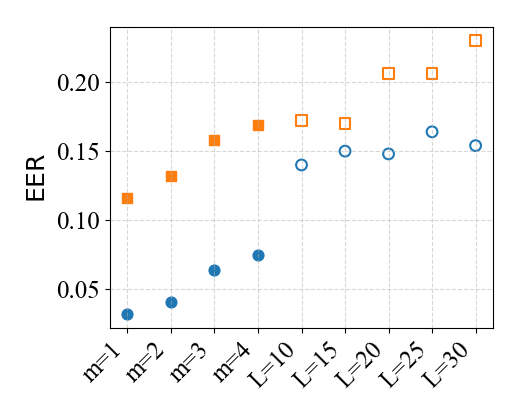}
  \end{minipage}

  \caption{Empirical EER of MirrorMark on LLaMA-2-7B with C4 prompts (left) and Gemma-7B-it~\citep{gemmateam2024gemmaopenmodelsbased} with ELI5 prompts~\citep{fan-etal-2019-eli5} (right), with token length $T=200$. For tournament-based MirrorMark, $m=1$.}
  \label{fig:eer_alignment}
\end{figure}

Fig.~\ref{fig:eer_alignment} reports empirical EER under a fixed token budget of $T=200$. In the high-entropy regime (left plot; $\tau=1.0$ with $\mathcal{H}\approx 1.7$), Gumbel-max and Tour-Wmean achieve comparable EER. When entropy decreases (e.g., $\tau=0.7$ with $\mathcal{H}\approx 0.91$), the advantage of Gumbel-max becomes much more pronounced, and the EER of tournament-based decoding increases more substantially than it does at $\tau=1.0$. Under even stricter entropy conditions (right plot; e.g., $\mathcal{H}\approx 0.54$ at $\tau=1.0$), the impact of the symbol size $m$ in Gumbel-max becomes evident, with EER increasing as $m$ grows, consistent with the theoretical dependence on $m$. In contrast, the tournament variant exhibits a stronger sensitivity to the choice of $L$; in particular, larger $L$ can lead to noticeably worse EER (e.g., $L=30$ under both $\tau=1.0$ and $\tau=0.7$). This behavior is expected in low-entropy regimes. While increasing the tournament depth $L$ intuitively aggregates more watermark evidence, each additional laye further consumes the limited randomness available during decoding, leading to a higher collision probability and diminishing returns from later layers. As a result, the marginal contribution of deeper layers quickly diminishes and may even become detrimental when noise dominates the signal. This effect is pronounced when Uniform-distributed $u$ values are used in tournament sampling, as they tend to induce higher collision probabilities across layers under constrained entropy. 

\begin{figure}[htbp]
  \centering
  \includegraphics[width=0.5\columnwidth]{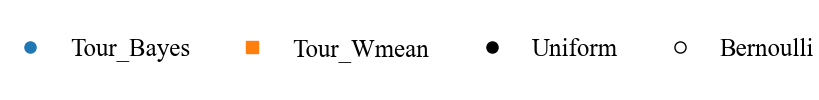}

  \begin{minipage}{0.8\columnwidth}
    \centering
    \includegraphics[width=0.5\linewidth]{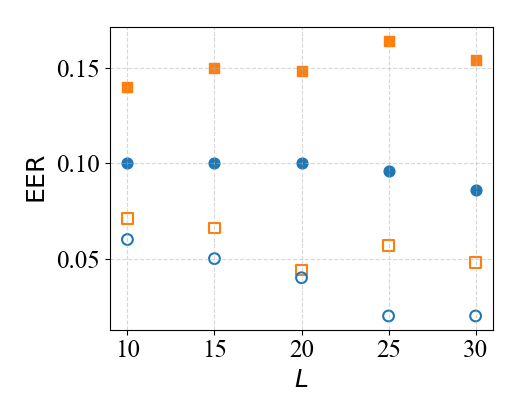}
  \end{minipage}

  \vspace{-0.4em}
  \caption{Empirical EER of tournament-based MirrorMark on Gemma-7B-it with ELI5 prompts, with token length $T=200$ and temperature $\tau=1.0$. $m=1$ and $H=1$.}
  \label{fig:eer_wmean_bayes}
\end{figure}

In Fig.~\ref{fig:eer_wmean_bayes}, we observe that Tour-Bayes with Bernoulli-distributed $u$ values achieves a faster EER reduction as the number of layers increases than with Uniform-distributed $u$ values. The reason is that Bernoulli values maximize the discreteness of the random variables, producing more polarized per-layer evidence. This stronger contrast is less susceptible to collision noise and is especially beneficial to Bayesian detection, which models the likelihood structure of $u$ values across tournament layers.

However, such values are binary and support only one bit per position. Since our focus is on mod-1 mirroring for multi-bit embedding, this advantage does not directly apply.

\subsection{Detectability and text quality on instruction task}\label{instruction-task}

Fig.~\ref{fig:gemma} reports the detectability of MirrorMark on instruction-following generation, where 500 randomly selected ELI5 prompts are evaluated using the Gemma-7B-it model, where the temperature is $1.0$. Here we include Tour-Bayes as a representative tournament-based detector under the main evaluation protocol. We observe that Gumbel-max–based MirrorMark achieves higher detection performance than tournament-based MirrorMark across all metrics. For example, in Fig.~\ref{fig:gemma}(b), the Gumbel-max–based MirrorMark reaches approximately 80\% TPR@1\%FPR at 100 tokens, while in Fig.~\ref{fig:gemma}(e), the tournament-based MirrorMark reaches around 65\% at the same length.

Despite the reduced detectability, MirrorMark maintains text quality comparable to non-watermarked text. Table~\ref{tab:ppl_instruction} reports the perplexity of 200-token instruction-following responses. Across all watermarking configurations, the perplexity of watermarked text remains very close to the non-watermarked baseline, and the 90\% confidence intervals largely overlap. Switching the tournament sampler from $\mathrm{Uniform}(0,1)$ to Bernoulli $u \in \{0,1\}$ slightly increases perplexity to $1.8208$ ($[1.7882, 1.8541]$) for $m{=}1$, but the change remains modest and does not indicate systematic degradation in output quality.

\begin{figure*}[htbp]
  \centering

  \begin{minipage}[t]{0.32\textwidth}
    \centering
    \subcaptionbox{Gumbel-max}{
      \includegraphics[width=\linewidth]{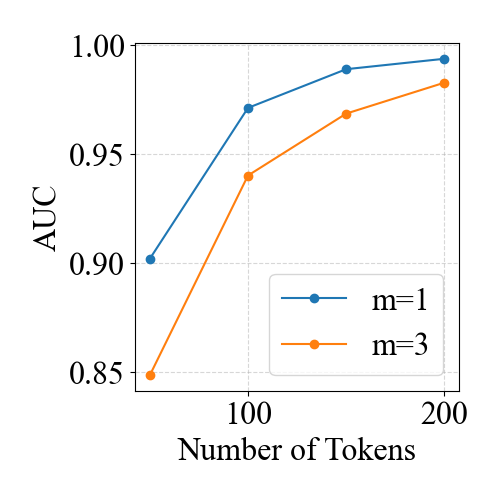}
    }
  \end{minipage}\hfill
  \begin{minipage}[t]{0.32\textwidth}
    \centering
    \subcaptionbox{Gumbel-max}{
      \includegraphics[width=\linewidth]{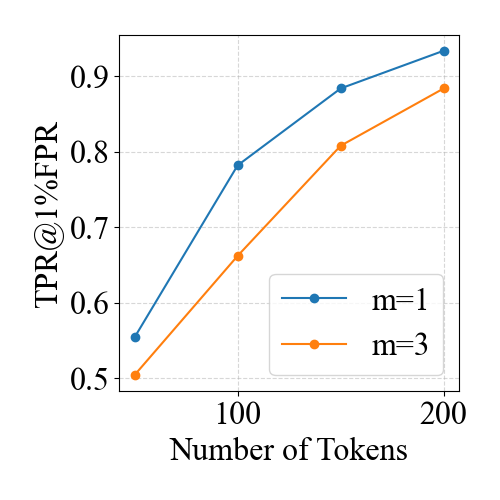}
    }
  \end{minipage}\hfill
  \begin{minipage}[t]{0.32\textwidth}
    \centering
    \subcaptionbox{Gumbel-max}{
      \includegraphics[width=\linewidth]{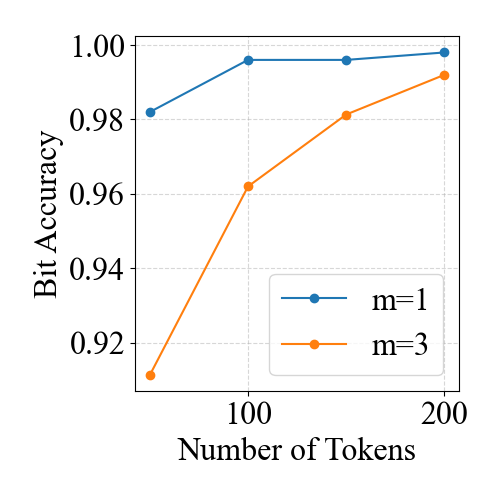}
    }
  \end{minipage}

  \begin{minipage}[t]{0.32\textwidth}
    \centering
    \subcaptionbox{Tour-Bayes, $m=1$, Uniform}{
      \includegraphics[width=\linewidth]{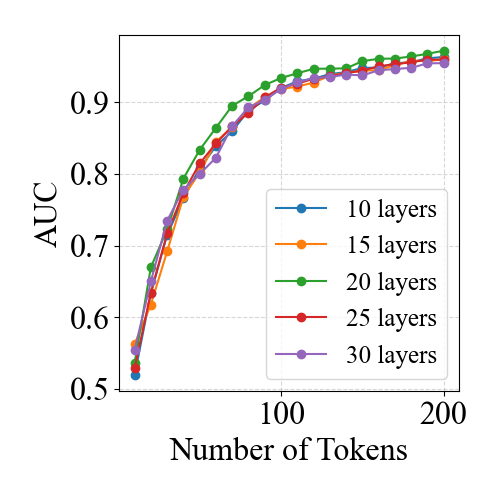}
    }
  \end{minipage}\hfill
  \begin{minipage}[t]{0.32\textwidth}
    \centering
    \subcaptionbox{ Tour-Bayes, $m=1$, Uniform}{
      \includegraphics[width=\linewidth]{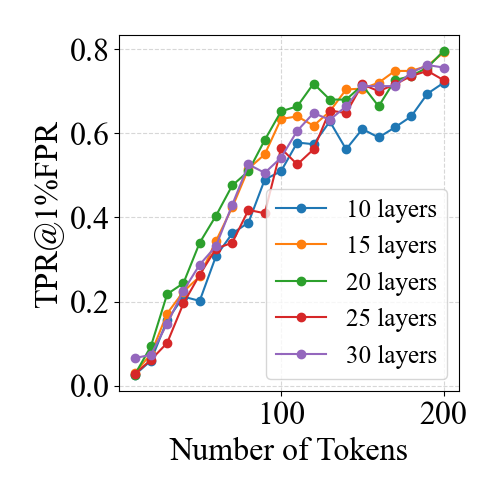}
    }
  \end{minipage}\hfill
  \begin{minipage}[t]{0.32\textwidth}
    \centering
    \subcaptionbox{Tour-Bayes, $m=1$, Uniform}{
      \includegraphics[width=\linewidth]{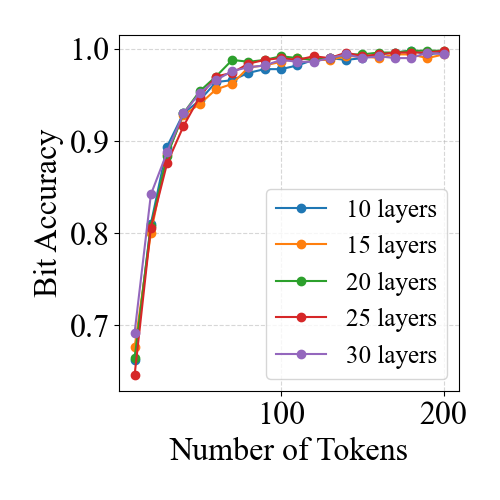}
    }
  \end{minipage}

\begin{minipage}[t]{0.32\textwidth}
    \centering
    \subcaptionbox{Tour-Bayes, $m=3$, Uniform}{
      \includegraphics[width=\linewidth]{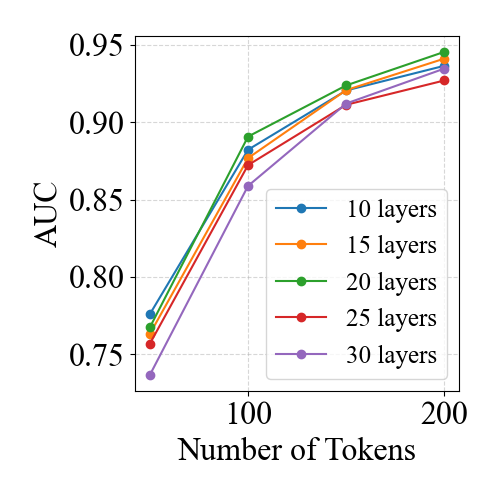}
    }
  \end{minipage}\hfill
  \begin{minipage}[t]{0.32\textwidth}
    \centering
    \subcaptionbox{Tour-Bayes, $m=3$, Uniform}{
      \includegraphics[width=\linewidth]{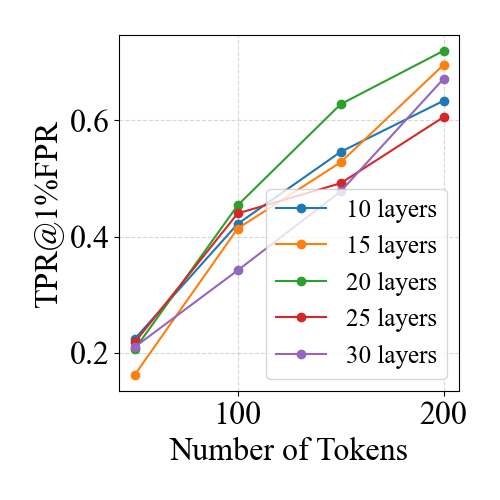}
    }
  \end{minipage}\hfill
  \begin{minipage}[t]{0.32\textwidth}
    \centering
    \subcaptionbox{ Tour-Bayes, $m=3$, Uniform}{
      \includegraphics[width=\linewidth]{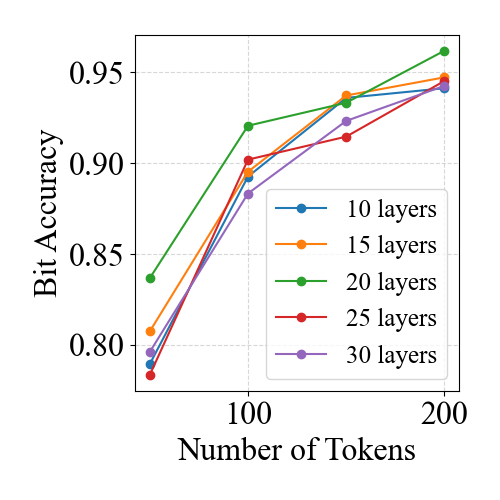}
    }
  \end{minipage}
  
  \begin{minipage}[t]{0.32\textwidth}
    \centering
    \subcaptionbox{Tour-Bayes, $m=1$, Bernoulli}{
      \includegraphics[width=\linewidth]{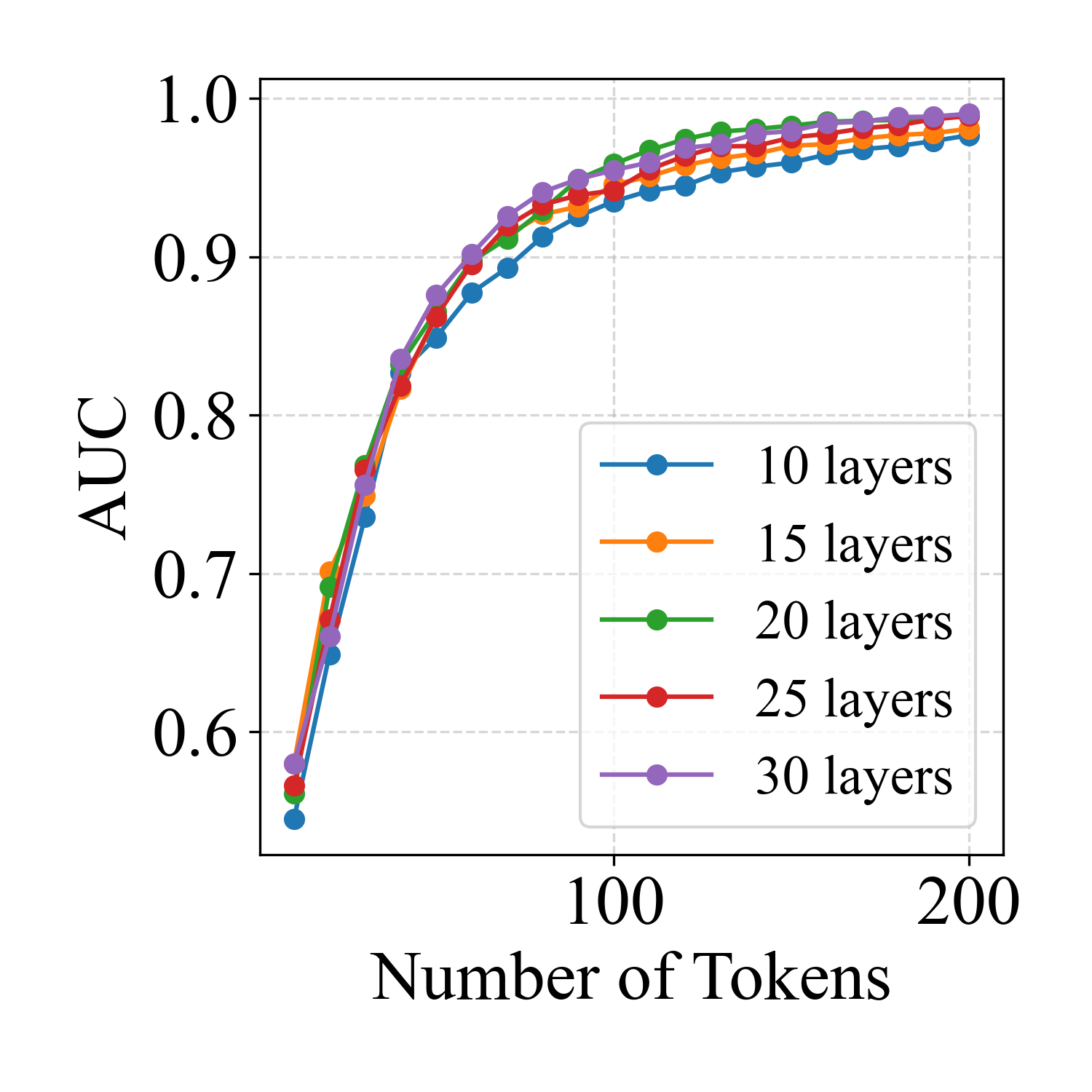}
    }
  \end{minipage}\hfill
  \begin{minipage}[t]{0.32\textwidth}
    \centering
    \subcaptionbox{Tour-Bayes, $m=1$, Bernoulli}{
      \includegraphics[width=\linewidth]{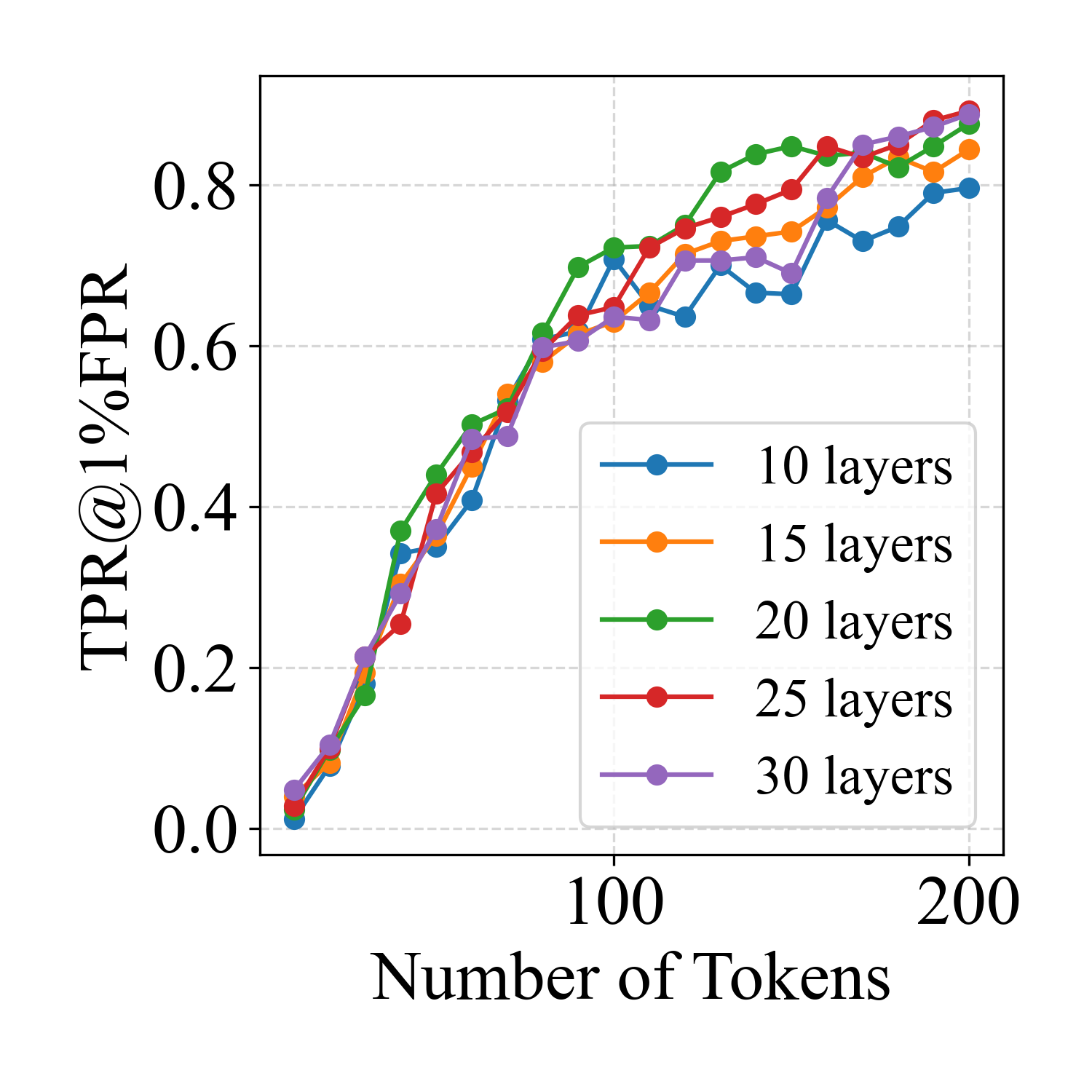}
    }
  \end{minipage}\hfill
  \begin{minipage}[t]{0.32\textwidth}
    \centering
    \subcaptionbox{ Tour-Bayes, $m=1$, Bernoulli}{
      \includegraphics[width=\linewidth]{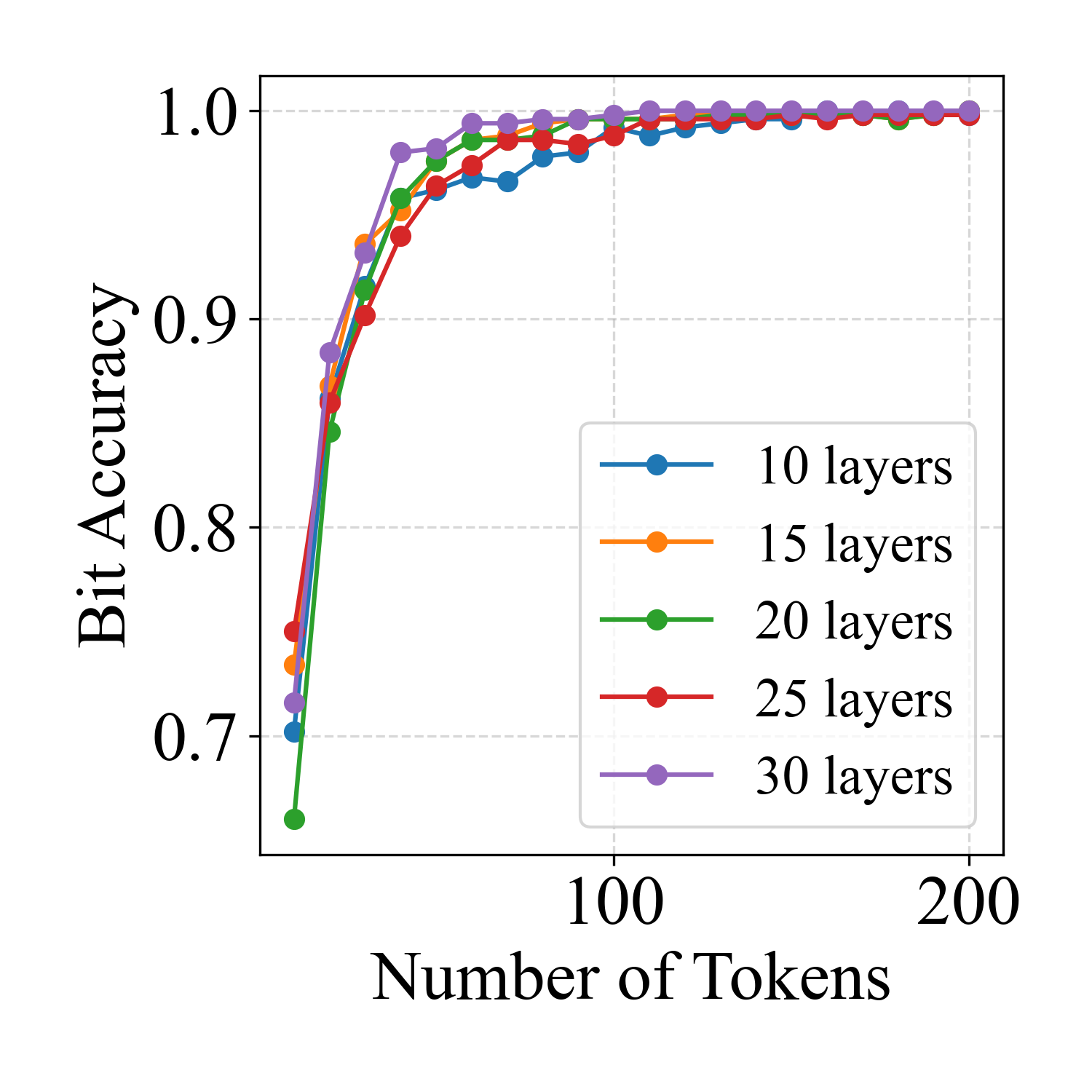}
    }
  \end{minipage}
  \caption{Detectability of MirrorMark on Gemma-7B-it and ELI5 prompts, with watermark of $m\in\{1, 3\}$ and $H=1$ embedded in each response. The temperature is set to be $1.0$}
  \label{fig:gemma}
\end{figure*}

\begin{table}[htbp]
\centering
\footnotesize
\setlength{\tabcolsep}{4pt}        
\renewcommand{\arraystretch}{0.95}  
\caption{Perplexity of 200-token instruction-following responses on Gemma-7B-it with ELI5 prompts. We report the mean perplexity with a 90\% bootstrap confidence interval. Bernoulli $u$ is only applicable to tournament-based MirrorMark and only when $m=1$.}
\label{tab:ppl_instruction}
\begin{tabular}{lccc}
\toprule
& \multicolumn{2}{c}{$m$=1} & $m$=3 \\
\cmidrule(lr){2-3}\cmidrule(lr){4-4}
& Uniform & Bernoulli & Uniform \\
\midrule

Non-watermark
& \multicolumn{3}{c}{\makecell[c]{\meanppl{1.7829}\\ \smallci{1.7569,\,1.8094}}} \\\midrule

Gumbel-max
& \makecell[c]{\meanppl{1.8101}\\ \smallci{1.7785,\,1.8432}}
& --
& \makecell[c]{\meanppl{1.7940}\\ \smallci{1.7621,\,1.8266}} \\

Tour-Bayes
& \makecell[c]{\meanppl{1.7784}\\ \smallci{1.7493,\,1.8079}}
& \makecell[c]{\meanppl{1.8208}\\ \smallci{1.7882,\,1.8541}}
& \makecell[c]{\meanppl{1.7986}\\ \smallci{1.7681,\,1.8300}} \\
\bottomrule
\end{tabular}
\end{table}

\subsection{Detetability Comparison over different $m$ for MirrorMark after copy-paste attack}\label{cp_appendix}

\begin{table*}[htbp]
\centering
\small
\setlength{\tabcolsep}{2.5pt}     
\renewcommand{\arraystretch}{1.06} 
\caption{Detectability for different approaches on 400 tokens with 36 bits embedded after copy-paste attack, where the edit fraction $\epsilon\in\{0.1, 0.3, 0.5\}$.}
\label{tab:detect_400tokens_copy_paste_135}

\begin{tabular}{@{}L
                C C C !{\vrule width 0.5pt}
                C C C !{\vrule width 0.5pt}
                C C C@{}}
\toprule
\multirow{2}{*}{Method}
& \multicolumn{3}{c}{$\epsilon=0.1$}
& \multicolumn{3}{c}{$\epsilon=0.3$}
& \multicolumn{3}{c}{$\epsilon=0.5$} \\
\cmidrule(lr){2-4}\cmidrule(lr){5-7}\cmidrule(lr){8-10}
& {\footnotesize AUC} & {\footnotesize TPR@1\%FPR} & {\footnotesize Bit Acc.}
& {\footnotesize AUC} & {\footnotesize TPR@1\%FPR} & {\footnotesize Bit Acc.}
& {\footnotesize AUC} & {\footnotesize TPR@1\%FPR} & {\footnotesize Bit Acc.} \\
\midrule
MPAC
& 0.9847 & 0.9025 & 0.9263
& 0.9729 & 0.8650 & 0.8725
& 0.9290 & 0.6075 & 0.7959 \\

RSBH
& 0.9840 & 0.4275 & 0.6156
& 0.9386 & 0.0150 & 0.6181
& 0.7243 & 0.01 & 0.5825 \\

StealthInk
& 0.9901 & 0.9100 & 0.8870
& 0.9636 & 0.6675 & 0.8213
& 0.8374 & 0.2575 & 0.7419 \\

Tour-Wmean
& 0.9989 & 0.9980 & 0.9357
& 0.9932 & 0.9680 & 0.8750
& 0.9671 & 0.7560 & 0.7880 \\

Tour-Bayes
& 0.9987 & 0.9960 & 0.9357
& 0.9944 & 0.9760 & 0.8750
& 0.9787 & 0.7920 & 0.7880 \\

Gumbel-max
& \textbf{1.0} & \textbf{1.0} & \textbf{0.9801}
& \textbf{1.0}& \textbf{1.0} & \textbf{0.9549}
& \textbf{1.0} & \textbf{1.0} & \textbf{0.8986} \\
\bottomrule
\end{tabular}
\end{table*}

Table \ref{tab:detect_400tokens_copy_paste_135} compares the detectability of different approaches under copy-paste attacks, where each watermarked text contains 36 embedded bits within a 400-token sequence. Besides, Fig.~\ref{fig:cp_m} and Fig.~\ref{fig:cp_m135} shows the detectability of MirrorMark against copy-paste attacks with 36 bits embedded in 400 tokens, where different $m$ are compared. $m\in\{2, 3, 4, 6\}$ is corresponding to $H\in\{18, 12, 9, 6\}$ respectively.

\subsection{Performance of $m=1$ defined by~\eqref{eq:mod1_m1} and~\eqref{eq:mod1_reflect_compact}}\label{1_bit_justification}

We evaluate the $m=1$ case using 50 tokens generated by LLaMA-2-7B and Gemma-7B-it on C4 and ELI5 prompts, respectively, both with temperature $\tau=1.0$. On LLaMA-2-7B responses, the entropy is higher, and both~\eqref{eq:mod1_m1} and~\eqref{eq:mod1_reflect_compact} achieve identical performance. On Gemma-7B-it where the entropy is lower, we observe mild bit-0/1 gaps for both formulations. Intuitively,~\eqref{eq:mod1_m1} corresponds to a threshold-based decision rule, while~\eqref{eq:mod1_reflect_compact} induces a more symmetric partition. Under low entropy, skewed token distributions make threshold-based decisions more sensitive to finite-sample noise.

\begin{table}[htbp]
\centering
\begin{tabular}{l c c c c}
\toprule
Model & Setting & Bit Acc. & Bit 1 & Bit 0 \\
\midrule
LLaMA & \eqref{eq:mod1_m1} & 1.00 & 1.00 & 1.00 \\
LLaMA & \eqref{eq:mod1_reflect_compact}  & 1.00 & 1.00 & 1.00 \\
Gemma & \eqref{eq:mod1_m1} & 0.94 & 0.92 & 0.96 \\
Gemma & \eqref{eq:mod1_reflect_compact}  & 0.95 & 0.95 & 0.96 \\
\bottomrule
\end{tabular}
\caption{Comparison of 1-bit mod-1 mirroring formulations under different entropy regimes.}
\label{tab:m1_comparison}
\end{table}

\subsection{Ablation study for CABS}

\subsubsection{Sensitivity of Parameters}\label{ablation_cabs_parameters}

To analyze the sensitivity of CABS to its design parameters, we conduct a comprehensive ablation over the frame size $f$, context window $W$, and maximum frame expansion factor $\texttt{max\_factor}$.
Tables~\ref{tab:cabs-insert-final}, \ref{tab:cabs-del-final}, and \ref{tab:cabs-sub-final} report results under insertion, deletion, and substitution attacks, respectively.
For each attack, we consider edit ratios $\epsilon \in \{0, 0.2, 0.4\}$, where $\epsilon=0$ corresponds to the no-attack setting.

Overall, the ablation results reveal clear and consistent trends across attack types.
Setting $f=3$ consistently achieves the highest bit accuracy and strong TPR@1\%FPR under all attacks, indicating an optimal balance between robustness and effective token utilization.
The context window $W=4$ performs best or near-best across all edit ratios, capturing sufficient contextual information without overfitting to local perturbations.
Similarly, $\texttt{max\_factor}=1.5$ yields the strongest robustness across edit rates, balancing frame-size flexibility and stability.
These observations collectively justify the default configuration used in the main paper: $f=3, W=4, \texttt{max\_factor}=1.5$.

We further observe that different attack types affect detectability and bit recovery in distinct ways.
Insertion primarily shifts token positions forward.
Even at $\epsilon=0.4$, MirrorMark maintains high detectability (AUC $=0.999$, TPR@1\%FPR $=0.992$), while bit accuracy drops to $0.790$, indicating that insertion mainly impairs bit recovery rather than WM/Non-WM separation.
Deletion is the most adversarial attack, as it reduces the number of available tokens.
At $\epsilon=0.4$, AUC remains above chance, i.e., $0.939$ and TPR@1\%FPR degrades to $0.604$, which arises not only from desynchronization of the token-to-position mapping but also from reduced detectability due to fewer surviving tokens.
In contrast, substitution preserves sequence length and is the least destructive.
At $\epsilon=0.4$, MirrorMark sustains strong detectability, i.e., AUC=$ 0.998$, TPR@1\%FPR=$ 0.992$, and relatively high bit accuracy, i.e., $0.75$--$0.78$, confirming that CABS effectively absorbs localized perturbations.

\begin{table*}[htbp]
\centering
\small
\begin{tabular}{l l c c c}
\toprule
 & Setting & AUC & TPR@1\%FPR & Bit Accuracy \\
\midrule
\multicolumn{5}{l}{\textbf{Varying $f$ (with $W=4$, max\_factor=1.5)}} \\
$\epsilon=0.0$ & $f=1$ & \textbf{1.000} & 0.998 & 0.939 \\
 & $f=2$ & \textbf{1.000} & 0.998 & 0.952 \\
 & $f=3$ & \textbf{1.000} & \textbf{1.000} & \textbf{0.985} \\
 & $f=4$ & \textbf{1.000} & 0.998 & 0.957 \\
\midrule
$\epsilon=0.2$ & $f=1$ & 0.999 & 0.996 & 0.828 \\
 & $f=2$ & 0.999 & 0.996 & 0.838 \\
 & $f=3$ & \textbf{1.000} & \textbf{0.998} & \textbf{0.852} \\
 & $f=4$ & \textbf{1.000} & 0.996 & 0.847 \\
\midrule
$\epsilon=0.4$ & $f=1$ & 0.998 & 0.984 & 0.772 \\
 & $f=2$ & \textbf{0.999} & 0.988 & 0.766 \\
 & $f=3$ & \textbf{0.999} & \textbf{0.992} & \textbf{0.790} \\
 & $f=4$ & \textbf{0.999} & \textbf{0.992} & 0.785 \\
\midrule
\multicolumn{5}{l}{\textbf{Varying $W$ (with $f=3$, max\_factor=1.5)}} \\
$\epsilon=0.0$ & $W=1$ & \textbf{1.000} & 0.998 & 0.945 \\
 & $W=2$ & \textbf{1.000} & 0.998 & 0.943 \\
 & $W=3$ & \textbf{1.000} & 0.998 & 0.946 \\
 & $W=4$ & \textbf{1.000} & \textbf{1.000} & \textbf{0.985} \\
 & $W=5$ & \textbf{1.000} & \textbf{1.000} & 0.944 \\
\midrule
$\epsilon=0.2$ & $W=1$ & 0.999 & 0.994 & 0.841 \\
 & $W=2$ & 0.999 & \textbf{0.998} & 0.843 \\
 & $W=3$ & \textbf{1.000} & \textbf{0.998} & 0.841 \\
 & $W=4$ & \textbf{1.000} & \textbf{0.998} & \textbf{0.852} \\
 & $W=5$ & \textbf{1.000} & 0.996 & 0.835 \\
\midrule
$\epsilon=0.4$ & $W=1$ & 0.998 & 0.990 & 0.769 \\
 & $W=2$ & 0.999 & 0.992 & 0.770 \\
 & $W=3$ & \textbf{1.000} & \textbf{0.994} & 0.769 \\
 & $W=4$ & 0.999 & 0.992 & \textbf{0.790} \\
 & $W=5$ & 0.998 & 0.982 & 0.763 \\
\midrule
\multicolumn{5}{l}{\textbf{Varying max\_factor (with $f=3$, $W=4$)}} \\
$\epsilon=0.0$ & $\text{max\_factor}=1.25$ & \textbf{1.000} & \textbf{1.000} & 0.948 \\
 & $\text{max\_factor}=1.50$ & \textbf{1.000} & \textbf{1.000} & \textbf{0.985} \\
 & $\text{max\_factor}=2.00$ & \textbf{1.000} & \textbf{1.000} & 0.955 \\
\midrule
$\epsilon=0.2$ & $\text{max\_factor}=1.25$ & \textbf{1.000} & 0.996 & 0.850 \\
 & $\text{max\_factor}=1.50$ & \textbf{1.000} & 0.998 & \textbf{0.852} \\
 & $\text{max\_factor}=2.00$ & \textbf{1.000} & \textbf{1.000} & 0.839 \\
\midrule
$\epsilon=0.4$ & $\text{max\_factor}=1.25$ & \textbf{0.999} & 0.988 & 0.768 \\
 & $\text{max\_factor}=1.50$ & \textbf{0.999} & \textbf{0.992} & \textbf{0.790} \\
 & $\text{max\_factor}=2.00$ & 0.998 & 0.990 & 0.778 \\
\bottomrule
\end{tabular}
\caption{Robustness of MirrorMark under insertion attacks with different CABS parameters, where $m=2$, $H=12$, and the number of tokens is 300.}
\label{tab:cabs-insert-final}
\end{table*}

\begin{table*}[htbp]
\centering
\small
\begin{tabular}{l l c c c}
\toprule
 & Setting & AUC & TPR@1\%FPR & Bit Accuracy \\
\midrule
\multicolumn{5}{l}{\textbf{Varying $f$ (with $W=4$, max\_factor=1.5)}} \\
$\epsilon=0.0$ & $f=1$ & \textbf{1.000} & 0.998 & 0.939 \\
 & $f=2$ & \textbf{1.000} & 0.998 & 0.952 \\
 & $f=3$ & \textbf{1.000} & \textbf{1.000} & \textbf{0.985} \\
 & $f=4$ & \textbf{1.000} & 0.998 & 0.957 \\
\midrule
$\epsilon=0.2$ & $f=1$ & 0.998 & \textbf{0.988} & 0.683 \\
 & $f=2$ & 0.997 & 0.984 & \textbf{0.702} \\
 & $f=3$ & \textbf{0.999} & \textbf{0.988} & 0.700 \\
 & $f=4$ & 0.998 & 0.982 & 0.700 \\
\midrule
$\epsilon=0.4$ & $f=1$ & 0.939 & \textbf{0.604} & 0.464 \\
 & $f=2$ & 0.942 & 0.566 & 0.471 \\
 & $f=3$ & \textbf{0.946} & 0.566 & 0.472 \\
 & $f=4$ & 0.945 & 0.584 & \textbf{0.474} \\
\midrule
\multicolumn{5}{l}{\textbf{Varying $W$ (with $f=3$, max\_factor=1.5)}} \\
$\epsilon=0.0$ & $W=1$ & \textbf{1.000} & 0.998 & 0.945 \\
 & $W=2$ & \textbf{1.000} & 0.998 & 0.943 \\
 & $W=3$ & \textbf{1.000} & 0.998 & 0.946 \\
 & $W=4$ & \textbf{1.000} & \textbf{1.000} & \textbf{0.985} \\
 & $W=5$ & \textbf{1.000} & \textbf{1.000} & 0.944 \\
\midrule
$\epsilon=0.2$ & $W=1$ & 0.998 & 0.980 & 0.686 \\
 & $W=2$ & 0.999 & \textbf{0.990} & 0.689 \\
 & $W=3$ & \textbf{1.000} & 0.986 & \textbf{0.707} \\
 & $W=4$ & 0.999 & 0.988 & 0.700 \\
 & $W=5$ & 0.999 & 0.988 & 0.700 \\
\midrule
$\epsilon=0.4$ & $W=1$ & \textbf{0.956} & 0.580 & 0.476 \\
 & $W=2$ & 0.948 & 0.564 & 0.465 \\
 & $W=3$ & 0.947 & \textbf{0.600} & \textbf{0.481} \\
 & $W=4$ & 0.946 & 0.566 & 0.472 \\
 & $W=5$ & 0.949 & 0.592 & 0.453 \\
\midrule
\multicolumn{5}{l}{\textbf{Varying max\_factor (with $f=3$, $W=4$)}} \\
$\epsilon=0.0$ & $\text{max\_factor}=1.25$ & \textbf{1.000} & \textbf{1.000} & 0.948 \\
 & $\text{max\_factor}=1.50$ & \textbf{1.000} & \textbf{1.000} & \textbf{0.985} \\
 & $\text{max\_factor}=2.00$ & \textbf{1.000} & \textbf{1.000} & 0.955 \\
\midrule
$\epsilon=0.2$ & $\text{max\_factor}=1.25$ & \textbf{0.999} & 0.986 & 0.682 \\
 & $\text{max\_factor}=1.50$ & \textbf{0.999} & 0.988 & \textbf{0.700} \\
 & $\text{max\_factor}=2.00$ & \textbf{0.999} & \textbf{0.990} & 0.693 \\
\midrule
$\epsilon=0.4$ & $\text{max\_factor}=1.25$ & \textbf{0.954} & \textbf{0.606} & 0.464 \\
 & $\text{max\_factor}=1.50$ & 0.946 & 0.566 & \textbf{0.472} \\
 & $\text{max\_factor}=2.00$ & 0.950 & 0.558 & 0.464 \\
\bottomrule
\end{tabular}
\caption{Robustness of MirrorMark under deletion attacks with different CABS parameters, where $m=2$, $H=12$, and the number of tokens is 300.}
\label{tab:cabs-del-final}
\end{table*}

\begin{table*}[htbp]
\centering
\small
\begin{tabular}{l l c c c}
\toprule
 & Setting & AUC & TPR@1\%FPR & Bit Accuracy \\
\midrule
\multicolumn{5}{l}{\textbf{Varying $f$ (with $W=4$, max\_factor=1.5)}} \\
$\epsilon=0.0$ & $f=1$ & \textbf{1.000} & 0.998 & 0.939 \\
 & $f=2$ & \textbf{1.000} & 0.998 & 0.952 \\
 & $f=3$ & \textbf{1.000} & \textbf{1.000} & \textbf{0.985} \\
 & $f=4$ & \textbf{1.000} & 0.998 & 0.957 \\
\midrule
$\epsilon=0.2$ & $f=1$ & \textbf{0.998} & 0.970 & 0.684 \\
 & $f=2$ & 0.997 & \textbf{0.984} & 0.710 \\
 & $f=3$ & \textbf{0.998} & \textbf{0.984} & \textbf{0.723} \\
 & $f=4$ & 0.997 & \textbf{0.984} & 0.709 \\
\midrule
$\epsilon=0.4$ & $f=1$ & \textbf{0.960} & 0.646 & 0.492 \\
 & $f=2$ & 0.945 & 0.644 & \textbf{0.504} \\
 & $f=3$ & 0.948 & \textbf{0.648} & 0.499 \\
 & $f=4$ & 0.957 & 0.622 & 0.486 \\
\midrule
\multicolumn{5}{l}{\textbf{Varying $W$ (with $f=3$, max\_factor=1.5)}} \\
$\epsilon=0.0$ & $W=1$ & \textbf{1.000} & 0.998 & 0.945 \\
 & $W=2$ & \textbf{1.000} & 0.998 & 0.943 \\
 & $W=3$ & \textbf{1.000} & 0.998 & 0.946 \\
 & $W=4$ & \textbf{1.000} & \textbf{1.000} & \textbf{0.985} \\
 & $W=5$ & \textbf{1.000} & \textbf{1.000} & 0.944 \\
\midrule
$\epsilon=0.2$ & $W=1$ & 0.998 & 0.968 & 0.703 \\
 & $W=2$ & \textbf{0.999} & \textbf{0.990} & 0.709 \\
 & $W=3$ & 0.997 & 0.984 & 0.709 \\
 & $W=4$ & 0.998 & 0.984 & \textbf{0.723} \\
 & $W=5$ & 0.997 & 0.984 & 0.690 \\
\midrule
$\epsilon=0.4$ & $W=1$ & 0.942 & 0.638 & \textbf{0.503} \\
 & $W=2$ & 0.933 & 0.574 & 0.494 \\
 & $W=3$ & \textbf{0.951} & 0.632 & 0.490 \\
 & $W=4$ & 0.948 & \textbf{0.648} & 0.499 \\
 & $W=5$ & 0.947 & 0.590 & 0.489 \\
\midrule
\multicolumn{5}{l}{\textbf{Varying max\_factor (with $f=3$, $W=4$)}} \\
$\epsilon=0.0$ & $\text{max\_factor}=1.25$ & \textbf{1.000} & \textbf{1.000} & 0.948 \\
 & $\text{max\_factor}=1.50$ & \textbf{1.000} & \textbf{1.000} & \textbf{0.985} \\
 & $\text{max\_factor}=2.00$ & \textbf{1.000} & \textbf{1.000} & 0.955 \\
\midrule
$\epsilon=0.2$ & $\text{max\_factor}=1.25$ & \textbf{0.998} & 0.978 & 0.708 \\
 & $\text{max\_factor}=1.50$ & \textbf{0.998} & \textbf{0.984} & \textbf{0.723} \\
 & $\text{max\_factor}=2.00$ & \textbf{0.998} & 0.982 & 0.709 \\
\midrule
$\epsilon=0.4$ & $\text{max\_factor}=1.25$ & 0.949 & 0.610 & 0.486 \\
 & $\text{max\_factor}=1.50$ & 0.948 & \textbf{0.648} & \textbf{0.499} \\
 & $\text{max\_factor}=2.00$ & \textbf{0.951} & 0.624 & 0.482 \\
\bottomrule
\end{tabular}
\caption{Robustness of MirrorMark under substitution attacks with different CABS parameters, where $m=2$, $H=12$, and the number of tokens is 300.}
\label{tab:cabs-sub-final}
\end{table*}

\subsubsection{The Effect of Position Allocation Schedulers on Watermarking Schemes}
\label{cabs_ablation_scheduler}

To disentangle the contribution of mod-1 mirroring from that of position allocation, we conduct an ablation that systematically combines different position schedulers with different watermarking schemes.
In particular, we incorporate the position schedulers used in MPAC and RSBH, which we denote as \emph{NaiveHash} and \emph{DPHash}, respectively.
NaiveHash (MPAC, Section~3.2) seeds a PRF using the previous $h$ tokens to randomly select a position, whereas DPHash (RSBH, Section~4.2) constructs a balanced token-to-segment mapping through a secret-key shuffle followed by a dynamic programming procedure.

Because the DPHash table released in the official implementation of RSBH is constructed with $h=1$, we evaluate performance under this setting in Fig.~\ref{fig:allocation2}.
In addition, since our main experiments use $h=4$ by default unless otherwise noted, we also report results under $h=4$ in Fig.~\ref{fig:allocation5}.
For both settings, we report the Gini coefficient\footnote{\url{https://en.wikipedia.org/wiki/Gini_coefficient}} in Fig.~\ref{fig:gini}, which quantifies how balanced the token allocation is across positions, where lower values indicate more balanced allocation.

Across all configurations, CABS consistently achieves significantly lower Gini scores, approaching zero, indicating near-uniform token allocation across positions.
This balanced allocation leads to substantial improvements in detectability for MirrorMark under both Tour-Bayes and Gumbel-max, with CABS outperforming NaiveHash and DPHash across all detection metrics.
In contrast, for MPAC, the AUC and TPR@1\%FPR of CABS are comparable to those obtained with NaiveHash and DPHash, while the bit accuracy of CABS is only slightly higher.
This behavior is expected, as balanced allocation primarily improves the reliability of message decoding.

The difference between MirrorMark and MPAC stems from how positional evidence is aggregated.
MirrorMark aggregates evidence from all positions, making it highly sensitive to positional imbalance.
For example, consider a text of 100 tokens distributed across four positions as 85--5--5--5.
In watermark text, the dominant position provides a strong signal for the correct message, whereas the remaining lightly populated positions contribute mostly noise.
When combined in the final score, this noisy evidence dilutes the strong signal, making watermark and non-watermark score distributions harder to separate.

In contrast, MPAC is robust under the same allocation.
The dominant position overwhelmingly votes for the correct message in watermark text, while non-watermark text remains approximately balanced across message candidates.
Since MPAC retains only the maximum vote per position and aggregates these maxima, lightly populated positions contribute little and do not introduce harmful noise.
As a result, the detectability of MPAC remains stable even under highly uneven token allocation.

\begin{figure*}[htbp]
  \centering
\includegraphics[width=0.7\textwidth]{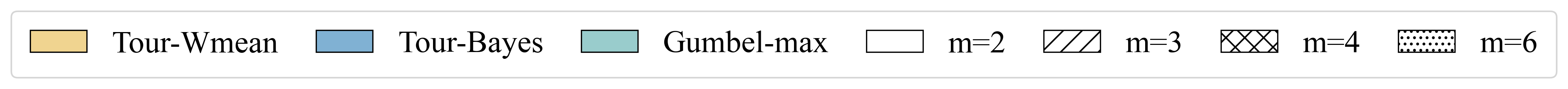}

  \vspace{-0.3em}

  \begin{minipage}[t]{0.32\textwidth}
    \centering
    \includegraphics[width=\linewidth]{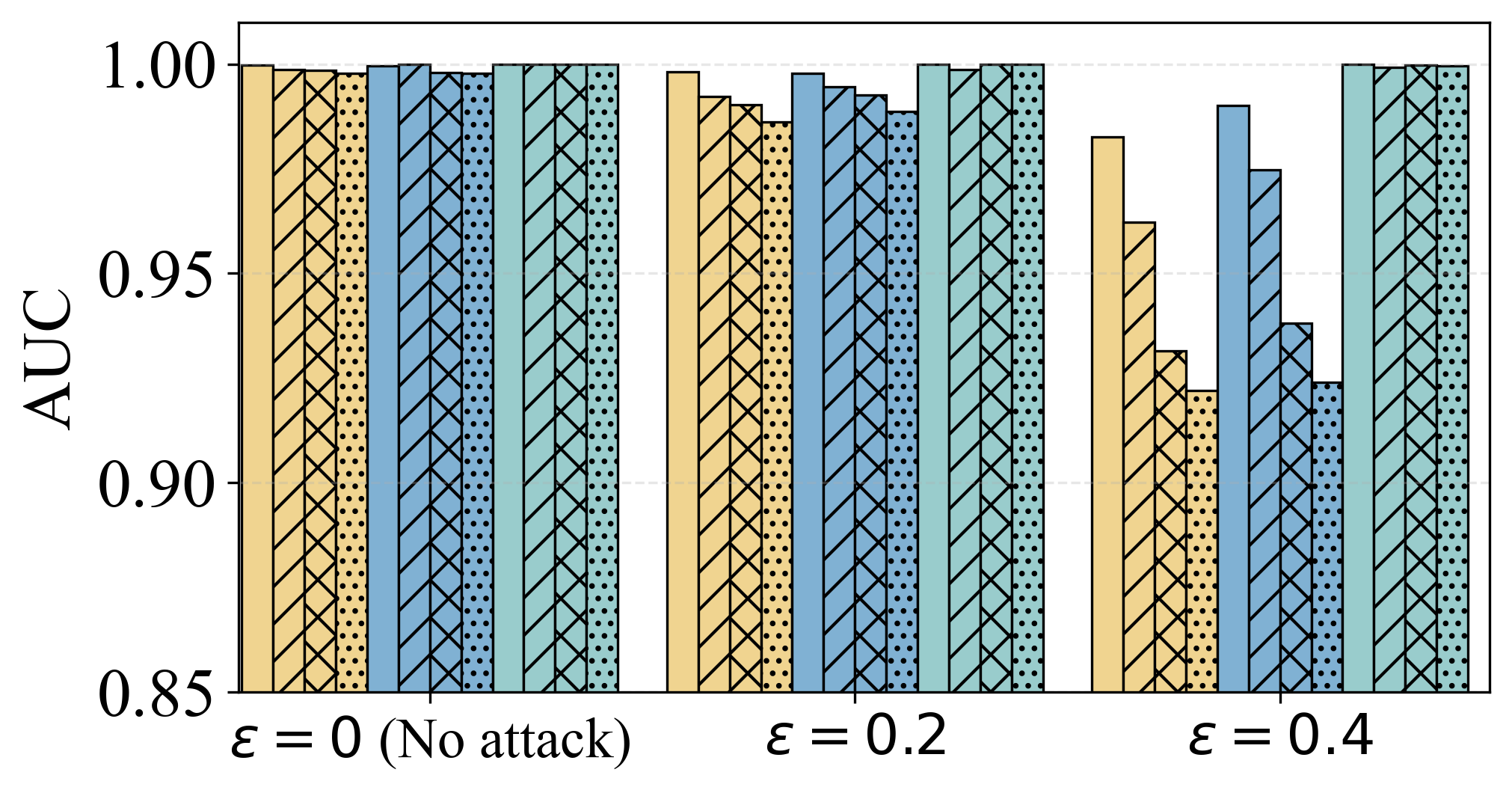}
\end{minipage}\hfill
  \begin{minipage}[t]{0.32\textwidth}
    \centering
    \includegraphics[width=\linewidth]{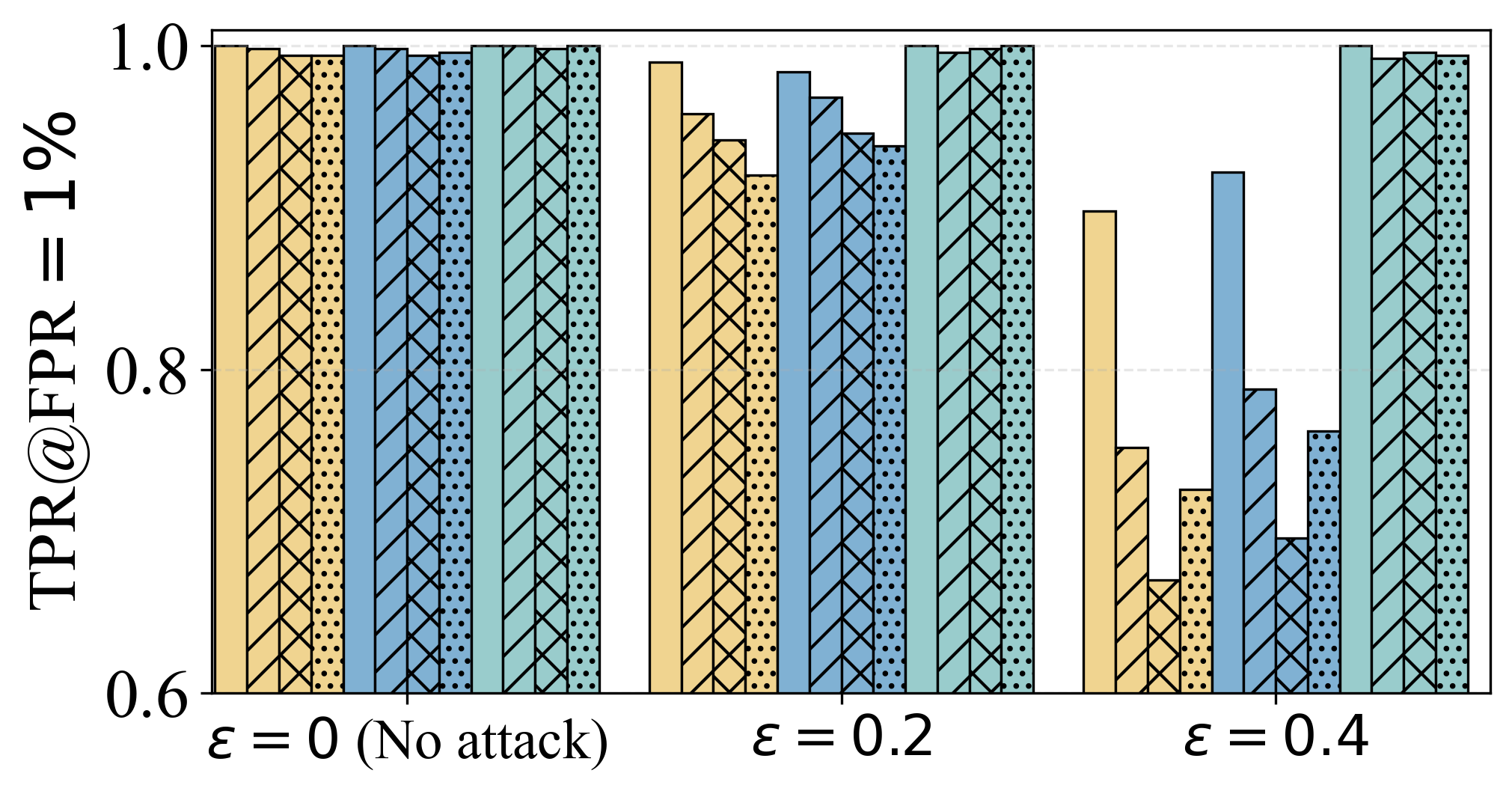}
  \end{minipage}\hfill
  \begin{minipage}[t]{0.32\textwidth}
    \centering
    \includegraphics[width=\linewidth]{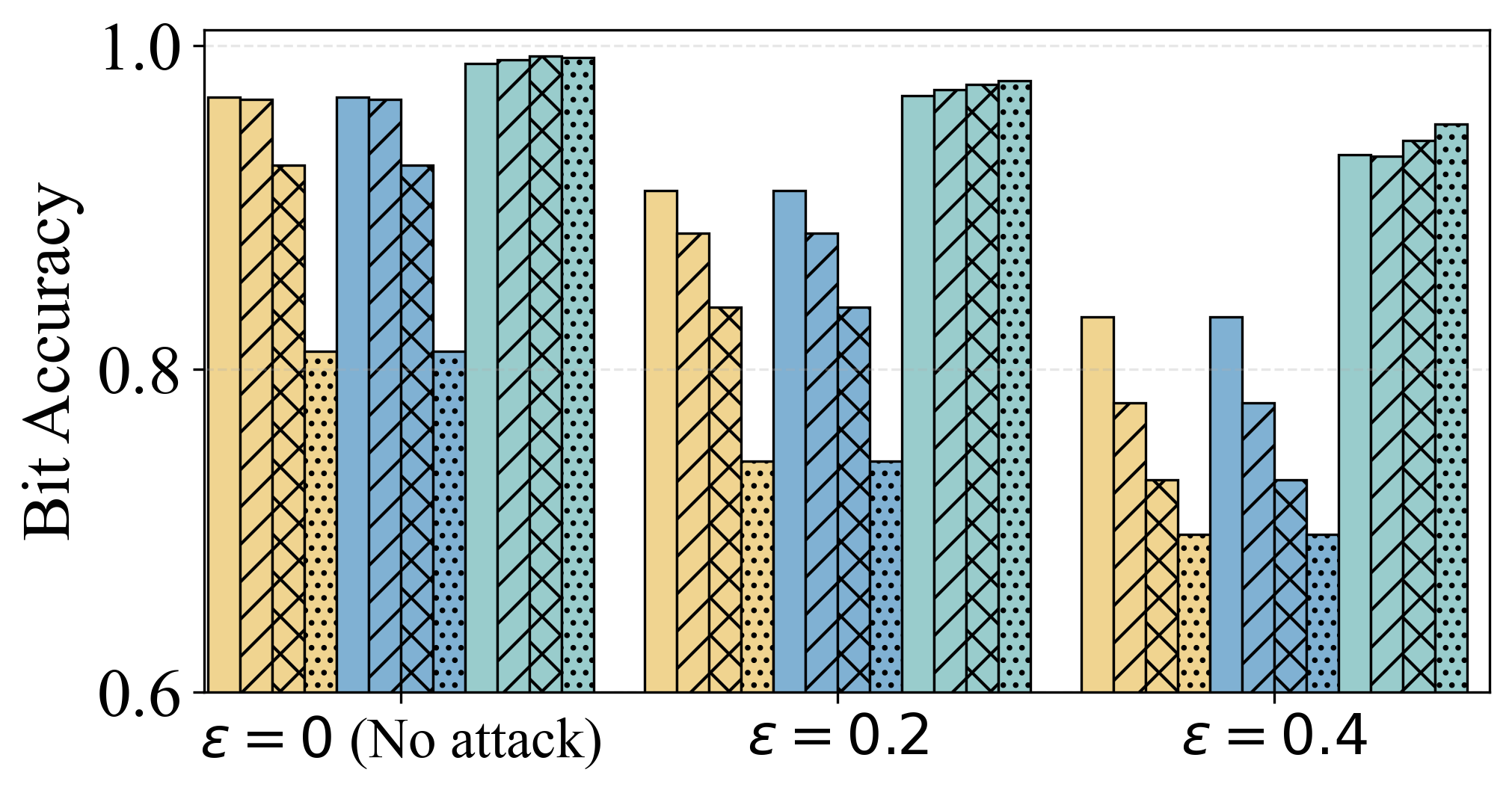}
  \end{minipage}
\caption{Detectability of MirrorMark against copy-paste attacks with 36 bits embedded in 400 tokens, where the edit fraction $\epsilon\in\{0, 0.2, 0.4\}$. To embed 36 bits, different $m$ applies for various number of positions $H$, i.e., $m\in\{2, 3, 4, 6\}$ is respectively corresponding to $H\in\{18, 12, 9, 6\}$.} \label{fig:cp_m}
\end{figure*}

\begin{figure*}[htbp]
  \centering
\includegraphics[width=0.7\textwidth]{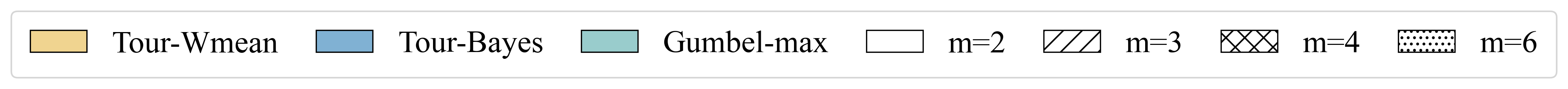}

  \vspace{-0.3em}

  \begin{minipage}[t]{0.32\textwidth}
    \centering
    \includegraphics[width=\linewidth]{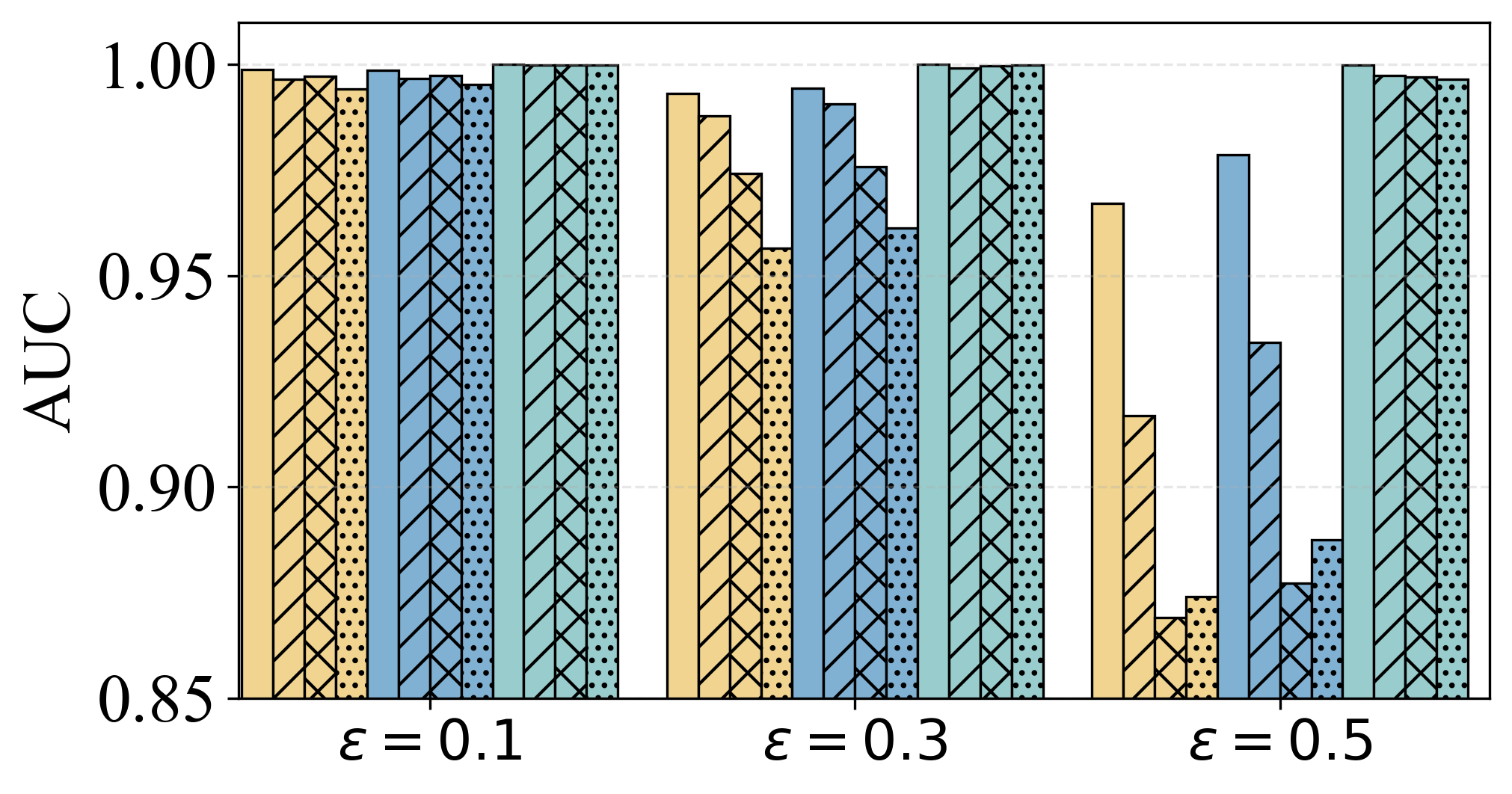}
\end{minipage}\hfill
  \begin{minipage}[t]{0.32\textwidth}
    \centering
    \includegraphics[width=\linewidth]{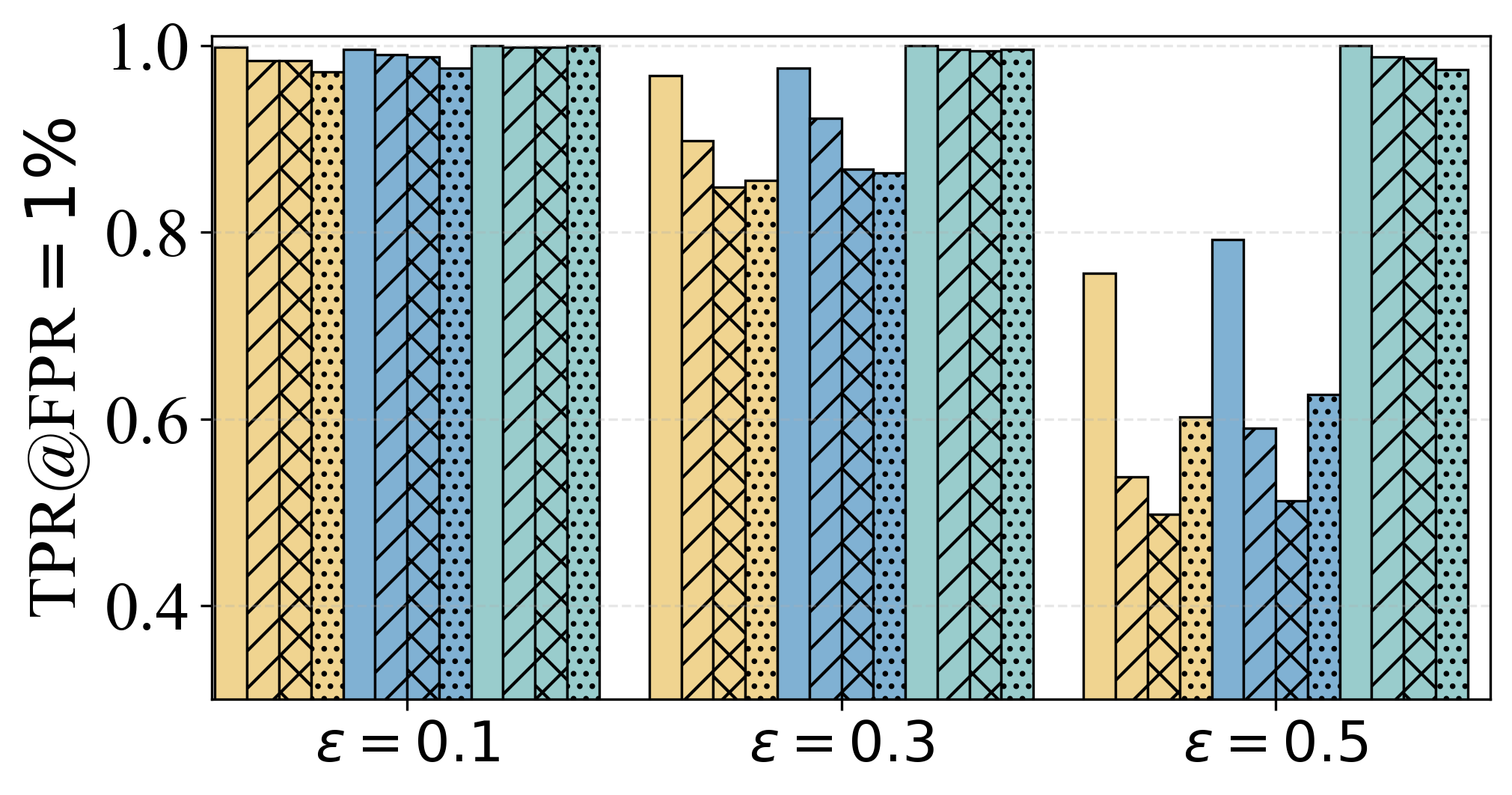}
  \end{minipage}\hfill
  \begin{minipage}[t]{0.32\textwidth}
    \centering
    \includegraphics[width=\linewidth]{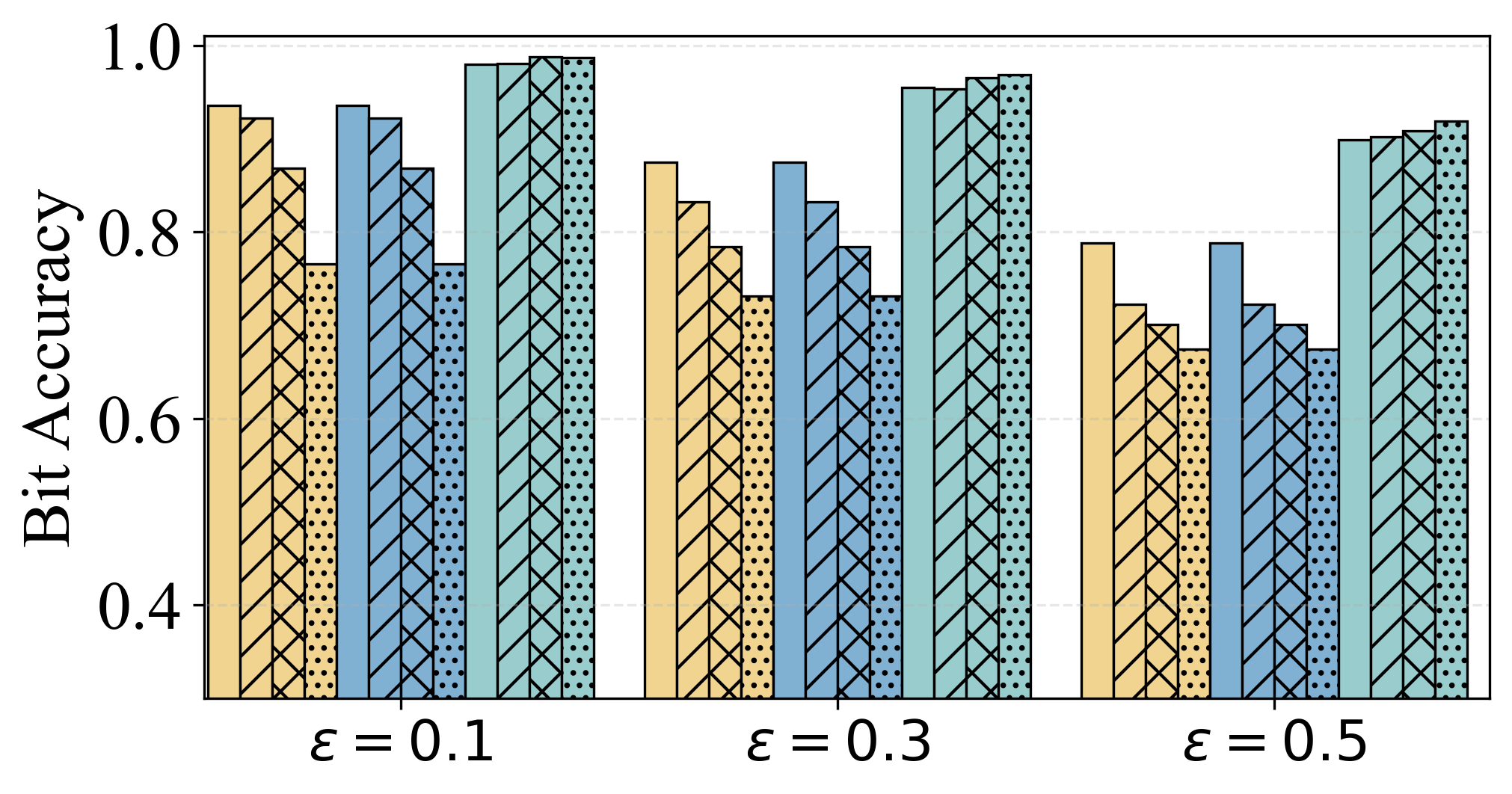}
  \end{minipage}
\caption{Detectability of MirrorMark against copy-paste attacks with 36 bits embedded in 400 tokens, where the edit fraction $\epsilon\in\{0.1, 0.3, 0.5\}$. To embed 36 bits, different $m$ applies for various number of positions, i.e., the number of positions for $m\in\{2, 3, 4, 6\}$ is respectively, 18, 12, 9, 6.} \label{fig:cp_m135}
\end{figure*}

\begin{figure*}[htbp]
  \centering
  \captionsetup[subfigure]{justification=centering, font=small}

  \begin{subfigure}[t]{0.48\textwidth}
    \centering
    \includegraphics[width=0.8\linewidth]{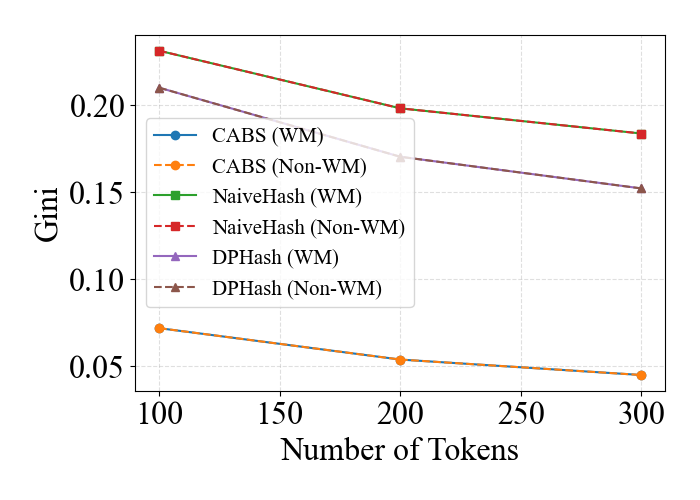}
    \caption{$h$=1}
  \end{subfigure}
  \hfill
  \begin{subfigure}[t]{0.48\textwidth}
    \centering
    \includegraphics[width=0.8\linewidth]{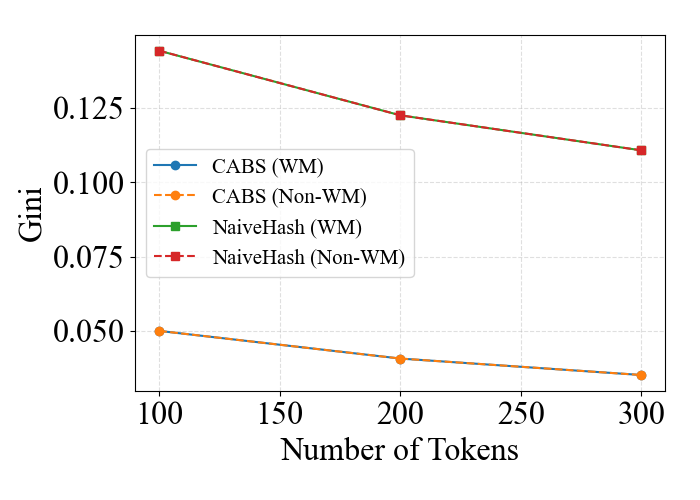}
    \caption{$h$=4}
  \end{subfigure}

  \caption{Comparison of token-allocation balance between different position scheduler. The setting is $m=2$, $H=12$, and the number of tokens is 300. The Gini coefficient is significantly lower (more balanced allocation) when using CABS, showing that CABS reduces position-allocation skew and improves uniformity.}
  \label{fig:gini}
\end{figure*}

\begin{figure*}[htbp]
  \centering

\includegraphics[width=0.7\textwidth]{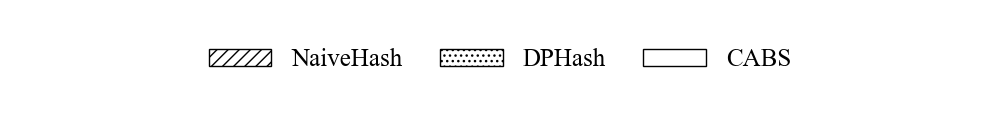}

  \begin{minipage}[t]{0.32\textwidth}
    \centering
    \subcaptionbox{$h$=1, 100 tokens\label{fig:r1c1}}{
      \includegraphics[width=\linewidth]{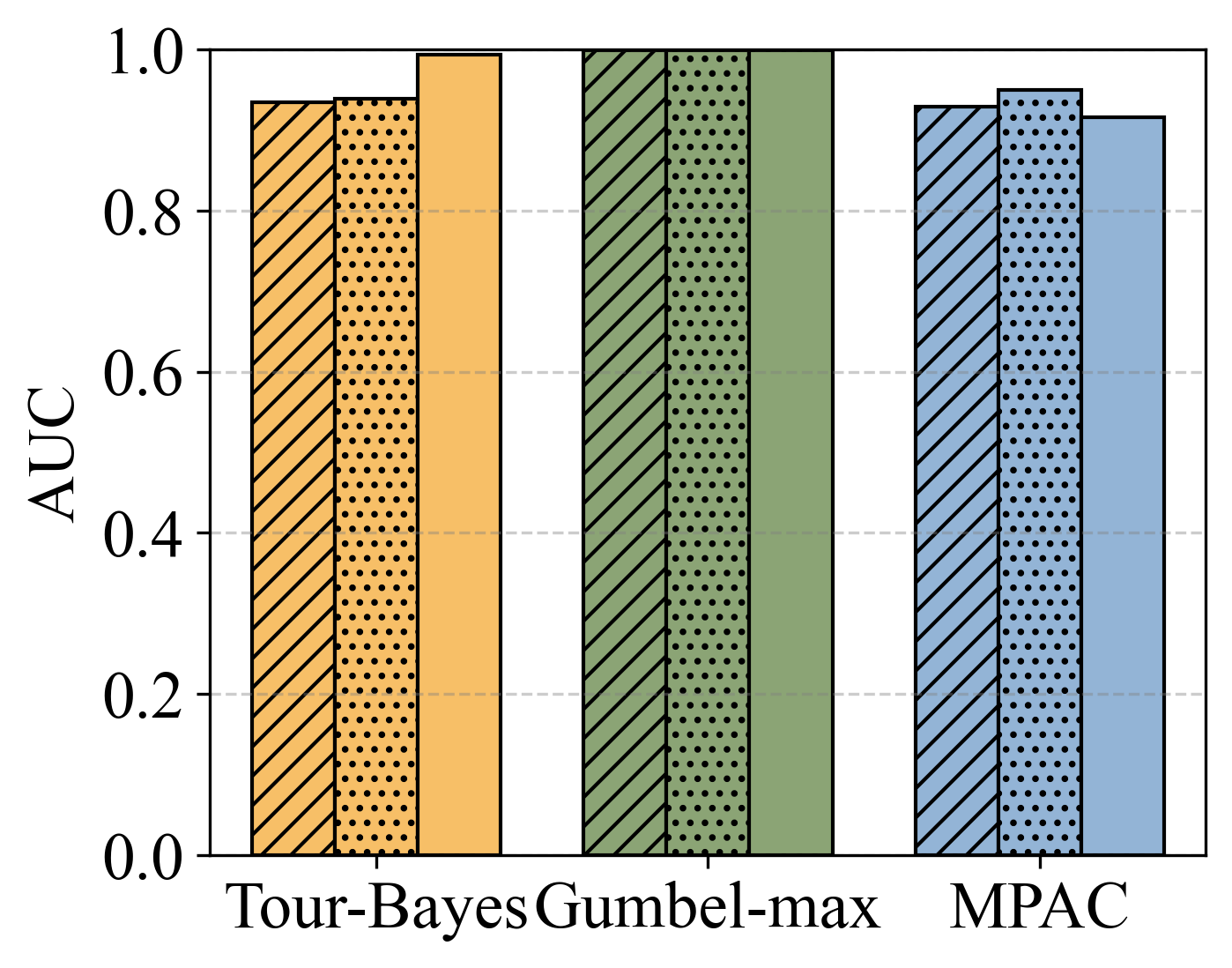}
    }
  \end{minipage}\hfill
  \begin{minipage}[t]{0.32\textwidth}
    \centering
    \subcaptionbox{$h$=1, 100 tokens}{
      \includegraphics[width=\linewidth]{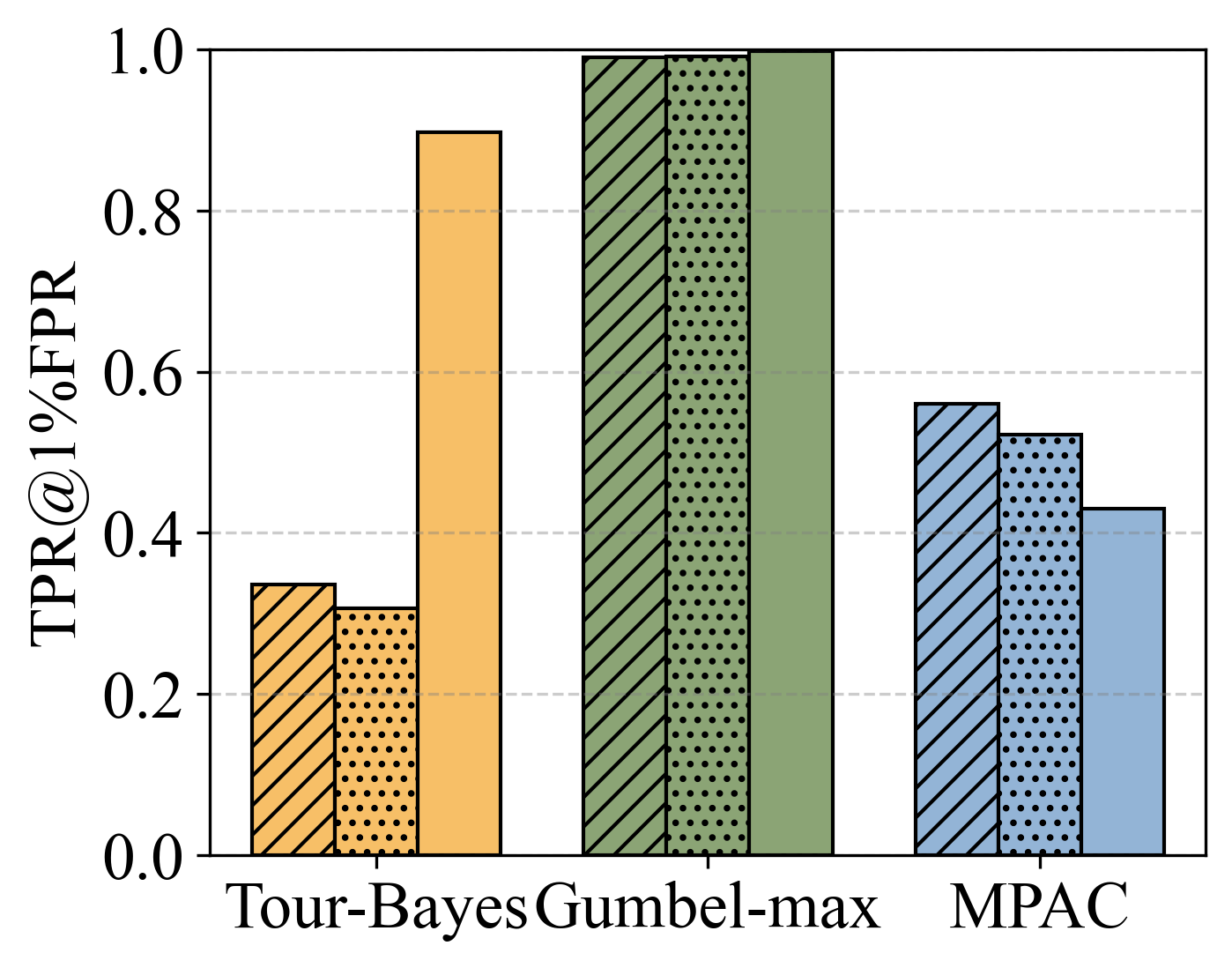}
    }
  \end{minipage}\hfill
  \begin{minipage}[t]{0.32\textwidth}
    \centering
    \subcaptionbox{$h$=1, 100 tokens}{
      \includegraphics[width=\linewidth]{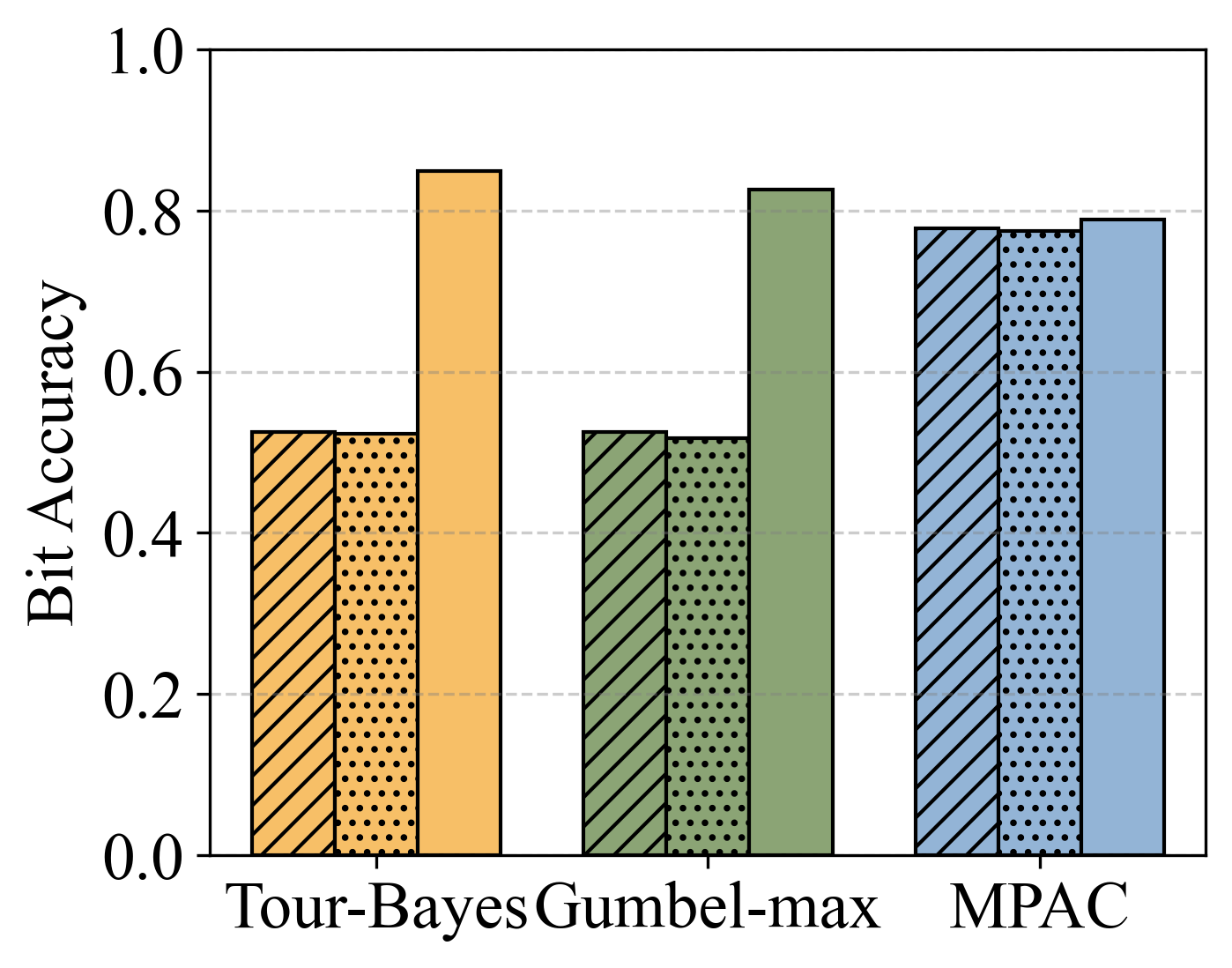}
    }
  \end{minipage}

  \begin{minipage}[t]{0.32\textwidth}
    \centering
    \subcaptionbox{$h$=1, 200 tokens}{
      \includegraphics[width=\linewidth]{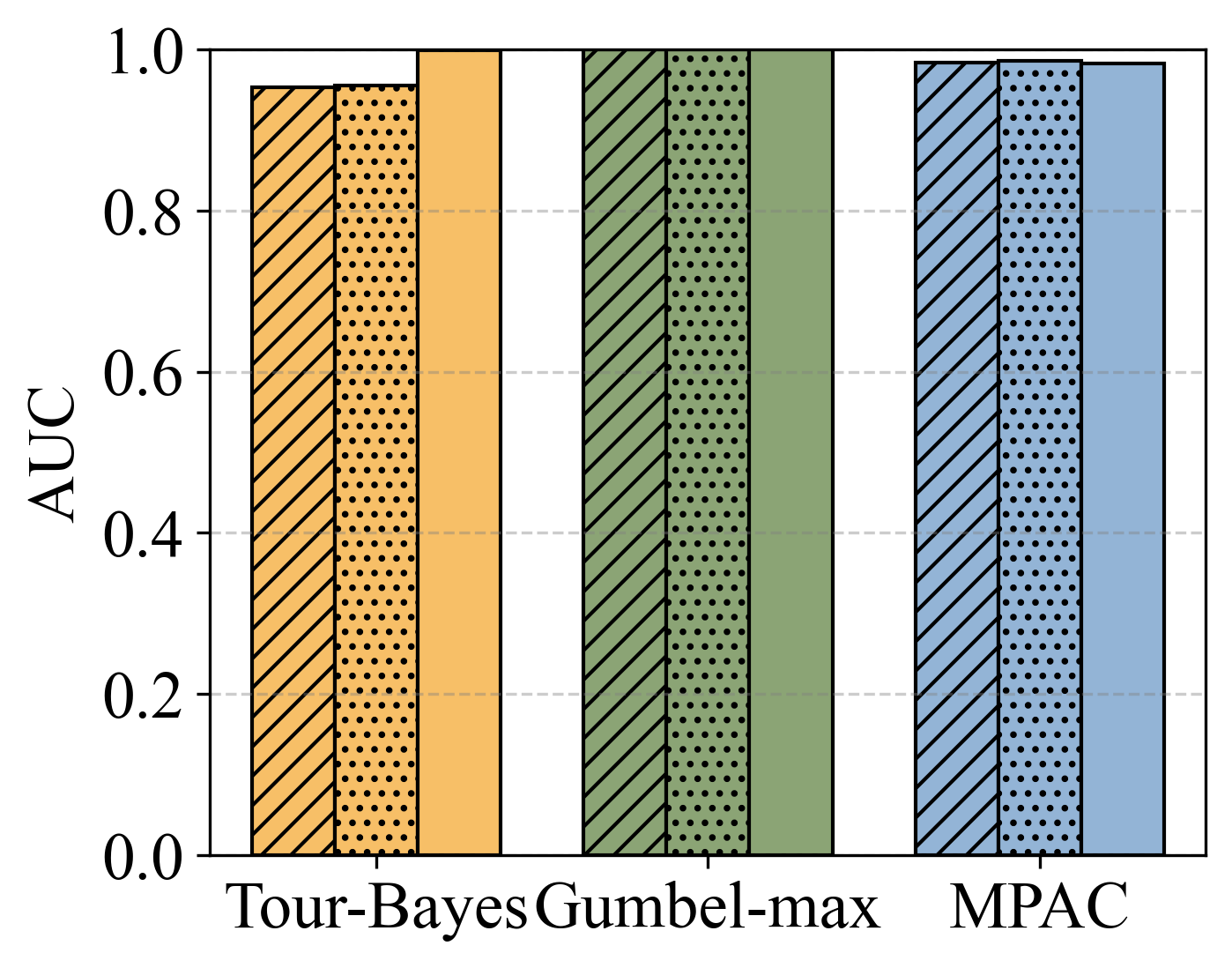}
    }
  \end{minipage}\hfill
  \begin{minipage}[t]{0.32\textwidth}
    \centering
    \subcaptionbox{ $h$=1, 200 tokens}{
      \includegraphics[width=\linewidth]{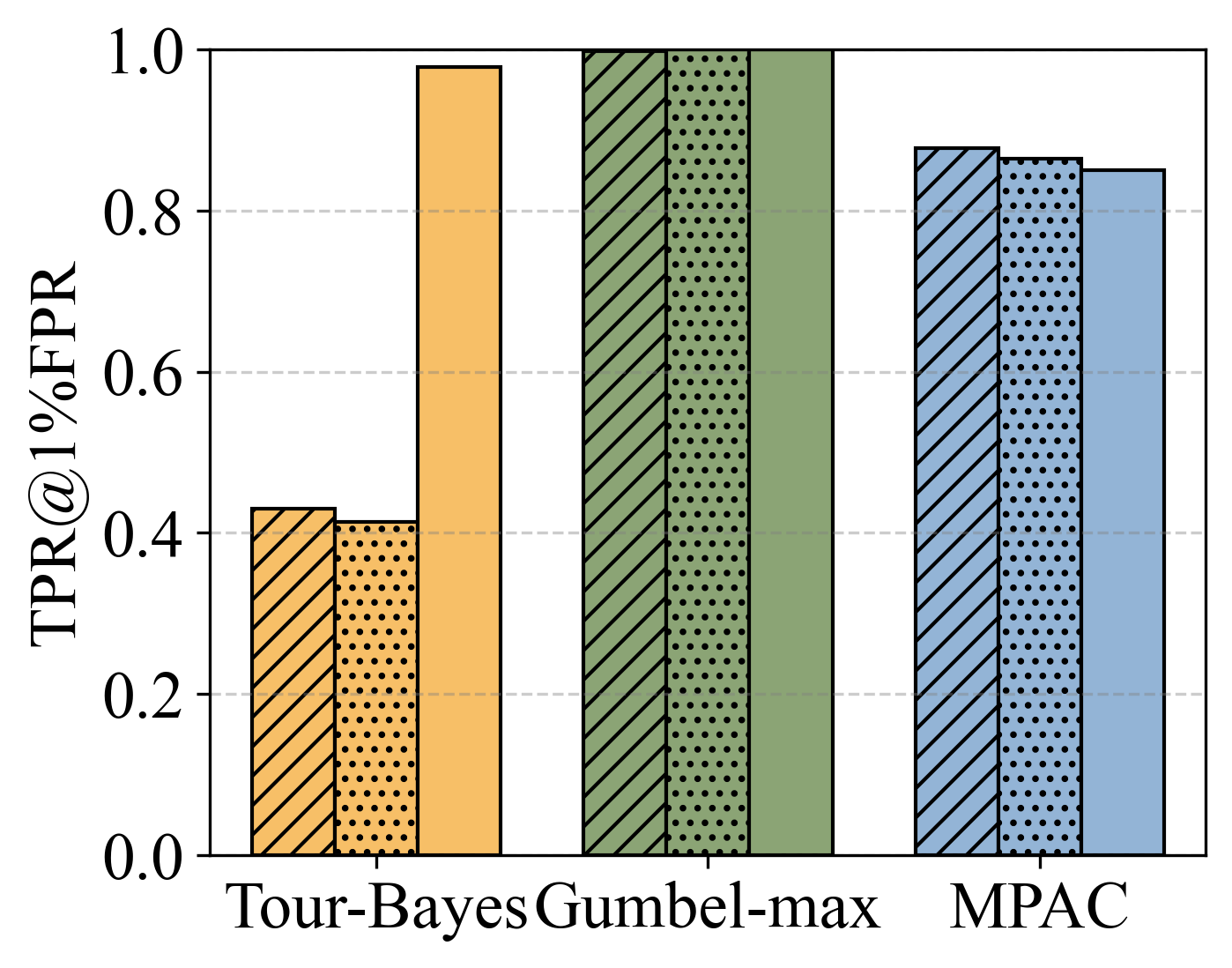}
    }
  \end{minipage}\hfill
  \begin{minipage}[t]{0.32\textwidth}
    \centering
    \subcaptionbox{$h$=1, 200 tokens}{
      \includegraphics[width=\linewidth]{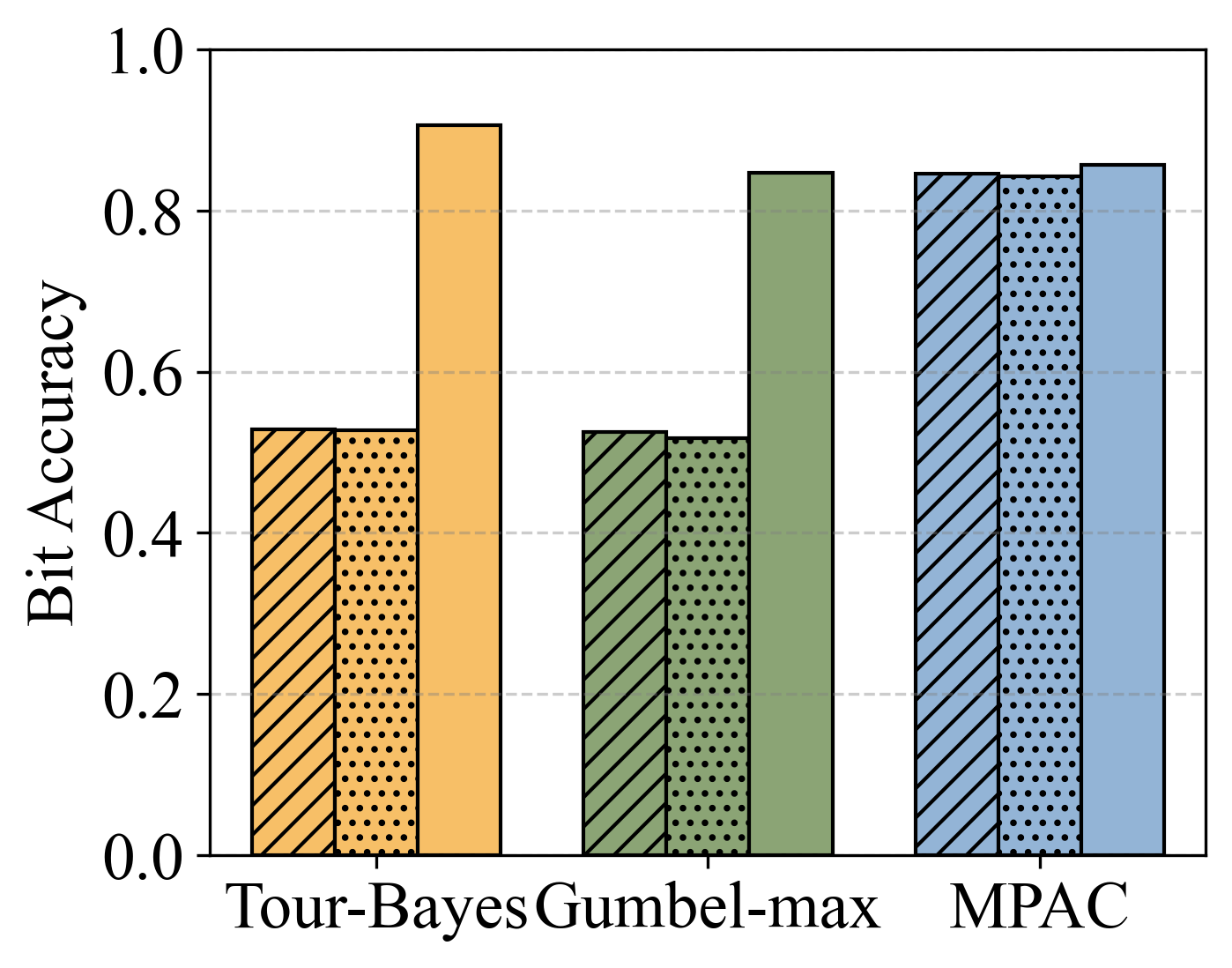}
    }
  \end{minipage}

\begin{minipage}[t]{0.32\textwidth}
    \centering
    \subcaptionbox{ $h$=1, 300 tokens}{
      \includegraphics[width=\linewidth]{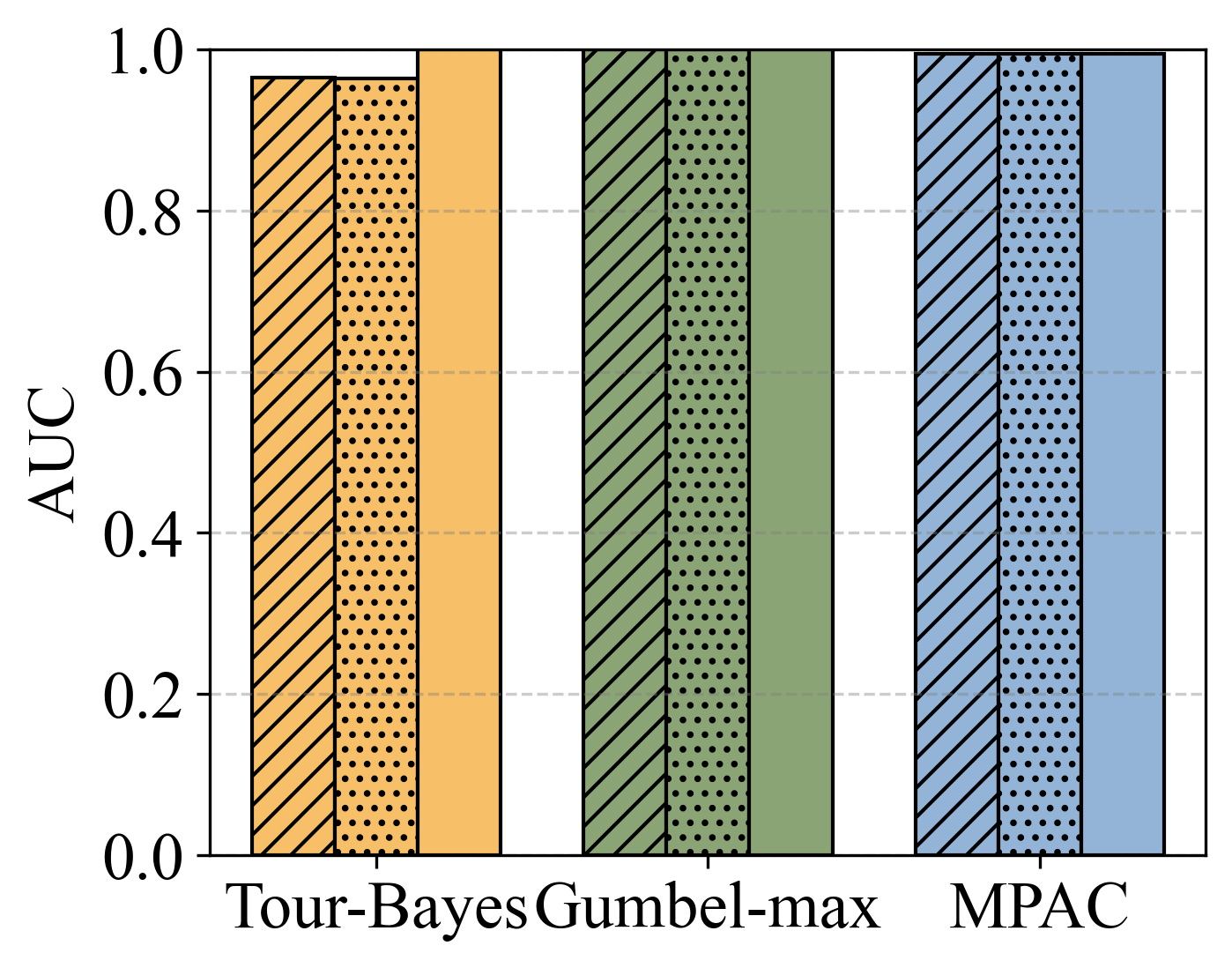}
    }
  \end{minipage}\hfill
  \begin{minipage}[t]{0.32\textwidth}
    \centering
    \subcaptionbox{$h$=1, 300 tokens}{
      \includegraphics[width=\linewidth]{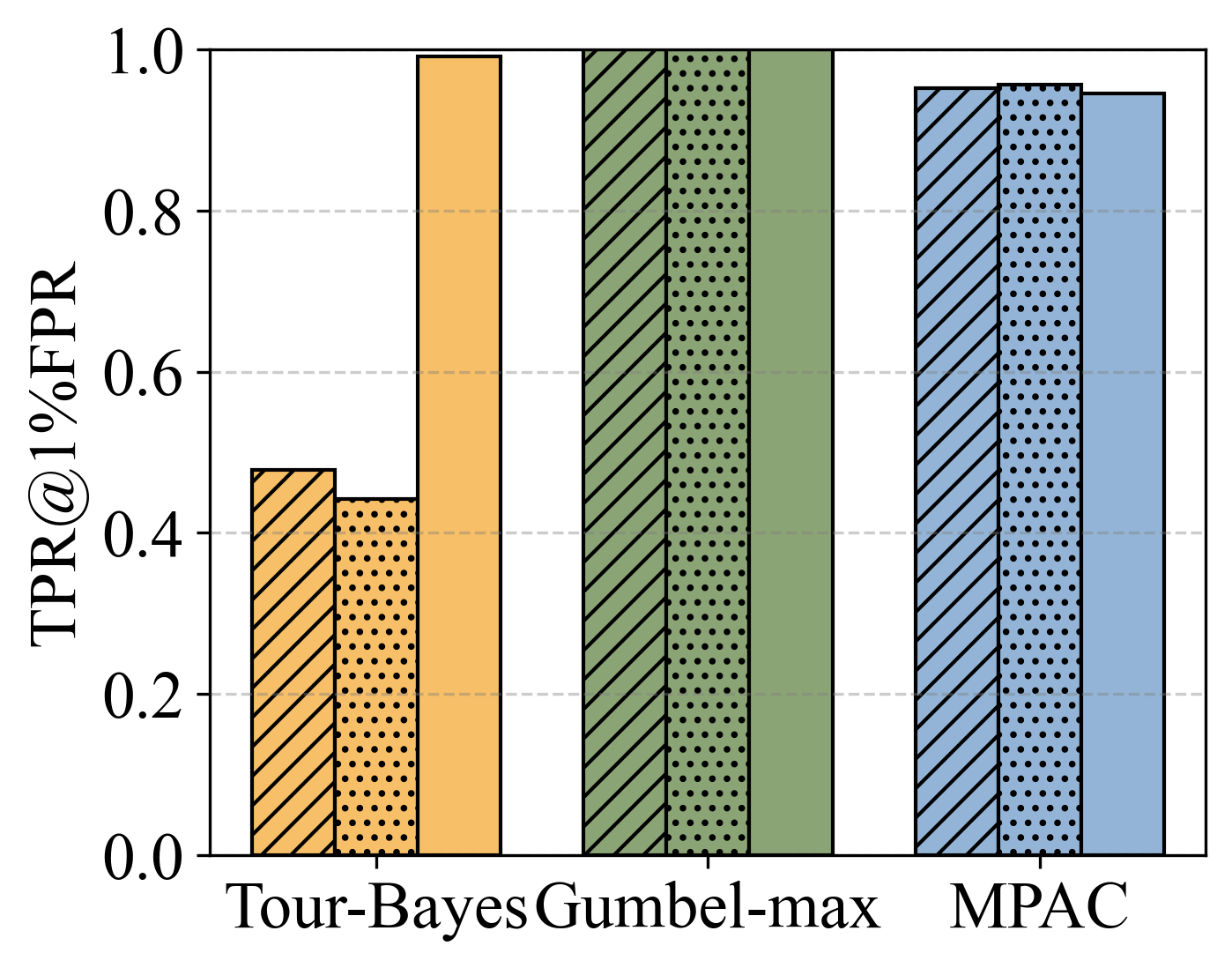}
    }
  \end{minipage}\hfill
  \begin{minipage}[t]{0.32\textwidth}
    \centering
    \subcaptionbox{ $h$=1, 300 tokens}{
      \includegraphics[width=\linewidth]{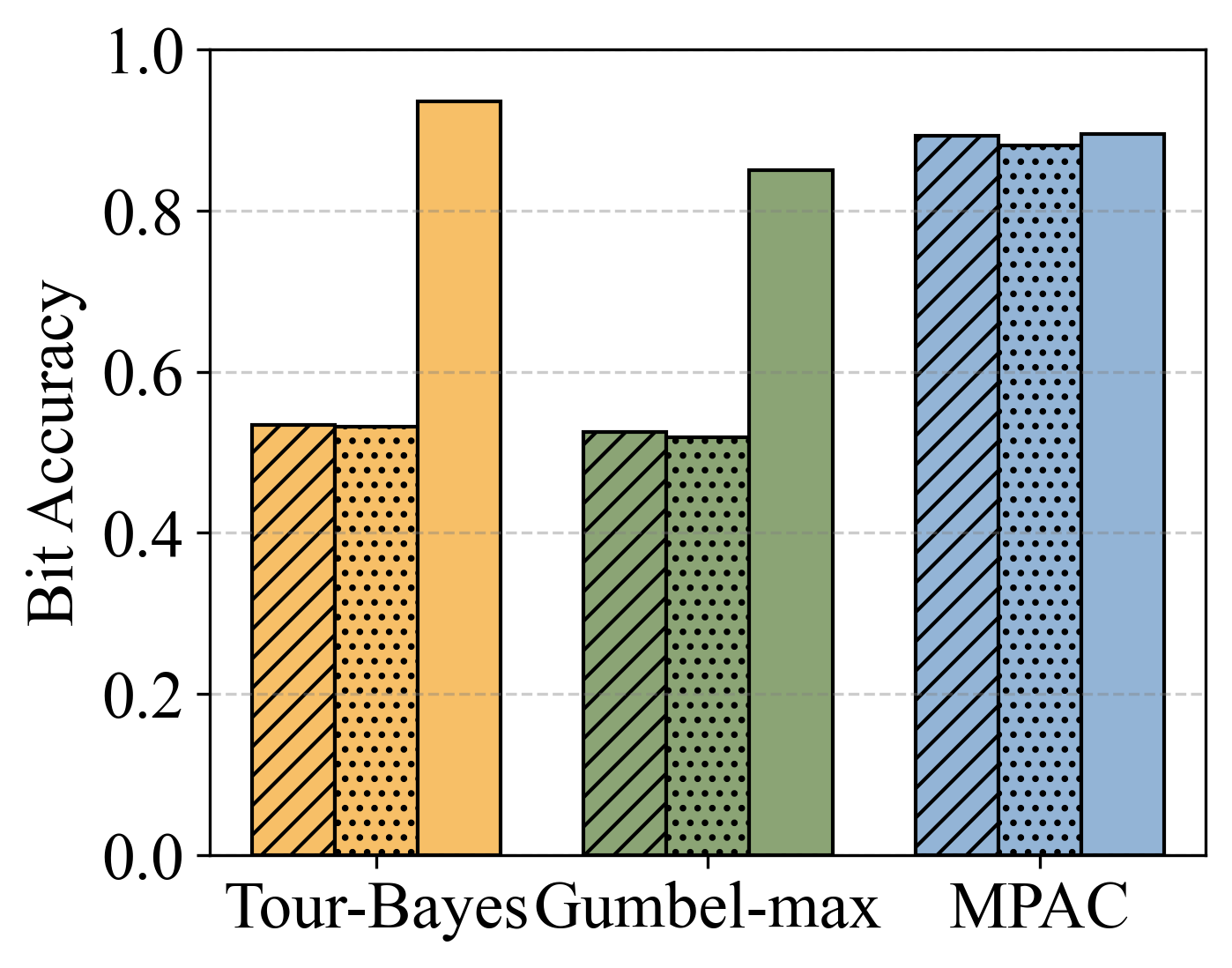}
    }
  \end{minipage}
  \caption{Detectability of MirrorMark with length of n-gram $h$=1. The setting is $m=2$ and $H=12$.}
  \label{fig:allocation2}
\end{figure*}

\begin{figure*}[htbp]
  \centering

\includegraphics[width=0.7\textwidth]{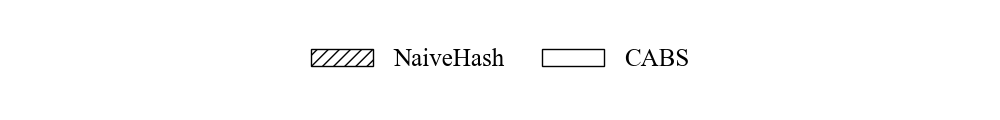}

  \begin{minipage}[t]{0.32\textwidth}
    \centering
    \subcaptionbox{$h$=4, 200 tokens}{
      \includegraphics[width=\linewidth]{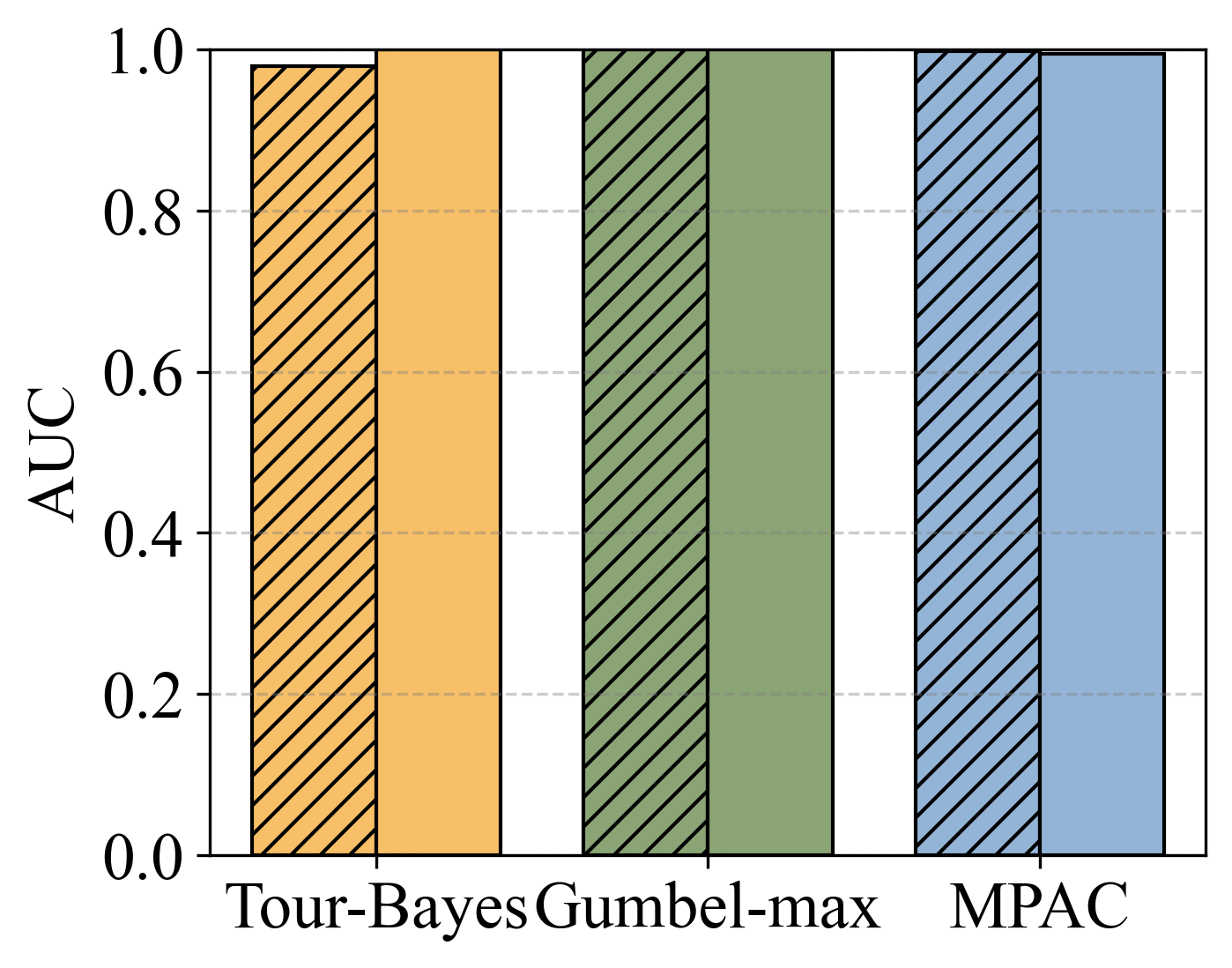}
    }
  \end{minipage}\hfill
  \begin{minipage}[t]{0.32\textwidth}
    \centering
    \subcaptionbox{$h$=4, 200 tokens}{
      \includegraphics[width=\linewidth]{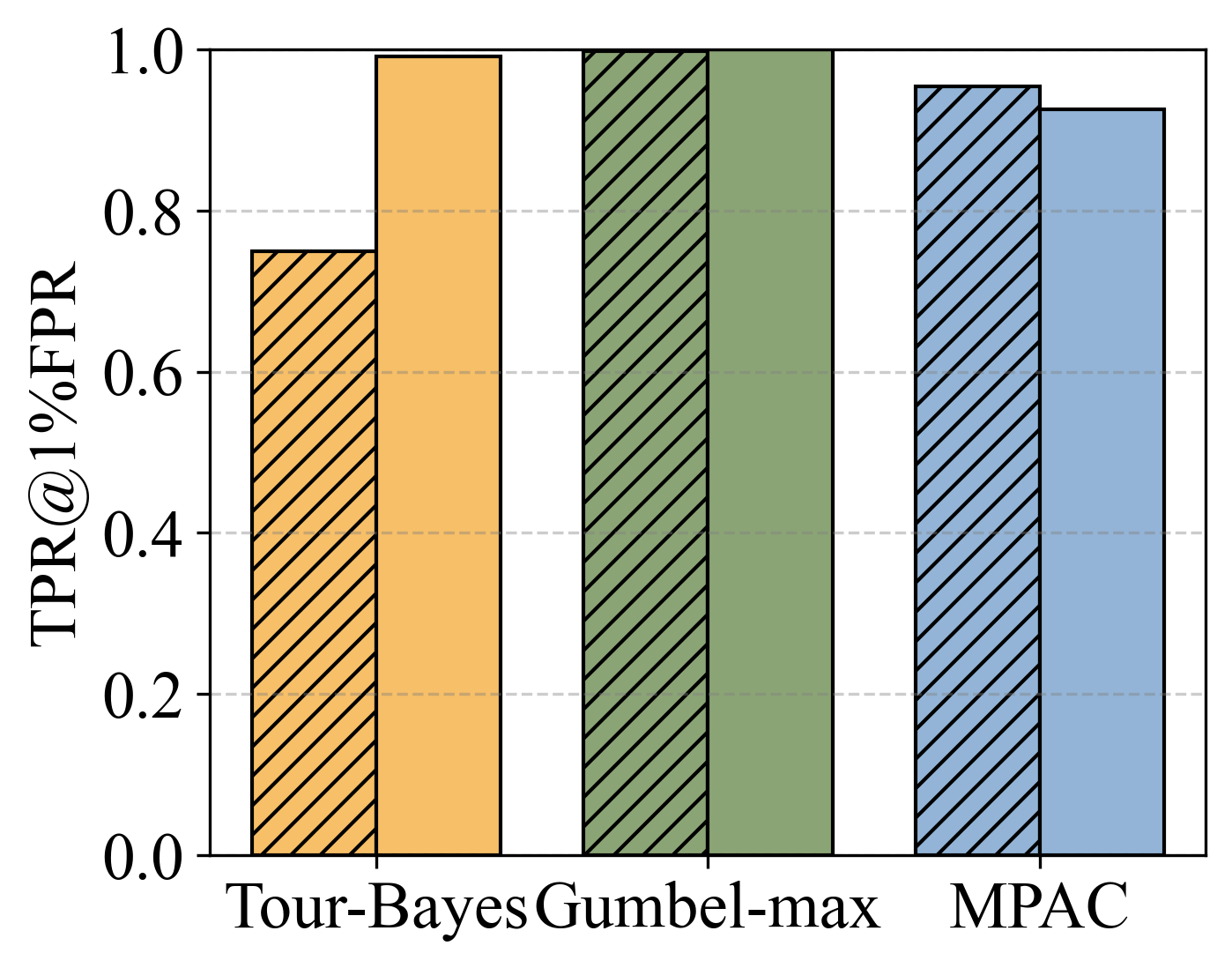}
    }
  \end{minipage}\hfill
  \begin{minipage}[t]{0.32\textwidth}
    \centering
    \subcaptionbox{$h$=4, 200 tokens}{
      \includegraphics[width=\linewidth]{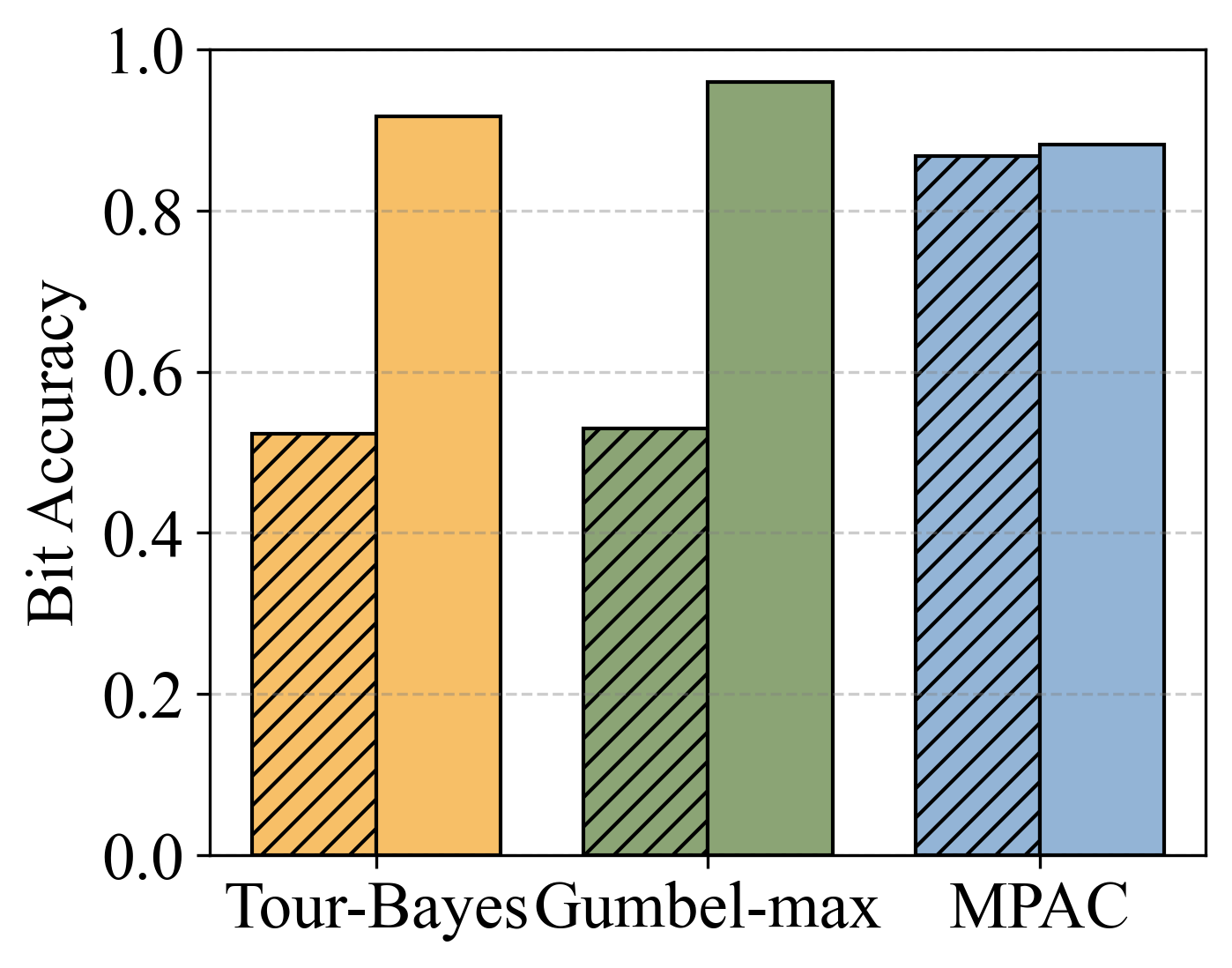}
    }
  \end{minipage}

  \begin{minipage}[t]{0.32\textwidth}
    \centering
    \subcaptionbox{$h$=4, 300 tokens}{
      \includegraphics[width=\linewidth]{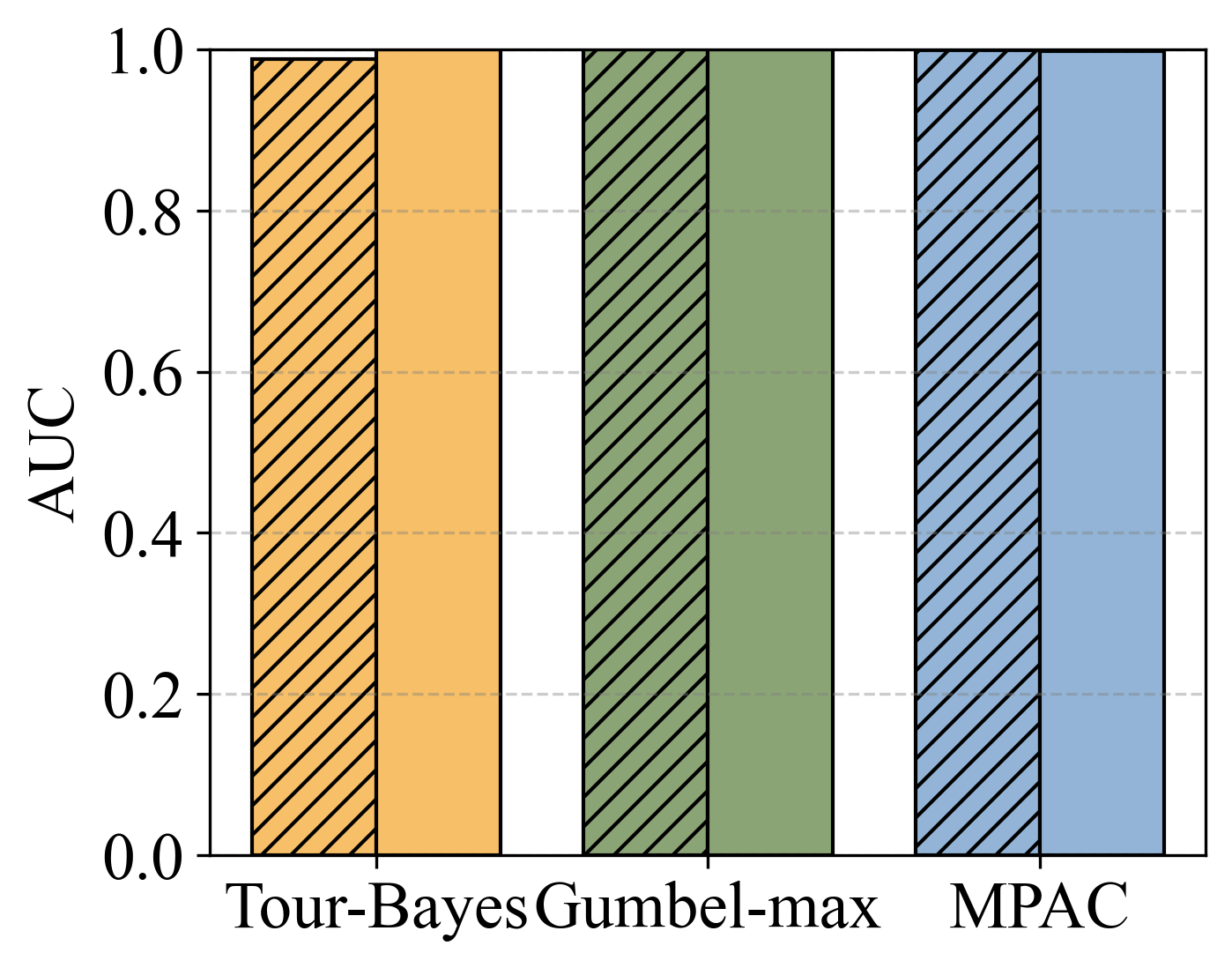}
    }
  \end{minipage}\hfill
  \begin{minipage}[t]{0.32\textwidth}
    \centering
    \subcaptionbox{ $h$=4, 300 tokens}{
      \includegraphics[width=\linewidth]{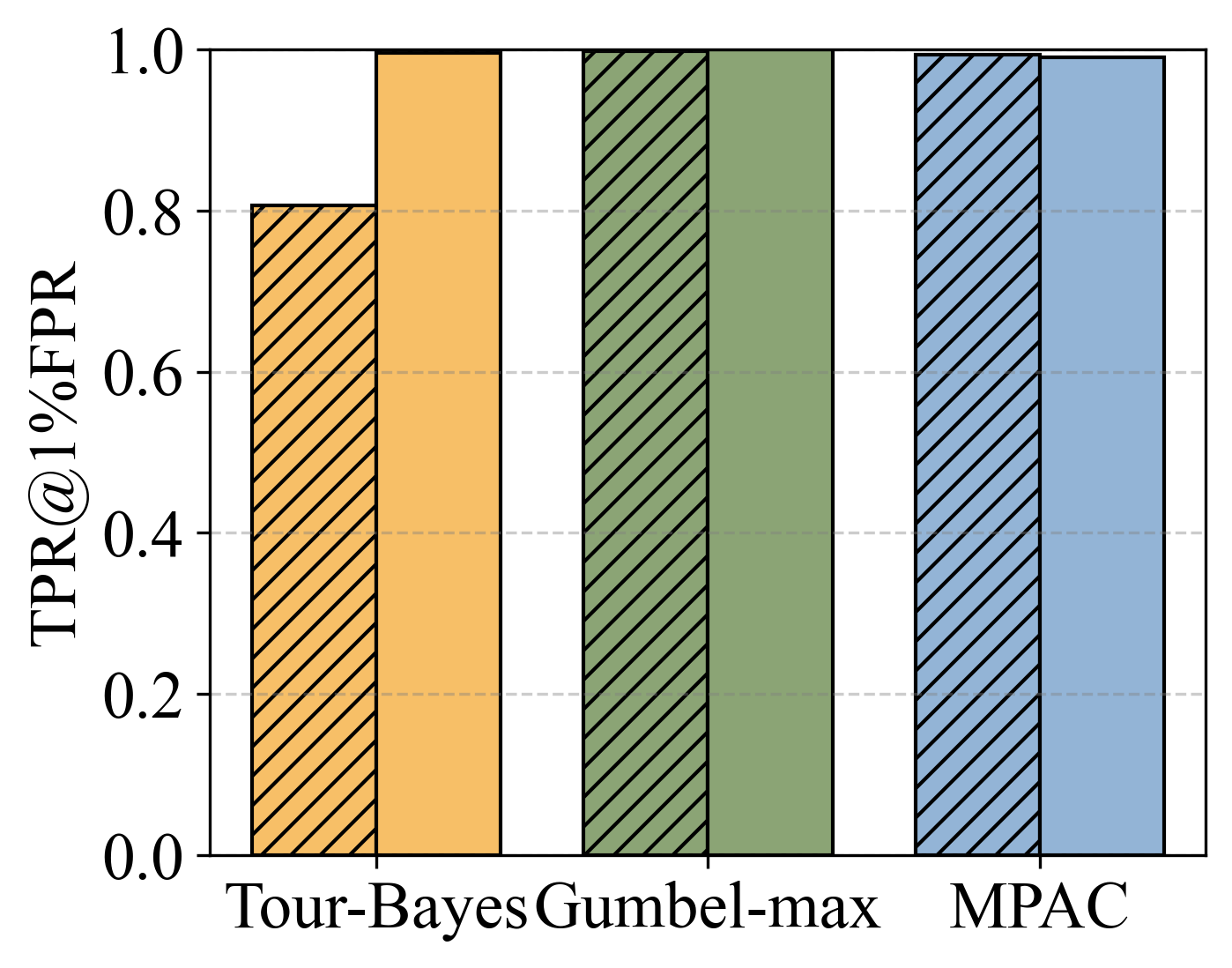}
    }
  \end{minipage}\hfill
  \begin{minipage}[t]{0.32\textwidth}
    \centering
    \subcaptionbox{$h$=4, 300 tokens}{
      \includegraphics[width=\linewidth]{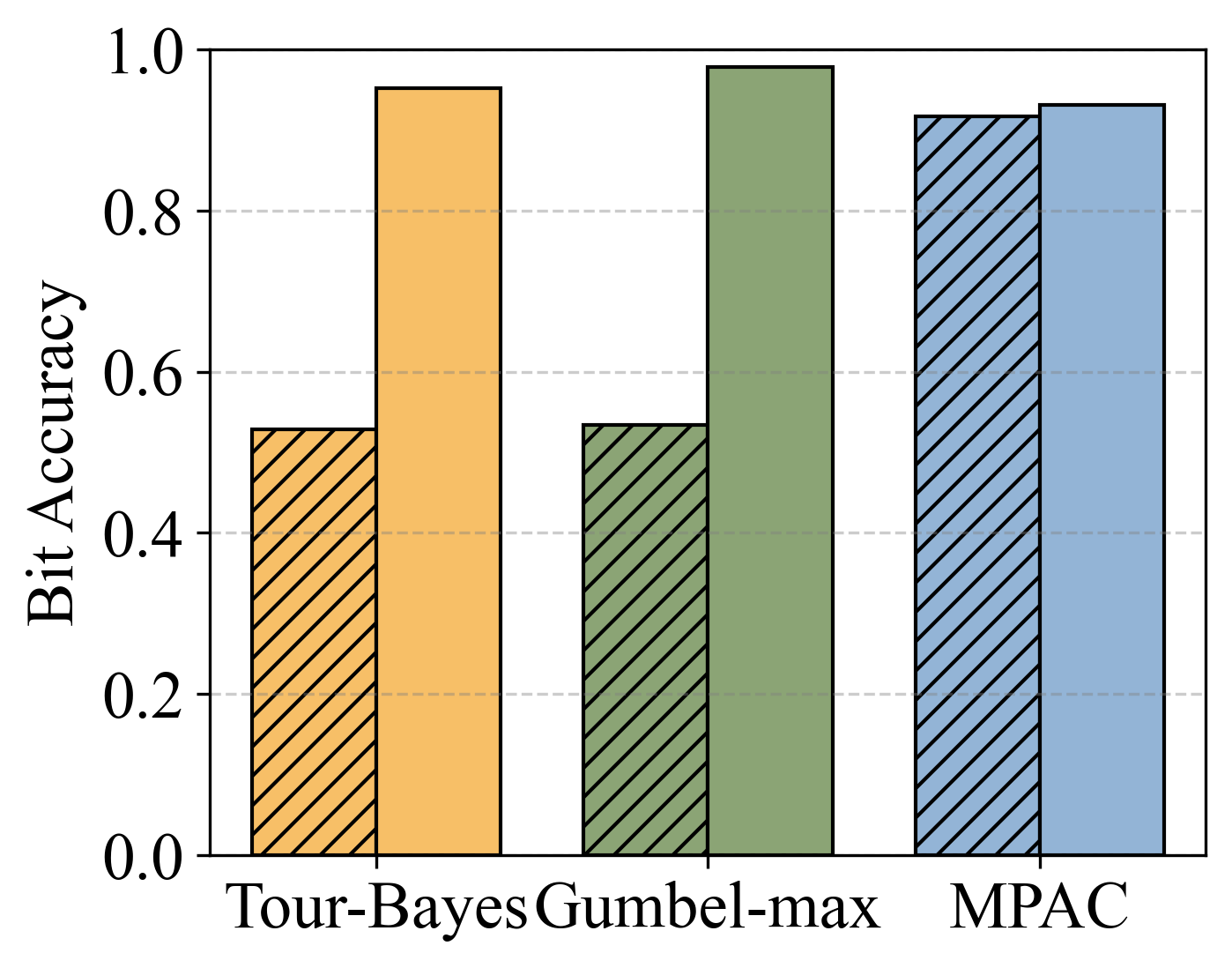}
    }
  \end{minipage}

\begin{minipage}[t]{0.32\textwidth}
    \centering
    \subcaptionbox{ $h$=4, 400 tokens}{
      \includegraphics[width=\linewidth]{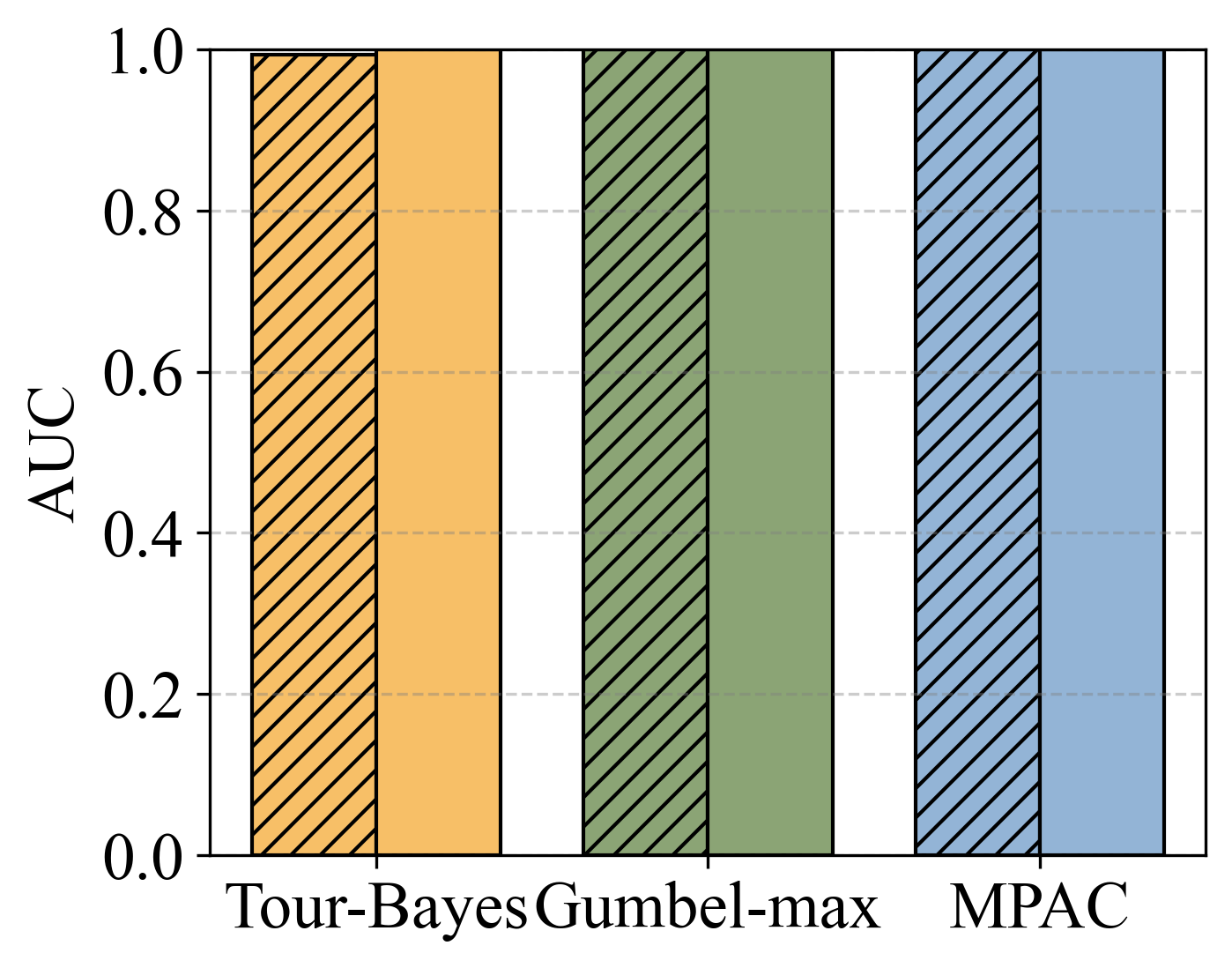}
    }
  \end{minipage}\hfill
  \begin{minipage}[t]{0.32\textwidth}
    \centering
    \subcaptionbox{$h$=4, 400 tokens}{
      \includegraphics[width=\linewidth]{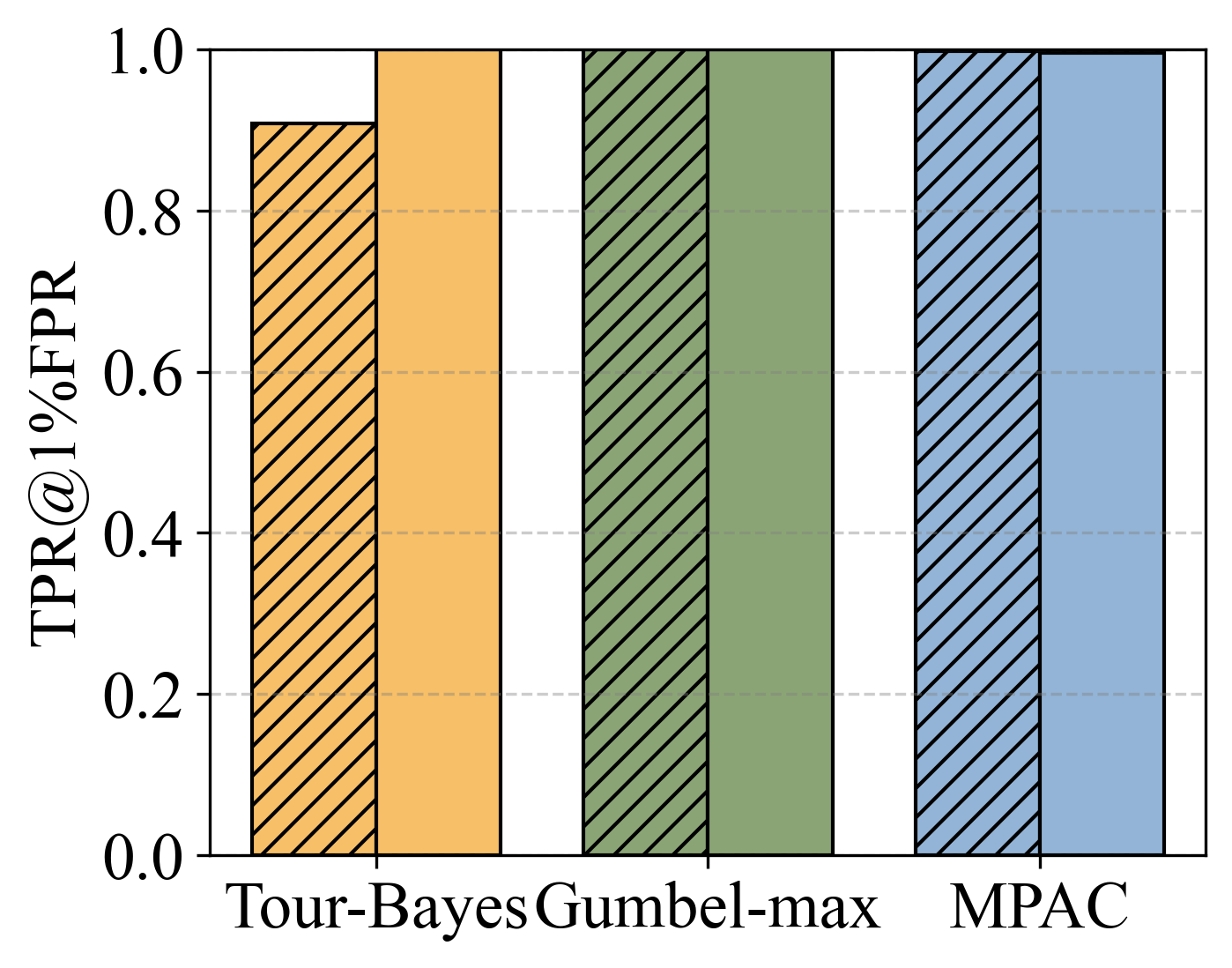}
    }
  \end{minipage}\hfill
  \begin{minipage}[t]{0.32\textwidth}
    \centering
    \subcaptionbox{ $h$=4, 400 tokens}{
      \includegraphics[width=\linewidth]{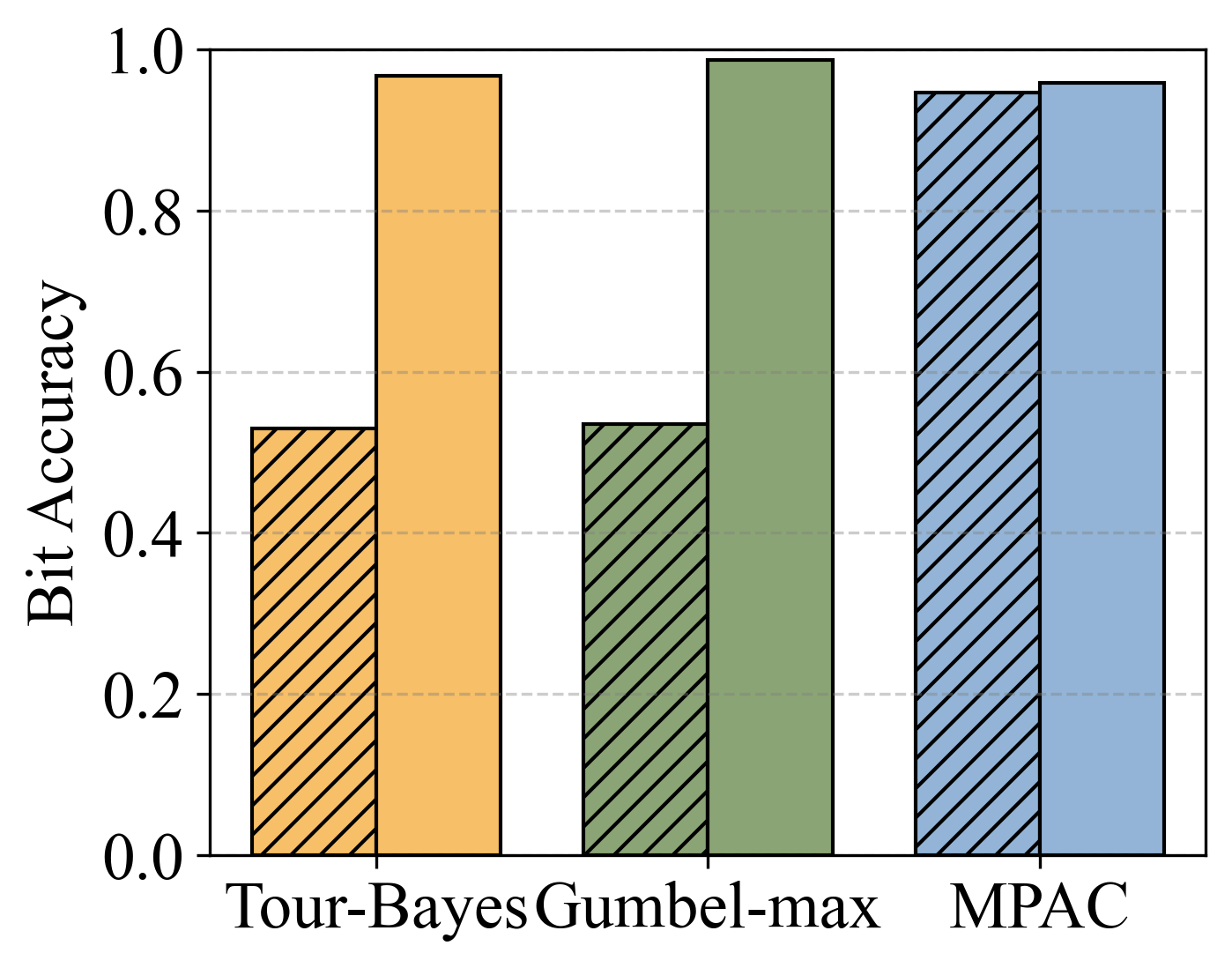}
    }
  \end{minipage}
  \caption{Detectability of MirrorMark across $h$=4. The setting is $m=2$ and $H=12$.}
  \label{fig:allocation5}
\end{figure*}


\end{document}